\documentclass[hidelinks,12pt,twoside,openany]{book}

\usepackage{amsmath,amssymb,amsthm,bm,bussproofs,caption,color,csquotes,enumitem,etoc,etoolbox,float,geometry,graphicx,hyperref,lipsum,listings,mathtools,mleftright,moreenum,multicol,physics,relsize,subcaption,tikz,thmtools,url,wrapfig,xcolor}
\usepackage[utf8]{inputenc}
\usepackage[english]{babel} 
\usepackage[pagestyles]{titlesec}
\usepackage[parfill]{parskip}
\usepackage[framemethod=TikZ]{mdframed}
\usepackage[style=alphabetic,backend=bibtex]{biblatex}
\usepackage[nodisplayskipstretch]{setspace}
\bibliography{ref}

\usetikzlibrary{calc,decorations.markings,shapes.callouts,arrows,intersections,cd,backgrounds}
\tikzcdset{every label/.append style = {font = \normalsize}}

\DeclareFontEncoding{LS1}{}{}
\DeclareFontSubstitution{LS1}{stix}{m}{n}
\DeclareSymbolFont{stix_symbols}{LS1}{stixscr}{m}{n}
\SetSymbolFont{stix_symbols}{bold}{LS1}{stixscr}{b}{n}
\DeclareSymbolFontAlphabet{\mathscr}{stix_symbols}

\newcounter{theorem}[chapter] 
\renewcommand{\thetheorem}{\arabic{chapter}.\arabic{theorem}}
\newenvironment{theorem}[2][]{%
    \refstepcounter{theorem}%
    \ifstrempty{#1}%
    {
        \mdfsetup{%
            frametitle={%
                \tikz[baseline=(current bounding box.east),outer sep=0pt]
                \node[anchor=east,rectangle,fill=gray!20]
                {\strut Theorem~\thetheorem};
            }
        }
    }%
    {
        \mdfsetup{%
            frametitle={%
                \tikz[baseline=(current bounding box.east),outer sep=0pt]
                \node[anchor=east,rectangle,fill=gray!20]
                {\strut Theorem~\thetheorem.~#1};
            }
        }%
    }%
    \mdfsetup{
        innertopmargin=0pt,%
        linecolor=gray!20,%
        linewidth=2pt,%
        topline=true,%
        skipabove=15pt,%
        nobreak=true,%
        frametitleaboveskip=\dimexpr-\ht\strutbox\relax
    }
    \begin{mdframed}[]\relax%
    \ifstrempty{#2}{}{\label{#2}}
}{\end{mdframed}}

\newenvironment{lemma}[2][]{
    \refstepcounter{theorem}%
    \ifstrempty{#1}%
    {
        \mdfsetup{%
            frametitle={%
                \tikz[baseline=(current bounding box.east),outer sep=0pt]
                \node[anchor=east,rectangle,fill=gray!20]
                {\strut Lemma~\thetheorem};
            }
        }
    }%
    {
        \mdfsetup{%
            frametitle={%
                \tikz[baseline=(current bounding box.east),outer sep=0pt]
                \node[anchor=east,rectangle,fill=gray!20]
                {\strut Lemma~\thetheorem.~#1};
            }
        }%
    }%
    \mdfsetup{
        innertopmargin=0pt,
        linecolor=gray!20,
        linewidth=2pt,
        topline=true,
        skipabove=15pt,
        nobreak=true,
        frametitleaboveskip=\dimexpr-\ht\strutbox\relax
    }
    \begin{mdframed}[]\relax%
    \ifstrempty{#2}{}{\label{#2}}
}{\end{mdframed}}

\newenvironment{proposition}[2][]{%
    \refstepcounter{theorem}%
    \ifstrempty{#1}%
    {
        \mdfsetup{%
            frametitle={%
                \tikz[baseline=(current bounding box.east),outer sep=0pt]
                \node[anchor=east,rectangle,fill=gray!20]
                {\strut Proposition~\thetheorem};
            }
        }
    }%
    {
        \mdfsetup{%
            frametitle={%
                \tikz[baseline=(current bounding box.east),outer sep=0pt]
                \node[anchor=east,rectangle,fill=gray!20]
                {\strut Proposition~\thetheorem.~#1};
            }
        }%
    }%
    \mdfsetup{
        innertopmargin=0pt,
        linecolor=gray!20,
        linewidth=1.5pt,
        topline=true,
        skipabove=15pt,
        nobreak=true,
        frametitleaboveskip=\dimexpr-\ht\strutbox\relax
    }
    \begin{mdframed}[]\relax%
    \ifstrempty{#2}{}{\label{#2}}
}{\end{mdframed}}

\newenvironment{corollary}[2][]{%
    \refstepcounter{theorem}%
    \ifstrempty{#1}%
    {
        \mdfsetup{%
            frametitle={%
                \tikz[baseline=(current bounding box.east),outer sep=0pt]
                \node[anchor=east,rectangle,fill=gray!20]
                {\strut Corollary~\thetheorem};
            }
        }
    }%
    {
        \mdfsetup{%
            frametitle={%
                \tikz[baseline=(current bounding box.east),outer sep=0pt]
                \node[anchor=east,rectangle,fill=gray!20]
                {\strut Corollary~\thetheorem.~#1};
            }
        }%
    }%
    \mdfsetup{
        innertopmargin=0pt,
        linecolor=gray!20,
        linewidth=1.5pt,
        topline=true,
        skipabove=15pt,
        nobreak=true,
        frametitleaboveskip=\dimexpr-\ht\strutbox\relax
    }
    \begin{mdframed}[]\relax%
    \ifstrempty{#2}{}{\label{#2}}
}{\end{mdframed}}

\theoremstyle{definition}
\newtheorem*{remark}{Remark}
\newtheorem*{convention}{Convention}
\newtheorem{definition}[theorem]{Definition}
\newtheorem{example}[theorem]{Example}

\newcommand{\subtitle}[1]{%
  \posttitle{%
    \par\end{center}
    \begin{center}\large#1\end{center}
    \vskip0.5em}%
}

\begingroup
  \makeatletter
  \@for\theoremstyle:=theorem,corollary,lemma,definition,remark,plain\do{%
    \expandafter\g@addto@macro\csname th@\theoremstyle\endcsname{%
      \addtolength\thm@preskip\parskip
    }}
  \makeatother
\endgroup

\makeatletter
\renewenvironment{proof}[1][\proofname]{\par
\vspace{-\topsep}
\pushQED{\qed}%
\normalfont
\topsep0pt \partopsep0pt 
\list{}{
    \settowidth{\leftmargin}{\itshape\proofname:\hskip\labelsep}%
    \setlength{\labelwidth}{0pt}%
    \setlength{\itemindent}{-\leftmargin}%
}
\item[\hskip\labelsep
        \itshape
    #1\@addpunct{.}]\ignorespaces
}{%
\popQED\endlist\@endpefalse
\addvspace{6pt plus 6pt} 
}
\makeatother

\makeatletter
\def\mathcolor#1#{\@mathcolor{#1}}
\def\@mathcolor#1#2#3{%
  \protect\leavevmode
  \begingroup
    \color#1{#2}#3%
  \endgroup
}
\makeatother

\makeatletter
\def\@makechapterhead#1{%
  {\parindent \z@ \raggedright \normalfont
    \ifnum \c@secnumdepth >\m@ne
        \huge\bfseries \@chapapp\space \thechapter
        \par\nobreak
        \vskip 10\p@
    \fi
    \interlinepenalty\@M
    \Huge \bfseries #1\par\nobreak
    \vskip 20\p@
  }}
\def\@makeschapterhead#1{%
  {\parindent \z@ \raggedright
    \normalfont
    \interlinepenalty\@M
    \Huge \bfseries  #1\par\nobreak
    \vskip 20\p@
  }}
\makeatother

\allowdisplaybreaks

\apptocmd{\thebibliography}{\csname phantomchapter\endcsname\addcontentsline{toc}{chapter}{Bibliography}}{}{}

\newcommand{\C}{\ensuremath{\mathbb{C}}}
\newcommand{\R}{\ensuremath{\mathbb{R}}}

\newcommand{\Z}{\ensuremath{\mathbb{Z}}}
\newcommand{\N}{\ensuremath{\mathbb{N}}}
\renewcommand{\d}{\ensuremath{\mathrm{d}}}
\newcommand{\id}{\ensuremath{\operatorname{id}}}
\renewcommand{\i}{\ensuremath{\mathrm{i}}}
\newcommand{\e}{\ensuremath{\operatorname{e}}}
\newcommand{\pa}{\ensuremath{\partial}}
\newcommand{\img}{\ensuremath{\operatorname{im}}}
\renewcommand{\ker}{\ensuremath{\operatorname{ker}}}
\newcommand{\lag}{\ensuremath{\mathcal{L}}}
\newcommand{\ord}[1]{\left\langle #1 \right\rangle}

\newcommand{\RP}{\mathbb{RP}}

\newcommand{\Hom}{\operatorname{Hom}}

\renewcommand{\geq}{\geqslant}
\renewcommand{\leq}{\leqslant}

\renewcommand{\bar}[1]{\overline{#1}}
\renewcommand{\tilde}[1]{\widetilde{#1}}
\renewcommand{\hat}[1]{\widehat{#1}}

\newcommand{\formpow}[1]{{\left[ \!{\left[ #1 \right]}\! \right]}}
\newcommand{\cotb}{\mathrm T^*}
\newcommand{\tanb}{\mathrm T}
\newcommand{\smooth}[1]{\mathrm C^\infty{\left( #1 \right)}}
\newcommand{\qweyl}{Q_{\mathrm{W}}}
\newcommand{\frg}{\mathfrak{g}}
\newcommand{\frh}{\mathfrak{h}}
\newcommand{\mm}{\mathfrak{m}}
\newcommand{\dpoly}[1][]{\mathscr{D}^{#1}_{\mathsf{poly}}(M)}
\newcommand{\dpolyrd}[1][]{\mathscr{D}^{#1}_{\mathsf{poly}}(\R^d)}
\newcommand{\tpoly}[1][]{\mathscr{T}^{#1}_{\mathsf{poly}}(M)}
\newcommand{\tpolyrd}[1][]{\mathscr{T}^{#1}_{\mathsf{poly}}(\R^d)}
\newcommand{\linf}{\mathrm{L}_\infty}
\newcommand{\dec}{\text{d\'ec}}

\begin{document}

\thispagestyle{empty}
    \newgeometry{left=2.5cm, right=2.5cm, top=3cm, bottom=3cm}
  \begin{center}
    \doublespacing
    \fontencoding{T1}\fontfamily{put}\selectfont
        { \Huge {\bfseries {Deformation Quantisation \\ via Kontsevich Formality Theorem}} \par}
    {\large \vspace*{25mm} 
    {{\Large \textbf{Peize Liu}} \par}{St.\ Peter's College \\ University of Oxford \par} 
    \vspace*{15mm}}
    {\large {\includegraphics[width=35mm]{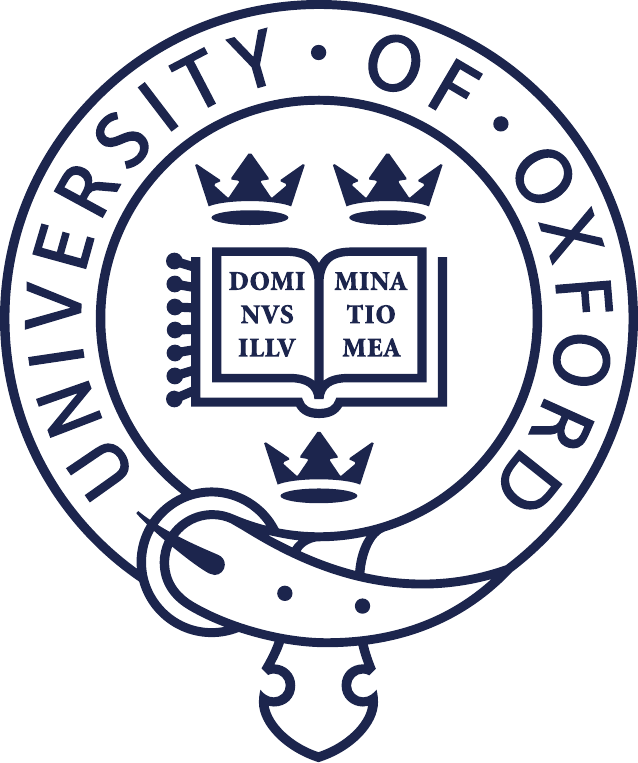} \par}
    \vspace*{10mm}
        {{A dissertation submitted in partial fulfilment \\ of the requirements for the degree of} \par}
        {\it {Master of Mathematical and Theoretical Physics} \par}
    \vspace*{2ex}
        {Trinity 2022}}
  \end{center}
  \restoregeometry

\newpage
\thispagestyle{plain}
\pagenumbering{roman}
\setcounter{page}{1} 
\begin{quote}
    \begin{center}
        \textbf{Abstract}
    \end{center}
    
    This dissertation is an exposition of Kontsevich's proof of the formality theorem and the classification of deformation quantisation on a Poisson manifold. We begin with an account of the physical background and introduce the Weyl--Moyal product as the first example. 
    Then we develop the deformation theory via differential graded Lie algebras and $\mathrm{L}_\infty$-algebras, which allows us to reformulate the classification of deformation quantisation as the existence of a $\mathrm{L}_\infty$-quasi-isomorphism between two differential graded Lie algebras, known as the formality theorem. Next we present Kontsevich's proof of the formality theorem in $\R^d$ and his construction of the star product. We conclude with a brief discussion of the globalisation of Kontsevich star product on Poisson manifolds.
\end{quote}

\vspace*{2em}
\begin{quote}
    \begin{center}
        \textbf{Acknowledgement}
    \end{center}

    I am grateful to my supervisor Prof.\ Christopher Beem for his guidance in the project. I would like to thank Haiqi Wu, Shuwei Wang, and Dekun Song for the constructive discussions. I would also like to thank Shuwei Wang for his kind assistance in \LaTeX\ related problems, without whom my dissertation would not be presented in such a clean manner. Finally, I would like to express my sincere appreciation to my college tutors, Prof.\ Bal\'{a}zs Szendr\H{o}i and Prof.\ Lionel Mason, who have provided invaluable support for my academic development along my undergraduate life.
\end{quote}

\vspace*{2em}
\begin{quote}
    \begin{center}
        \textbf{Declaration of Authorship}
    \end{center}

    I hereby declare that the dissertation I am submitting is entirely my own work except where otherwise indicated. It has not been submitted, either wholly or substantially, for another Honour School or degree of this University, or for a degree at any other institution. I have clearly signalled the presence of quoted or paraphrased material and referenced all sources. I have not copied from the work of any other candidate.
\end{quote}

\newpage
    \tableofcontents
\thispagestyle{plain}

\newpage
\pagenumbering{arabic}
\setcounter{page}{1}
\raggedbottom

\chapter{Introduction} 

Ever since the discovery of quantum theory in the early twentieth century, the problem of finding a correspondence between classical systems and quantum systems has been a central theme of theoretical physics. The defining features of quantum mechanics include the probabilistic interpretations of measurements and the non-commutativity of multiplications of physical observables. Proposed by Bayen \emph{et.\ al.\ }in \cite{bayen78I} and \cite{bayen78II}, deformation quantisation originates from an attempt of understanding quantisation ``\emph{as a deformation of the structure of the classical observable, rather than as a radical change in the natural of the observables.}'' In particular, the Hilbert space formalism of quantum mechanics is replaced by a deformation of classical mechanics on the phase space, which is known as a star product. In a mathematical formulation, deformation quantisation attempts to classify the star products on an associative algebra which recovers the commutative multiplication and the Poisson algebra in the formal classical limit $\hbar \to 0$. 

The Weyl--Moyal product proposed by Groenewold (\cite{groenewold46}) is well-known as the first example of deformation quantisation on the flat phase space. In 1983, De Wilde and Lecomte proved the existence of star products first on the cotangent bundle of a smooth manifold (\cite{lecomte83cotan}) and later on a symplectic manifold (\cite{lecomte83symp}). Independently, Fedosov (\cite{fedosov94}) proved the existence on a symplectic manifold in 1985 by constructing a flat connection of the Weyl bundle, which globalises the local Weyl--Moyal product in a coordinate-free manner. The classification up to equivalence on symplectic manifolds were studied by Deligne (\cite{deligne95}), Gutt (\cite{gut97equiv}, \cite{gutt99}), Weinstein and Xu (\cite{weinstein97}) \emph{et.\ al.} In contrast to the case of symplectic manifolds, where the Poisson bracket takes constant coefficients locally in the Darboux coordinates, deformation quantisation problems in the more general case of Poisson manifolds have greater complexity, since the local expression of Poisson bracket is arbitrary. 

In 1997, Kontsevich (\cite{kon97conjecture}) reformulated the problem as the existence of an $\linf$-quasi-isomorphism between two differential graded Lie algebras, which is known as the \emph{formality conjecture}. He proved the conjecture in the ground-breaking paper \cite{kon97}, in which he provided an explicit formula for the star product in $\R^d$. The work is among one of his four accomplishments in geometry (\cite{konwork}) for which he won the 1998 Fields Medal Prize. The proof of globalisation in an arbitrary smooth manifold was sketched in \cite{kon97} and was made precise by Cattaneo, Felder \& Tomassini in \cite{cattaneo02global} and also by Kontsevich in his subsequent work \cite{kontsevich01alg}. Kontsevich's formality theorem has inspired numerous development in symplectic geometry, algebraic geometry, and quantum field theory. 
\begin{itemize}
    \item The underlying idea of Kontsevich's construction comes from string theory. In the paper \cite{cattaneo00string} by Cattaneo \& Felder, Kontsevich's formula for the star product is interpreted as the perturbative expansion of the functional integral of Poisson sigma model, thus providing a Feynman path integral quantisation explanation to a construction arising in the context of canonical quantisation. 
    \item In \cite{tamarkin98}, Tamarkin provides another proof of the formality theorem in a pure algebraic setting using the theory of operads. He proved the formality of the differential graded Lie algebra associated to a finite dimensional vector space over a field of characteristic $0$. His approach is surveyed in \cite{kontsevich99op} and \cite{hinich03}.
    \item Deformation quantisation of smooth algebraic varieties over a field of characteristic $0$ is studied in \cite{kontsevich01alg} and \cite{yekutieli05}.
\end{itemize}

\subsubsection{Plan of the Dissertation}

In this dissertation, we will review Kontsevich's proof of formality theorem and the classification of deformation quantisation on a Poisson manifold. We will develop the deformation theory of differential graded Lie algebras and $\linf$-algebras required for the proof. The organisation of this dissertation is outlined as follows.

Chapter 1 is a general introduction to the problem of deformation quantisation and its physical background. In Section \ref{sec:quantisation}, we review the basics of classical and quantum mechanics and discuss the attempts of finding a categorical correspondence between the classical theory and the quantum theory. We sketch the proof of Groenewold no-go theorem which demonstrates the non-existence of a strict quantisation. In Section \ref{sec:star}, we introduce the theory of deformation quantisation and define the star products. Then we generalise the Weyl ordering and study the Weyl--Moyal product associated to the standard symplectic structure of $\R^{2n}$ as the first example of a star product in a flat space. 

Chapter 2 is a comprehensive review of deformation theory. In Section \ref{sec:defo_asso}, we generalise the star product to an algebraic setting of deformation of an associative algebra. Then we build up an associated differential graded Lie algebra (DGLA), whose general theory is studied in Section \ref{sec:DGLA}. We examine how the solutions of Maurer--Cartan equation of a DGLA control the deformation problem. In Section \ref{sec:poly_DGLA}, we specify again to the differential setting and construct two differential graded Lie algebras, $\tpoly$ and $\dpoly$, which correspond to the deformation of usual multiplication and that of Poisson bracket respectively. We prove that the Hochschild--Kostant--Rosenberg (HKR) map is a quasi-isomorphism of differential graded vector space between these DGLAs. The failure of the HKR map being a morphism of the Lie brackets suggests the necessity of a larger category then DGLA, for which we introduce the $\linf$-algebras in Section \ref{sec:L_infty}. Section \ref{sec:L_infty} builds up to the central result in deformation theory that homotopy equivalent DGLAs induce isomorphic deformation functors. Then we state the formality theorem, which claims that $\tpoly$ and $\dpoly$ are homotopy equivalent. This result allows us to classify deformation quantisations on any smooth manifold.

Chapter 3 explores the proof of the formality theorem due to Kontsevich (\cite{kon97}), Section \ref{sec:Kont_con} explains Kontsevich's construction of the $\linf$-quasi-isomorphism from $\tpolyrd$ to $\dpolyrd$. In particular an explicit star product is presented in $\R^d$ using admissible graphs and weight integrals over compactified configuration spaces. A sketch of the proof is given in Section \ref{sec:proof_formality}. With Stokes' theorem, the central problem reduces to the analysis of weight integral over two types of boundary strata of the configuration spaces. In Section 3.3, we outline the globalisation of Kontsevich's star product on a general Poisson manifold.

The major references of this dissertation are Kontsevich's pioneering paper \cite{kon97}, the expositions by Gutt (\cite{gutt01notes}), Keller (\cite{keller03}), and Cattaneo (\cite{cattaneo04notes}), and the excellent paper \cite{arnal02} by Arnal, Manchon \& Masmoudi which addresses the problem of signs in \cite{kon97}. For the general introduction to deformation quantisation, we also consult \cite{bordemann08} and \cite{weinstein95}. For the deformation theory of differential graded Lie algebras and $\linf$-algebras, we refer to \cite{doubek07}, \cite{manetti05}, \cite{jurvco19} and \cite{Fukaya01}.

In this dissertation we assume the results in algebra and geometry at the undergraduate level. Especially, we shall use differential geometry, homological algebra, and Lie algebras extensively, with the standard references \cite{chern99}, \cite{weibel}, and \cite{humphreys}. It is also helpful to have some knowledge in classical and quantum mechanics for understanding the material. The most relevant references are \cite{hall} and \cite{takhtajan}.

\section{Quantisation: From Classical to Quantum}\label{sec:quantisation}

In this section we review the physical backgrounds and motivations of quantisation. We also define all the geometric objects involved in the picture.

\subsection{Classical Mechanics}

In classical mechanics, a dynamical system is described by a smooth manifold $M$, whose dimension corresponds to the degree of freedom of that system. In the Hamiltonian formalism, the dynamics of the system is governed by the Hamilton's equations on the \textbf{phase space}, which is the cotangent bundle $\cotb M$ of $M$:
\begin{center}$
    \dfrac{\d q^i}{\d t} = \dfrac{\pa H}{\pa p_i}, \qquad \dfrac{\d p_i}{\d t} = -\dfrac{\pa H}{\pa q^i}, \qquad i = 1,...,n,
$\end{center}
where $(q^1,...,q^n,p_1,...,p_n)$ is a set of local coordinates on $\cotb M$ and $H\colon \cotb M \to \R$ is the Hamiltonian function. The cotangent bundle has a canonical symplectic form $\omega$, given in local coordinates by $\omega = \displaystyle\sum_{i=1}^{\dim M} \d q^i \wedge \d p_i$. Hamiltonian mechanics can be generalised on the setting of a symplectic manifold $(N,\omega)$ without difficulty.

\begin{definition}
    A \textbf{symplectic manifold} $(N,\omega)$ is a smooth manifold $N$ with a closed non-degenerate 2-form $\omega$, called the symplectic form. A diffeomorphism between two symplectic manifolds that preserve the symplectic form is called a symplectomorphism.
\end{definition}

An important result in symplectic geometry in the following:

\begin{lemma}[Darboux's Theorem]{}
    A $2n$-dimensional symplectic manifold is locally symplectomorphic to $(\R^{2n},\omega)$, where $\omega = \displaystyle\sum_{i=1}^n \d q^i \wedge \d p_i$ is the standard symplectic form on $\R^{2n}$.
\end{lemma}

The symplectic form $\omega$ induces the symplectic involution $J\colon \cotb N \to \tanb N$ via $\omega(X,J(\xi)) = \xi(X)$ for $\xi \in \cotb N$ and $X \in \tanb N$, which is a linear isomorphism. For a smooth function $f \in \smooth{N}$, the Hamiltonian vector field associated to $f$ is defined by $X_f := J(\d f)$. For two smooth functions $f$ and $g$, we can define the \textbf{Poisson bracket} $\left\{ -,- \right\}\colon \smooth{N}\times \smooth{N} \to \smooth{N}$:
\[\left\{ f,g \right\} := \omega(X_f,X_g) = X_f(g).  \]
The Lie bracket of the vector fields and the Poisson bracket are related by
\begin{center}$
    [X_f, X_g] = X_{\left\{ f,g \right\}}.
$\end{center}
Alternatively, the Poisson bracket defines a skew-symmetric bivector field $\pi \in \Gamma(N,\bigwedge^2 \tanb N)$ such that the bracket is given by the pairing $\left\{ f,g \right\} = \pi(\d f \wedge \d g)$. In local coordinates, the corresponding matrices satisfy $\pi^{ij} = (\omega^{-1})_{ij}$. Darboux's theorem demonstrates that it is always possible to find a chart such that the Poisson bivector field $\pi$ has constant coefficients in the local coordinates. Furthermore, the coordinates can be chosen such that 
\begin{equation}
    \left\{ q^i, q^j \right\} = 0, \qquad \left\{ p_i, p_j \right\} = 0, \qquad \left\{ q^i, p_j \right\} = \delta^i_j, \qquad i,j=1,...,n.
    \label{equ:I_canon_com_cm}
\end{equation}
In these coordinates the Poisson bracket of two functions $f,g$ is given by 
\[ \left\{ f,g \right\} = \displaystyle\sum_{i=1}^{\dim M}\left(\dfrac{\pa f}{\pa q^i}\dfrac{\pa g}{\pa p_i} - \dfrac{\pa g}{\pa q^i}\dfrac{\pa f}{\pa p_i}\right). \]
In terms of the Poisson bracket, Hamilton's equations show that the time evolution of any smooth function $f(t)\colon N \times \R \to N$ is given by
\begin{equation}
    \dfrac{\d f}{\d t} = \left\{ f,H \right\}. \label{equ:I_time_ev_cm}
\end{equation}

It is straightforward to prove that the Poisson bracket satisfies the conditions for a Lie bracket on $\smooth{N}$. Moreover it is a derivation in both of its arguments. We extract these properties to make the following definition:

\begin{definition}
    Let $A$ be an associative algebra over a field $k$. A Poisson bracket $\left\{ -,- \right\}$ on $A$ is a $k$-bilinear map, satisfying:
    \begin{itemize}[nosep] 
        \item Skew-symmetry: $\left\{ u,v \right\} = -\left\{ v,u \right\}$;
        \item Jacobi identity: $\left\{ u, \left\{ v,w \right\} \right\} + \left\{ v, \left\{ w,u \right\} \right\} = \left\{ w, \left\{ u,v \right\} \right\} = 0$;
        \item Leibniz rule: $\left\{ uv, w \right\} = u\left\{ v,w \right\} + \left\{ u,w \right\}v$.
    \end{itemize}
    where $u,v,w \in A$. An algebra $A$ equipped with a Poisson bracket is called a \textbf{Poisson algebra}. If $A = \smooth{M}$ is the algebra of smooth functions on a smooth manifold $M$, then $M$ is called a \textbf{Poisson manifold}.
\end{definition}
\begin{remark}
    A symplectic manifold must be even-dimensional due to the non-degeneracy of the symplectic form. On the other hand, any smooth manifold can be a Poisson manifold, whose Poisson bracket is allowed to be degenerate.
\end{remark}

\subsection{Quantum Mechanics}

\newcommand{\hilb}{\mathcal{H}}
For convenience we consider the flat phase space $N = \R^{2n}$ in this part. According to the Dirac--von Neumann axiomatisation of quantum mechanics, a quantum mechanical system is described by a separable complex Hilbert space $\hilb$ with a self-adjoint operator $\hat H$ acting on $\hilb$, which is called the Hamiltonian. Every physical observable $f$ is associated with a self-adjoint operator $\hat f$ on a Hilbert space $\hilb$. The spectrum of $\hat f$ is the set of possible outcomes of measuring $\hat f$. In the Heisenberg picture, the time evolution of the operator $\hat f$ is controlled by the Hamiltonian:
\begin{equation}
    \dfrac{\d \hat f}{\d t} = \dfrac{1}{\i\hbar}[\hat f, \hat H] \label{equ:I_time_ev_qm}
\end{equation}
where $\hbar$ is the \textbf{reduced Planck constant} and $[-,-]$ is the commutator. The coordinates $q^i$ and $p_i$ of the phase space $\R^{2n}$ are upgraded to self-operators $\hat q^i$ and $\hat p_i$ which satisfy the canonical commutation relation:
\begin{equation}
    \left[ \hat q^i , \hat q^j \right] = 0, \qquad \left[ \hat p_i, \hat p_j \right] = 0, \qquad \left[ \hat q^i, \hat p_j \right] = \i\hbar \delta^i_j\cdot \id_{\hilb}, \qquad i,j = 1,...,n. \label{equ:I_canon_com_qm}
\end{equation}
By Stone--von Neumann theorem, there is an irreducible unitary representation of $\hilb$ on $\mathrm L^2(\R^n)$ which is unique up to isomorphism such that 
\begin{equation}
    \hat q^i\colon \psi(x) \longmapsto x_i\psi(x), \qquad \hat p_i \colon \psi(x) \longmapsto -\i\hbar\dfrac{\pa\psi}{\pa x_i}. \label{equ:I_hat_qp}
\end{equation}
which is often referred as the position representation in physics. Comparing  \eqref{equ:I_canon_com_cm}, \eqref{equ:I_time_ev_cm} with \eqref{equ:I_time_ev_qm}, \eqref{equ:I_canon_com_qm}, we recognise that a classical system is quantised to a quantum system under the naive correspondence
\begin{center}
    \begin{tikzcd}[row sep = 0em]
        {\smooth{\R^{2n}}} && {\{\text{self-adjoint operators on }\mathrm L^2(\R^n)\}} \\
        \left\{ -,- \right\} && \dfrac{1}{\i\hbar}[-,-]
        \arrow["Q", shift left=1, from=1-1, to=1-3]
        \arrow["{\hbar \to 0}", shift left=1, from=1-3, to=1-1]
        \arrow[from=2-1, to=2-3]
    \end{tikzcd}
\end{center}
which is called \textbf{canonical quantisation}. In particular, one may postulate that there exists a $\R$-linear quantisation map $Q$ which ``canonically quantises'' the coordinates while mapping the Poisson brackets to commutators: $[Q(f),Q(g)] = \i\hbar Q{\left( \left\{ f,g \right\} \right)}$ for all $f,g \in \smooth{\R^{2n}}$. In 1946, Groenewold proved in his thesis \cite{groenewold46} that such quantisation map does not exist. We outline the idea of the proof following \cite{hall}.

\begin{theorem}[Groenewold's No-Go Theorem]{thm:I_Groenewold}
    Let $P_{\leq m}$ be the set of polynomials in $\R[q^1,...,q^n;p_1,...,p_n]$ with $\deg \leq m$. We can identify $P_{\leq m}$ as a subset of $\smooth{\R^{2n}}$. On the other hand, let $\mathscr D(\R^n)$ be the space of differential operators with coefficients in $\R[x_1,...,x_n]$. Consider $\hat q^i, \hat p_i \in \mathscr D(\R^n)$ as given in \eqref{equ:I_hat_qp}. There is no linear map $Q\colon P_{\leq 4} \to \mathscr D(\R^n)$ such that
    \begin{itemize}[nosep]
        \item $Q(1) = \id$;
        \item $Q(q^i) = \hat q^i$ and $Q(p_i) = \hat p_i$ for all $i = 1,...,n$;
        \item $[Q(f),Q(g)] = \i\hbar Q(\left\{ f,g \right\})$ for all $f,g \in P_{\leq 3}$.
    \end{itemize}
\end{theorem}
\begin{proof}
    Let $\qweyl$ be the Weyl quantisation map (\emph{cf.\ }Definition \ref{def:I_qweyl}). It can be proven that $[\qweyl(f),\qweyl(g)] = \i\hbar \qweyl(\left\{ f,g \right\})$ for all $f \in P_{\leq 2}$ and $g \in \R[q^1,...,q^n;p_1,...,p_n]$. Moreover, $\qweyl$ commutes with powers:
    \begin{center}$
        \qweyl((q^i)^m) = (\hat q^i)^m, \qquad \qweyl(p_i^m) = (\hat p_i)^m. 
    $\end{center}
    The next step is to show that if $Q$ satisfies the conditions in the Groenewold's Theorem, then $Q = \qweyl$ on $P_{\leq 3}$. Finally we consider $n=1$. Let $q := q^1$ and $p := p_1$. Suppose that the required map $Q$ exists. Note that we have the following identity of Poisson brackets:
    \begin{center}$
        \left\{ q^3, p^3 \right\} - \dfrac{1}{12}\left\{ \left\{ q^3, p^2 \right\}, \left\{ q^2, p^3 \right\} \right\} = 0.
    $\end{center}
    On the other hand, if we apply the map $Q$ on the left-hand side, then $Q = \qweyl$ on this expression. After some computation we find that
    \begin{center}$
        \qweyl(\left\{ q^3, p^3 \right\}) - \dfrac{1}{12}\qweyl(\left\{ \left\{ q^3, p^2 \right\}, \left\{ q^2, p^3 \right\} \right\}) = -3(\i\hbar)^2\id \ne 0.
    $\end{center} 
    This is a contradiction.
\end{proof}
The impossibility of finding a strict morphism taking Poisson brackets to commutators suggests that the conditions for the quantisation have to be relaxed, which leads to different quantisation schemes, the most famous among which are geometric quantisation and deforamtion quantisation. In this dissertation we only study the latter. In deformation quantisation, the condition $Q{\left( \left\{ f,g \right\} \right)} = \i\hbar \left[ Q(f), Q(g) \right]$ is replaced by an asympototic equality:
    \begin{equation}
        \left[ Q(f), Q(g) \right]= \i\hbar Q{\left( \left\{ f,g \right\} \right)}   + \mathcal{O}(\hbar^2). \label{equ:I_asym_quan}
    \end{equation}
Furthermore, the constant $\hbar$ is treated as a formal indeterminate, so that \eqref{equ:I_asym_quan} holds in the formal sense. We remove the Hilber space from the picture and works completely on the classical phase space. The ``quantum'' aspect of the theory is captured in a non-commutative associative product, which is the object we study in the next section.

\section{Star Products}\label{sec:star}

In this section we take the phase space $M$ to be an arbitrary Poisson manifold. Let $A:=\smooth{M}$ be the $\R$-algebra of smooth functions on $M$. Let $Q$ be the quantisation map. To study the quantisation on the phase space, we assume that $Q$ is an isomorphism and consider the pull-back of the product:
\begin{center}$
    f \star g := Q^{-1}(Q(f)Q(g)), \qquad f,g \in A.
$\end{center}
Writing this product in a formal power series: 
\begin{center}$
    f \star g = \displaystyle\sum_{m=0}^\infty \hbar^m B_m(f,g) \in A\formpow{\hbar},
$\end{center}
we infer from \eqref{equ:I_asym_quan} that $B_0(f,g) = B_0(g,f) = fg$ and $B_1(f,g) - B_1(g,f) = \i\left\{ f,g \right\}$. This is an example of a star product. We give the formal definition below.

\begin{definition}
    A \textbf{star product} on $A$ is an $\R\formpow{\hbar}$-bilinear map $\star\colon A\formpow{\hbar} \times A\formpow{\hbar} \to A\formpow{\hbar}$ such that 
    \begin{itemize}[nosep]
        \item Associativity: $(f \star g) \star h = f \star (g \star h)$ for $f,g \in A\formpow{\hbar}$;
        \item Unit: $1 \star f = f \star 1 = f$, for $f \in A\formpow{\hbar}$, where $1$ is the constant function on $M$;
        \item $f \star g = fg + \displaystyle\sum_{m=1}^\infty \hbar^m B_m(f,g)$ for $f,g \in A$, where $B_m$ are bi-differential operators of globally bounded order, that is, differential operators with respect to each argument in any local coordinates of $M$.
    \end{itemize}
    The star products $\star$ and $\star'$ are said to equivalent, if there exists a $\R$-linear isomorphism $\varphi\colon A\formpow{\hbar} \to A\formpow{\hbar}$ such that $\varphi(f) = f + \displaystyle\sum_{m=1}^\infty \hbar^m\varphi_m(f)$ for $f \in A$ and $\varphi(f) \star' \varphi(g) = \varphi(f \star g)$ for $f,g \in A\formpow{\hbar}$.
\end{definition}

\begin{remark}
    If $\mu$ denotes the multiplication of functions on $A$, then we can write $\star = \mu + \displaystyle\sum_{m=1}^\infty \hbar^m B_m$. The notion of star products will be discuss in the broader concept of deformations of associative multiplication in Section \ref{sec:defo_asso}.
\end{remark}

A star product gives rise to a Poisson bracket on $M$:
\begin{lemma}[]{lem:I_star2Poi}
    Suppose that $\star = \mu +  \displaystyle\sum_{m=1}^\infty \hbar^m B_m$ is a star product on $A$. Then $\left\{ f,g \right\} := B_1(f,g) - B_1(g,f)$ defines a Poisson bracket on $A$. Moreover, the bracket only depends on the equivalence class of $\star$.
\end{lemma}
\begin{proof}
    The skew-symmetry and bilinearity of $\left\{ -,- \right\}$ are immediate from definition. For the Jacobi identity, we can expand the bracket $\hbar^2\left\{ f,\left\{ g,h \right\} \right\}$ in $\R[\hbar]/{\ord{\hbar^3}}$:
    \begin{center}$
        \hbar^2\left\{ f,\left\{ g,h \right\} \right\} = f \star (g \star h) - f \star (h \star g) - (g \star h) \star f + (h \star g) \star f.
    $\end{center}
    The expansion of two other brackets are similar. By associativity of the star product, we note that the sum vanishes:
    \begin{center}$
        \left\{ f,\left\{ g,h \right\} \right\} + \left\{ g,\left\{ h,f \right\} \right\} + \left\{ h,\left\{ f,g \right\} \right\} = 0.
    $\end{center} 
    The Leibniz rule for $\left\{ -,- \right\}$ follows from the Leibniz rule of $B_1$. Therefore $\left\{ -,- \right\}$ is a Poisson bracket. Suppose that $\star$ is equivalent to $\star'$ via the isomorphism $\varphi$. Then at order $\hbar$ we find that 
    \begin{center}$
        g\varphi_1(f) + f\varphi_1(g) + B_1'(f,g) = \varphi_1(fg) + B_1(f,g).
    $\end{center}
    Note that $B_1' - B_1$ is symmetric in the two arguments and hence does not contribute to the Poisson bracket. The star product $\star'$ induces the same Poisson bracket as $\star$. 
\end{proof}

\begin{definition}
    Let $A$ be a Poisson algebra. We say that $(A\formpow{\hbar},\star)$ is a \textbf{deformation quantisation} of $A$, if $\star = \mu + \displaystyle\sum_{m=1}^\infty \hbar^m B_m$ is a star product on $A\formpow{\hbar}$ such that $\left\{ f,g \right\} = B_1(f,g) - B_1(g,f)$ coincides with the Poisson bracket on $A$. 
\end{definition}
\begin{remark}
    Note that the definition has a factor $\i$ different from our picture in quantum mechanics. Since $\hbar$ is considered as a formal symbol, it is possible to absorb $\i$ into the $\hbar$ and to work entirely over $\R$. 
\end{remark}
We are more interested in the converse of Lemma \ref{lem:I_star2Poi}, that is, to find a deformation quantisation for a given Poisson bracket. This seems more difficult and requires some additional structure on the algebra $A$ as we want to constructing $B_2, B_3,...$ from $B_1$. The classification theorem of deformation quantisation, stated in its greatest generality, is the following bijective correspondence on a smooth manifold $M$:
\begin{equation}
    \dfrac{\{\text{Formal Poisson structures}\}}{\text{Equivalence}} \quad\longleftrightarrow\quad \dfrac{\{\text{Star products}\}}{\text{Equivalence}}. \label{equ:I_defo_quan_bij}
\end{equation}
By the end of Chapter 2, we shall elucidate this matter as a consequence of Kontsevich's formality theorem \ref{thm:II_formality}.

\subsection{Weyl--Moyal Product}

In this part we study a simple construction of deformation quantisation on $\R^{2n}$. The idea of Weyl quantisation and Moyal bracket originates from the works by Weyl (\cite{weyl}), Moyal (\cite{moyal}), Wigner (\cite{wigner}), and Groenewold (\cite{groenewold46}) which had been studied before deformation quantisation appeared. The story begins with the ordering ambiguity in canonical quantisation of monomials like $qp^2 \in \smooth{\R^2}$. The naive strategy of ``replacing $q$ by $\hat q$ and $p$ by $\hat p$'' does not work because $\hat q \hat p^2$ fails to be self-adjoint on $\mathrm L^2(\R)$. A more serious approach is to take the total symmetrisation of the expression:
\begin{center}$
    Q(qp^2) = \dfrac{1}{3}\left( \hat q \hat p^2 + \hat p \hat q \hat p + \hat p^2 \hat q \right) = \dfrac{1}{2}\left(  \hat q \hat p^2 + \hat p^2 \hat q \right).
$\end{center}
This leads to the definition of Weyl quantisation:
\begin{definition}\label{def:I_qweyl}
    For a monomial $f = q^Ip^J\in \R[q^1,...,q^n;p_1,...,p_n]$, we define the \textbf{Weyl quantisation} $\qweyl(f)$ of $f$ to be the total symmetrisation of the monomial $\hat q^I \hat p^J$, where $I,J$ are multi-indices. Then we extend $\qweyl$ linearly to $\R[q^1,...,q^n;p_1,...,p_n]$. Furthermore, Weyl quantisation can be extended to $\mathrm L^2(\R^{2n})$ functions as follows:
    \begin{center}$
        \qweyl(f) := \dfrac{1}{(2\pi)^n}\displaystyle\int_{\R^{n}}\displaystyle\int_{\R^{n}}\mathcal{F}f(\bm a,\bm b)\e^{\i(\bm a \cdot \hat{\bm q} + \bm b \cdot \hat{\bm p})}\, \d^n\bm a\, \d^n\bm b,
    $\end{center}
    where $\mathcal{F}f(\bm a,\bm b)$ is the Fourier transform of $f(\bm q, \bm p)$, and the exponential of operators is formally understood by the \emph{Baker--Campbell--Hausdorff formula} (\emph{cf.\ }\eqref{equ:II_BCH}). If we identify a suitable codomain of $\qweyl$, it becomes invertible with inverse $\qweyl^{-1}$ known as the \emph{Wigner transform}.
\end{definition}

\begin{definition}
    The \textbf{Weyl--Moyal product} of $f,g \in \smooth{\R^{2n}}$ is defined as $f \star g := \qweyl^{-1}(\qweyl(f)\qweyl(g))$. 
\end{definition}
    In this definition, $\hbar$ is treated as a real parameter. It can be proven (\emph{cf.\ }\cite{hall}) that $f \star g \to fg$ uniformly as $\hbar \to 0$, so that $f \star g$ is an \emph{actual} deformation of $fg$, not just a formal one. However, it is sufficient to consider the formal aspect of the the Weyl--Moyal product. The closed form of the Weyl--Moyal product is given by
\begin{equation}
    \begin{aligned}[b]
        (f \star g)(x) &= \left.\exp( \dfrac{\i\hbar}{2}\displaystyle\sum_{i,j=1}^{2n}\pi^{ij}\dfrac{\pa}{\pa x^i}\dfrac{\pa}{\pa y^j})f(x)g(y)\right|_{x=y}. \\
        &= \left. (fg)(x) + \displaystyle\sum_{m=1}^\infty \dfrac{(\i\hbar)^m}{2^m m!} \displaystyle\prod_{k=1}^{m}\left( \displaystyle\sum_{i_k,j_k=1}^{2n}\pi^{i_k,j_k}\dfrac{\pa}{\pa x^{i_k}}\dfrac{\pa}{\pa y^{j_k}} \right)f(x)g(y)\right|_{x=y}.
    \end{aligned} \label{equ:I_Moyal}
\end{equation}
where $\pi$ is the Poisson bivector field associated to the standard symplectic structure of $\R^{2n}$. Now we consider $\hbar$ as a formal symbol, so that $f \star g \in \smooth{\R^{2n}}\formpow{\hbar}$. Next we replace $\R^{2n}$ by $\R^d$ and allow the constant matrix $(\pi^{ij})_{i,j}$ to have arbitrary form. 

\begin{proposition}[]{}
    By extending the Weyl--Moyal product $\R\formpow{\hbar}$-bilinearly, we obtain a star product on $\smooth{\R^{d}}\formpow{\hbar}$, which is called the \textbf{Moyal star product}.
\end{proposition}
\begin{proof}
    The only non-trivial part is checking associativity. For $f,g,h \in \smooth{\R^{d}}$, 
    \begin{align*}\hspace*{-2em}
            ((f \star g) \star h)(x) 
            &= \left.\exp( \dfrac{\i\hbar}{2}\pi^{ij}\dfrac{\pa}{\pa x_i}\dfrac{\pa}{\pa z_j})(f\star g)(x)g(z)\right|_{x=z} \\
            &= \left.\exp( \dfrac{\i\hbar}{2}\pi^{ij}\left(\dfrac{\pa}{\pa x^i} + \dfrac{\pa}{\pa y^i}\right)\dfrac{\pa}{\pa z^j})\exp( \dfrac{\i\hbar}{2}\pi^{k\ell}\dfrac{\pa}{\pa x^k}\dfrac{\pa}{\pa y^\ell})f(x)g(y)g(z)\right|_{x=y=z} \\
            &= \left.\exp( \dfrac{\i\hbar}{2}\pi^{ij}\left(\dfrac{\pa}{\pa x^i}\dfrac{\pa}{\pa y^j} + \dfrac{\pa}{\pa y^i}\dfrac{\pa}{\pa z^j} + \dfrac{\pa}{\pa x^i}\dfrac{\pa}{\pa z^j}\right))f(x)g(y)g(z)\right|_{x=y=z} \\
            &= (f\star (g \star h))(x). \qedhere
    \end{align*}
\end{proof}

\begin{remark}
    A natural question that arises is the possibility to extend this construction to include Poisson bivector field $\pi$ with coefficients depending on the coordinates. This proves to be the most difficult part of Kontsevich's quantisation theorem and is a consequence of the formality theorem. There is a canonical star product associated to $\pi$, whose formula is given in \eqref{equ:III_kont_star}.
\end{remark}

\chapter{Deformation Theory}

\section{Deformations of Associative Algebras} \label{sec:defo_asso}

In this section, we consider the deformation problems in a more general setting. From now on, we denote by $k$ a fixed ground field of characteristic $0$. Let $A$ be an associative $k$-algebra with the multiplication map $\mu\colon A \otimes_k A \to A$. We would like to deform $\mu$ into another product $\star\colon A_R \otimes_k A_R \to A_R$ using a test algebra $R$.

\subsection{Deformation Functor}

\newcommand{\rtest}{\mathscr{R}}
\begin{definition}\label{def:II_test_alg}
    A \textbf{test algebra} $R$ is a commutative local Artinian $k$-algebra with the maximal ideal $\mm_R$ and the residue field $R/\mm_R \cong k$. The category of test algebras is denoted by $\rtest$, whose morphisms are local $k$-algebra homomorphisms, that is, the $k$-algebra homomorphisms $f\colon R \to S$ satisfying $f(\mm_R) \subseteq \mm_S$.
\end{definition}

\begin{remark}
    Note that $R$ is commutative local Artinian implies that $\mm_R$ is nilpotent. We will use this fact in the proof of Lemma \ref{lem:II_dgla_nil}.
\end{remark}

\begin{definition}\label{def:II_R_defo}
    Let $R$ be a test algebra. Let $A_R := A \otimes_k R$. An \textbf{$\bm{R}$-deformation} of the multiplication $\mu$ on $A$ is an associative $R$-linear map $\star\colon A_R \otimes_R A_R \to A_R$ such that $\mu = \star$ on $A_R/\mm_R \otimes_R A_R/\mm_R$. This can be represented in the commutative diagram:
    \begin{center}
        \begin{tikzcd}
            {A_R \otimes_R A_R} && {A \otimes_k A} \\
            {A_R} && A
            \arrow["{\bmod \mm_R}", from=1-1, to=1-3]
            \arrow["\star"', from=1-1, to=2-1]
            \arrow["{\bmod \mm_R}", from=2-1, to=2-3]
            \arrow["\mu", from=1-3, to=2-3]
        \end{tikzcd}
    \end{center}
    Note that an $R$-deformation is uniquely determined by its restriction on $A \otimes_k A$. 
    The $R$-deformations $\star$ and $\star'$ are said to be equivalent, if there exists an $R$-module automorphism $\varphi\colon A_R  \to A_R$ such that $\varphi(u \star v) = \varphi(u) \star' \varphi(v)$ and $\varphi = \id$ on $A_R/\mm_R$.
\end{definition}

\newcommand{\defo}{\mathsf{Defo}}
\begin{definition}\label{def:II_defo_func}
    For a test algebra $R$, let $\defo_A(R)$ be the set of equivalence classes of $R$-deformations of $A$. Then we obtain a functor $\defo_A\colon \mathscr R \to \mathsf{Set}$, called the \textbf{deformation functor}.
\end{definition}

\begin{example}
    An \emph{infinitesimal deformation} of $A$ is a $R$-deformation with $R = k[\hbar]/{\ord{\hbar^2}}$.
\end{example}

\begin{example}\label{eg_II_formal_defo}
    A \emph{formal deformation} is a $R$-deformation with $R = k\formpow{\hbar}$. Although $k\formpow{\hbar}$ is not Artinian, it can be expressed as a completion in the $\hbar$-adic topology: $
        k\formpow{\hbar} = \varprojlim_{m} k[\hbar]/{\ord{\hbar^m}}$.
    In fact, we have an natural isomorphism of functors:
    \begin{center}$
        \defo_A(k\formpow{\hbar}) \simeq \varprojlim_m \defo_A(k[\hbar]/{\ord{\hbar^m}}).
    $\end{center}
    So we can study the formal deformations by studying the deformation functor.
\end{example}

\begin{example}\label{eg_II_star}
    Let $A := \smooth{M}$. A \emph{star product} $\star$ on $A$ is an associative $\R\formpow{\hbar}$-deformation such that $u \star v = uv + \displaystyle\sum_{n=1}^\infty \hbar^n B_n(u,v)$ for $u,v \in A$, where $B_n$ are bi-differential operators in the local coordinates of $M$. 
\end{example}

\subsection{Hochschild Cohomology}

The information of deformations are encoded in the Hochschild complex defined below. This part consists of standard material which can be found in any homological algebra textbook. We include the details here for completeness.

\newcommand{\dhoch}{\delta_{\mathsf{H}}}
\newcommand{\dhochm}{\d_{\mathsf{H}}}
\begin{definition}\label{def:II_Hoch_cplx}
    For each $p \in \Z_{>0}$, let $C^p(A,A):= \Hom_k(A^{\otimes p},A)$. We define the \textbf{Hochschild differential} $\dhoch\colon C^p(A,A) \to C^{p+1}(A,A)$ as 
    \begin{center}$
        (\dhoch f)(u_0,...,u_{p}) = u_0 f(u_1,...,u_p) + \displaystyle\sum_{r=1}^p (-1)^r f(u_0,...,u_{r-1}u_r,...,u_p) + (-1)^{p+1}f(u_0,...,u_{p-1})u_p.
    $\end{center}
    The chain can be augmented to $p = 0$ by defining $\dhoch \colon C^0(A,A) \cong A \to C^1(A,A)$ as 
    \begin{center}$
        (\dhoch a)(u_0) = u_0a - au_0.
    $\end{center}
\end{definition}

\begin{lemma}[Hochschild Complex]{}
    The Hochschild differential $\dhoch$ satisfies $\dhoch^2 = 0$. Therefore $(C^p(A,A))_{p \geq 0}$ is a cochain complex, called the Hochschild complex.
\end{lemma}
\begin{proof}
    Let $A^e := A \otimes_k A^{\mathrm{op}}$. We define the \emph{bar complex} $(\mathsf{B}_nA,b_n)_{n \geq 0}$ as follows. Let $\mathsf{B}_nA := A \otimes_k A^{\otimes n} \otimes_k A$, with the left $A^e$-module structure
    \begin{center}$
        (u \otimes u')(u_1 \otimes \cdots u_n) := uu_1 \otimes \cdots \otimes u_nu'.
    $\end{center}
    We define $b_n\colon \mathsf{B}_nA \to \mathsf{B}_{n-1}A$ as the $A^e$-module homomorphism given by
    \begin{center}$
        b_n(u_0 \otimes  \cdots \otimes u_{n+1}) := \displaystyle\sum_{k=0}^n (-1)^k \cdots \otimes u_{k}u_{k+1} \otimes \cdots.
    $\end{center}
    We claim that $b_{n-1} \circ b_{n} = 0$. Indeed, $b_{n-1} \circ b_{n}(u_0 \otimes  \cdots \otimes u_{n+1})$ is a linear combination of the elements of form $\cdots \otimes u_iu_{i+1} \otimes \cdots \otimes u_ju_{j+1}\otimes \cdots$ and $\cdots \otimes u_{i-1}u_iu_{i+1} \otimes \cdots$.
    For each of these elements, there are exactly two ways of contractions from $u_0 \otimes  \cdots \otimes u_{n+1}$, with opposite signs, which cancels each other. Hence $(\mathsf{B}_nA,b_n)_{n \geq 0}$ is indeed a chain complex. 

    Next we have an $A^e$-module isomorphism $\varphi\colon C^n(A,A) \xrightarrow{\sim} \Hom_{A^e}(\mathsf{B}_nA,A)$ given by $f \mapsto (u_0 \otimes  \cdots \otimes u_{n+1} \mapsto u_0f(u_1,...,u_n)u_{n+1})$. Moreover, the following diagram commutes:
    \begin{center}
        \begin{tikzcd}
            {C^{n+1}(A,A)} & {\Hom_{A^e}(\mathsf{B}_{n+1}A,A)} & {\mathsf{B}_{n+1}A} \\
            {C^n(A,A)} & {\Hom_{A^e}(\mathsf{B}_nA,A)} & {\mathsf{B}_{n+1}A}
            \arrow["\dhoch", from=2-1, to=1-1]
            \arrow["\varphi", from=1-1, to=1-2]
            \arrow["\varphi", from=2-1, to=2-2]
            \arrow["{\text{dual}}", from=1-2, to=1-3]
            \arrow["{\text{dual}}", from=2-2, to=2-3]
            \arrow["b", from=1-3, to=2-3]
            \arrow["{b^\vee}", from=2-2, to=1-2]
        \end{tikzcd}
    \end{center}
    This implies that $\dhoch^2 = 0$.
\end{proof}

\newcommand{\hoch}{\operatorname{HH}}
The cohomology of the Hochschild complex $(C^\bullet(A,A),\dhoch)$ is called \textbf{Hochschild cohomology}, and is denoted by $\hoch^\bullet(A,A)$. We shall see in Section \ref{sec:HKR} that, when $A$ is the algebra of smooth functions on a smooth manifold $M$, the Hochschild cohomology computes the polyvector fields on $M$.

\begin{remark}
    The $n$-th Hochschild cohomology $\hoch^n(A,A)$ can also be identified with the Ext group $\operatorname{Ext}^n_{A^e}(A,A)$ in the category of $A^e$-modules.
\end{remark}

    

\subsection{Gerstenhaber Bracket}
We define a graded Lie bracket on the shifted\footnote{For a graded vector space $V := \bigoplus V^i$, the shifted space $V[k]$ has the grading $V[k]^i := V^{i+k}$.} Hochschild complex $C^\bullet(A,A)[1]$. 

\newcommand{\Ger}[1]{\left[ #1 \right]_{\mathsf{G}}}
\begin{definition}
    For $f \in C^{p+1}(A,A)$ and $g \in C^{q+1}(A,A)$, we define the \textbf{Gerstenhaber product} $f \bullet g \in C^{p+q+1}(A,A)$ as
    \begin{align*}
        (f \bullet g)(u_0,...,u_{p+q}) 
        &:= \displaystyle\sum_{i=0}^p (-1)^{iq} (f \bullet_i g)(u_0,...,u_{p+q}) \\
        &:= \displaystyle\sum_{i=0}^p (-1)^{iq}f(u_0,...,u_{i-1},g(u_{i},...,u_{i+q}),u_{i+q+1},...,u_{p+q}).
    \end{align*}
    Then we define the \textbf{Gerstenhaber bracket} of $f$ and $g$ as 
    \begin{center}$
        \Ger{f,g} := f \bullet g - (-1)^{pq} g \bullet f
    $\end{center}
    and extend it $k$-bilinearly to $C^\bullet(A,A)$. 
\end{definition}


\begin{lemma}[Graded Jacobi identity for Gerstenhaber bracket]{lem:II_Ger_Jacobi}
    For $f \in C^{p+1}(A,A)$, $g \in C^{q+1}(A,A)$, and $h \in C^\bullet(A,A)$, the Gerstenhaber bracket satisfies \vspace*{-1em}
    \begin{center}$
        \Ger{f,\Ger{g,h}} = \Ger{\Ger{f,g},h} + (-1)^{pq}\Ger{g,\Ger{f,h}}.
    $\end{center}
\end{lemma}
\begin{proof}
    By linearity, we may assume that $h \in C^{r+1}(A,A)$. Consider the \emph{associator} of the Gerstenhaber product:
    \begin{center}$
        A(f,g,h) := (f\bullet g)\bullet h - f \bullet (g \bullet h).
    $\end{center}
    A direct computation shows that
    \begin{align*}
        \Ger{\Ger{f,g},h} +{} &(-1)^{|f|\cdot|g|}\Ger{g,\Ger{f,h}} - \Ger{f,\Ger{g,h}} \\
        = -&(-1)^{pq}A(g,f,h) + A(f,g,h) - (-1)^{qr}A(f,h,g) \\
        {}+{} &(-1)^{pq+pr}A(g,h,f) + (-1)^{pr+qr}A(h,f,g) - (-1)^{pq+pr+qr}A(h,g,f).
    \end{align*}
    The right-hand side of the above equation would vanish if the associator satisfies the graded symmetry:
    \begin{center}$
        A(f,g,h) = (-1)^{qr}A(f,h,g).
    $\end{center}
    By definition we have:
    \begin{equation}
        \begin{aligned}
            A(f,g,h) &- (-1)^{qr}A(f,h,g)
            = \displaystyle\sum_{i=0}^p\displaystyle\sum_{j=0}^{p+q}(-1)^{i(q+r)}(f \bullet_i g)\bullet_j h 
            - \displaystyle\sum_{i=0}^p\displaystyle\sum_{j=0}^{q}(-1)^{iq}f \bullet_i (g \bullet_j h) \\
            & \quad- (-1)^{qr} \left( \displaystyle\sum_{i=0}^p\displaystyle\sum_{j=0}^{p+r}(-1)^{i(q+r)}(f \bullet_i h)\bullet_j g
            - \displaystyle\sum_{i=0}^p\displaystyle\sum_{j=0}^{r}(-1)^{ir}f \bullet_i (h \bullet_j g)
            \right)
        \end{aligned} \label{equ:associator}
    \end{equation} 
    The composition of partial products satisfies:
    \begin{center}$
        (f \bullet_i g)\bullet_j h = \begin{cases}
            (f \bullet_j h) \bullet_i g, & i > j, \\
            f \bullet_i (g \bullet_{j-i} h), & i \leq j \leq i+q, \\
            (f \bullet_{j-q} h) \bullet_i g, & i + q < j. 
        \end{cases}
    $\end{center}
    Substituting this into \eqref{equ:associator} we see that the right-hand side is cancalled. Hence $A(f,g,h) = (-1)^{qr}A(f,h,g)$.
\end{proof}

\begin{lemma}[Hochschild Differential via Gerstenhaber Bracket]{lem:II_hoch_ger}
    For $f \in C^{p}(A,A)$, we have
    \begin{center}$
        \dhoch f = (-1)^{p+1}\Ger{\mu,f},
    $\end{center}
    where $\dhoch$ is the Hochschild differential and $\mu \in C^2(A,A)$ is the multiplication on $A$.
\end{lemma}
\begin{proof}
    This is straightforward by definition. 
    \begin{align*}
        &\Ger{\mu,f}(u_0,...,u_{p}) = (\mu \bullet f)(u_0,...,u_{p}) - (-1)^{p-1}(f \bullet \mu)(u_0,...,u_{p}) \\
        &\quad = f(u_0,...,u_{p-1})u_{p} + (-1)^{p-1}u_0f(u_1,...,u_{p}) + \displaystyle\sum_{r=1}^{p-1} (-1)^{r+p}f(u_0,...,u_{r-1}u_r,...,u_{p}) \\
        &\quad = (-1)^{p+1}(\dhoch f)(u_0,...,u_{p+1}). \qedhere
    \end{align*}
\end{proof}
\begin{remark}
    For reasons which will be clear shortly, we define a \textbf{modified differential} $\d$ on the Hochschild complex by $\dhochm f = (-1)^{p+1}\dhoch f$ for $f \in C^{p}(A,A)$. Then the new differential satisfies $\dhochm f = \Ger{\mu,f}$ for any $f \in C^\bullet(A,A)$. We will be using $\dhochm$ instead of $\dhoch$ from now on.
\end{remark}

\begin{corollary}[]{}
    The multiplication $\mu$ on $A$ has vanishing Gerstenhaber bracket: $\Ger{\mu,\mu} = 0$.
\end{corollary}
\begin{proof}
    This follows from associativity of $\mu$:
    \begin{equation*}
        {\Ger{\mu,\mu}}(u,v,w) = 2u(vw) - 2(uv)w = 0. \qedhere
    \end{equation*}
\end{proof}
\begin{remark}
    The same proof shows that the associativity of any $\eta \in C^2(A,A)$ is equivalent to $\Ger{\eta,\eta} = 0$.
\end{remark}

\begin{lemma}[Graded Leibniz Rule]{}
    For $f \in C^{p+1}(A,A)$ and $g \in C^{q+1}(A,A)$, we have $\dhochm \Ger{f,g} = \Ger{\dhochm f, g} + (-1)^p\Ger{f, \dhochm g}$.
\end{lemma}
\begin{proof}
    By Lemma \ref{lem:II_hoch_ger}, this is the the graded Jacobi identity, which is Lemma \ref{lem:II_Ger_Jacobi}.
\end{proof}

\begin{remark}
    The compatible structure of a differential and a graded Lie bracket on the Hochschild complex gives rise to a structure of differential graded Lie algebra, which is introduced in the next section.
\end{remark}

\section{Differential Graded Lie Algebras}\label{sec:DGLA}

\begin{definition}
    A \textbf{graded Lie algebra} $\frg = \bigoplus_{i \in \Z} \frg^i$ is a $\Z$-graded vector space over a field $k$ with a bilinear map $[-,-]\colon \frg \times \frg \to \frg$, satisfying that
    \begin{itemize}
        \item $[x,y] \in \frg^{|x|+|y|}$;
        \item Graded skew-symmetry: $[x,y] = -(-1)^{|x|\cdot|y|}[y,x]$;
        \item Graded Jacobi identity: $(-1)^{|x|\cdot|y|}[[x,y],z] + (-1)^{|y|\cdot|z|}[[y,z],x] + (-1)^{|z|\cdot|x|}[[z,x],y] = 0$.
    \end{itemize}
    where $x,y,z \in \frg$ are homogeneous elements with grading $|x|,|y|,|z|$ respectively.

    A \textbf{differential graded Lie algebra} (DGLA) is a graded Lie algebra $\frg$ with a differential $\d\colon \frg^i \to \frg^{i+1}$ which is $k$-linear of degree $1$ and satisfies that 
    \begin{itemize}
        \item Graded Leibniz rule: $\d[x,y] = [\d x,y] + (-1)^{|x|}[x,\d y]$;
        \item Sqaure-zero: $\d^2 = 0$.
    \end{itemize}
\end{definition}
\begin{remark}
    Let $\frg$ be a graded Lie algebra. Then the zeroth grading $\frg^0$ is a Lie algebra in the usual sense.
\end{remark}
\begin{example}
    Following the results established in the previous section, the shifted Hochschild complex $C^\bullet(A,A)[1]$ with the Gerstenhaber bracket $\Ger{-,-}$ and the modified Hochschild differential $\d $ is a DGLA.
\end{example}

\begin{example}
    Let $\frg$ be a DGLA over $k$, and $S$ be a graded commutative $k$-algebra, that is, $ab = (-1)^{|a|\cdot|b|}ba$ for homogeneous $a,b \in S$. We can define a DGLA structure on $\frg \otimes_k S$ as follows:
    \begin{itemize}
        \item $\frg \otimes_k S$ is graded by $(\frg \otimes_k S)^i = \displaystyle\bigoplus_{p+q=i}(\frg^p \otimes_k S^q)$;
        \item The differential is given by $\d(x \otimes a) = \d x \otimes a$ for $x \in \frg$ and $a \in S$;
        \item The graded Lie bracket is given by $[x \otimes a, y \otimes b] = (-1)^{|a|\cdot|y|}[x,y] \otimes ab$ for homogeneous $x,y \in \frg$ and $a,b \in S$.
    \end{itemize}
\end{example}

A DGLA $(\frg,\d)$ is naturally a cochain complex. The corresponding cohomology $\mathrm H^\bullet(\frg) := \bigoplus_{i \in \Z}\mathrm H^i(\frg)$ is a graded vector space with an induced graded Lie bracket. $\mathrm H^\bullet(\frg)$ with the zero differential is again a DGLA.

\begin{definition}
    A \textbf{morphism of DGLAs} $f\colon (\frg,[,]_{\frg},\d_{\frg}) \to (\frh,[,]_{\frh},\d_{\frh})$ is a $k$-linear map $f\colon \frg \to \frh$ homogeneous of degree $0$ such that $f\circ \d_{\frg} = \d_{\frh}\circ f$ and $f([x,y]_{\frg}) = [f(x),f(y)]_{\frh}$.
    \begin{center}
        \begin{tikzcd}
            {\frg^i} & {\frg^{i+1}} \\
            {\frh^i} & {\frh^{i+1}}
            \arrow["{\d_{\frg}}", from=1-1, to=1-2]
            \arrow["{\d_{\frh}}", from=2-1, to=2-2]
            \arrow["f"', from=1-1, to=2-1]
            \arrow["f", from=1-2, to=2-2]
        \end{tikzcd}
    \end{center}
    If $f$ induces the isomorphism of the cohomologies $\mathrm H^\bullet(f)\colon \mathrm H^\bullet(\frg) \to \mathrm H^\bullet(\frh)$, then $f$ is called a \textbf{quasi-isomorphism of DGLAs}. 
\end{definition}

\subsection{Maurer--Cartan Equation}\label{sec:DGLA_MC}

In general deformation theory, the philosophy due to Deligne (\cite{deligne87letter}) is that deformation problems in characteristic $0$ are controlled by some differential graded Lie algebra. For deformation of algebras, the DGLA is the shifted Hochschild complex. We show that the deformation functor we defined in the previous section can be identified with a quotient of the set of zeroes of the Maurer--Cartan equation.

\newcommand{\MC}{\mathsf{MC}}
\newcommand{\MCbar}{\overline{\mathsf{MC}}}
\newcommand{\ad}{\operatorname{ad}}
\begin{definition}\label{def:II_MC_equ}
    Let $\frg$ be a DGLA. For $x \in \frg^1$, the \textbf{Maurer--Cartan equation} is given by
    \begin{equation}
        \d x + \dfrac{1}{2}[x,x] = 0. \label{equ:II_MC}
    \end{equation}
    The set of solutions is denoted by $\MC(\frg)$. It is clear from definition that the solutions of the Maurer--Cartan equation is preserved under a morphism of DGLAs. 
\end{definition}

$\MC(\frg)$ is also preserved under the gauge actions of $\frg^0$, which will be defined and studied in the following. For this part we follow the approach outlined in \cite{manetti05}.

\newcommand{\adr}{\operatorname{\bar{ad}}}
We would like to adjoint an element $\delta$ of degree 1 to $\frg$ such that $\d x = [\delta,x]$. To do this, we construct a new DGLA $\frg_\delta$ given by $\frg^1_\delta = \frg^1 \oplus k\delta$ and $\frg^i_\delta = \frg^i$ for $i \ne 1$. The graded Lie bracket and the differential are defined for homogeneous elements by
\begin{center}$
    [x+v\delta,y+w\delta]_\delta = [x,y] + v\d y + (-1)^{|x|}w\d x, \qquad \d_\delta(x+v\delta) = \d x = [x,v\delta]_\delta.
$\end{center}
The Maurer--Cartan equation can be expressed as
\begin{equation}
    \d x + \dfrac{1}{2}[x,x] = 0 \quad \iff \quad 
    [x+\delta, x+ \delta]_\delta = 0. \label{equ:II_MC_delta}
\end{equation}
We define the right adjoint action of $\frg^0$ on $\frg^1_\delta$, given by $\adr(x_0)(y+v\delta) = [y+v\delta,x_0]_\delta$. 
Suppose that $\frg^0$ is ad-nilpotent. The following lemma shows that the Lie algebra action can be exponentiated to a group action: 

\begin{lemma}[Exponentiation of Nilpotent Lie Algebras]{}
    There exists a functor $\exp$, which attaches a group $G$ to each nilpotent Lie algebra $\frh$, such that there is a surjection $\exp\colon \frh \to G$ satisfying $\exp(x)\exp(y) = \exp(x \star y)$, where $x \star y\in \frh$ is given by the Baker--Campbell--Hausdorff formula:
    \begin{equation}
        x \star y := x+y+\frac{1}{2}[x, y]+\frac{1}{12}[x,[x, y]]-\frac{1}{12}[y,[x, y]]+\cdots. \label{equ:II_BCH}
    \end{equation}
    The infinite sum truncates at some finite term because of the nilpotency of $\frh$.        
\end{lemma}
\begin{proof} 
    A discussion of this fact can be found in e.g.\ \cite{getzler09}.
\end{proof}
\begin{remark}
    Moreover, for any representation $\rho\colon \frh \to \mathfrak{gl}(V)$ such that $\rho(\frh)$ is a nilpotent subalgebra of $\mathfrak{gl}(V)$, the corresponding exponential representation $\exp(\rho)\colon \exp(\frh) \to \operatorname{GL}(V)$ is given by $(\exp \rho)\e^a = \e^{\rho(a)}$, where $\e^a$ denotes the usual exponential of an endomorphism.
\end{remark}
Therefore $\adr$ induces the action of $\exp(\frg^0)$ on $\frg^1_\delta$. We embed $\frg^1$ into $\frg^1_\delta$ as a hyperplane $\left\{ x+\delta\colon x \in \frg^1 \right\}$, which is preserved under this group action. 
Explicitly, we obtain a right action of $\exp(\frg^0)$ on $\frg^1$, given by
\begin{center}$
    \exp(x_0) \longmapsto \left( y \longmapsto \e^{\adr(x_0)}(y+\delta) - \delta = y + \d x_0 + [y,x_0] + \dfrac{1}{2}\left( [\d x_0,x_0] + [[y,x_0],x_0] \right) + \cdots \right)
$\end{center}
The group $\exp(\frg^0)$ is called the \textbf{gauge group}. On the other hand, $\adr$ induces the right action of $\frg^0$ on $\frg^1$ given by 
\begin{center}$
    x_0 \longmapsto \left( y \longmapsto \d x_0 + [y,x_0]\right).
$\end{center} 
The key observation is that:
\begin{lemma}[]{}
    The gauge action of $\exp(\frg^0)$ on $\frg^1$ preserves the set $\MC(\frg)$ of the solutions of the Maurer--Cartan equation.
\end{lemma}
\begin{proof}
    We have $[x+\delta,x+\delta]_\delta = 0$ if and only if $x \in \MC(\frg)$ by \eqref{equ:II_MC_delta}. For convenience we denote by $x \cdot \exp(x_0)$ the right action of $\exp(x_0)$ on $x \in \frg^1$. Let $x \in \MC(\frg)$. We have
    \begin{align*}
        &\d(x\cdot \exp(x_0)) + \dfrac{1}{2}[x\cdot \exp(x_0),x\cdot \exp(x_0)] \\
        &\quad= [\delta, \e^{\adr(x_0)}(x+\delta)]_\delta + \dfrac{1}{2}\left([\e^{\adr(x_0)}(x+\delta),\e^{\adr(x_0)}(x+\delta)]_\delta - 2[\delta, \e^{\adr(x_0)}(x+\delta)]_\delta\right) \\
        &\quad= \e^{\adr(x_0)}[x+\delta,x+\delta]_\delta = 0.
    \end{align*}
    This shows that $x \cdot \exp(x_0) \in \MC(\frg)$.
\end{proof}




\begin{lemma}[]{lem:II_dgla_nil}
    Suppose that $R$ is a test $k$-algebra (\emph{cf.\ }Definition \ref*{def:II_test_alg}) and $\frg$ is any DGLA. Then the DGLA $\frg \otimes_k \mm_R$ is ad-nilpotent in its zeroth grading.
\end{lemma}
\begin{proof}
    The grading of $\mm_R$ is centred at zero, so $(\frg \otimes_k \mm_R)^0 = \frg^0 \otimes_k \mm_R$. The graded Lie bracket is given by $[x \otimes a, y \otimes b] = [x,y] \otimes ab$, where $x,y \in \frg$ are homogeneous and $a,b \in \mm_R$. Since $R$ is local Artinian, $\mm_R$ is nilpotent. For sufficiently large $n$, we have
    \begin{center}$
        \ad(x \otimes a)^n (y \otimes b) = (\ad x)^n(y) \otimes a^nb = 0.
    $\end{center}
    By taking linear combinations we deduce that $(\frg \otimes_k \mm_R)^0$ is ad-nilpotent.
\end{proof}

Following \ref{lem:II_dgla_nil}, for a test algebra $R$, we have a well-defined action of $\exp(\frg^0 \otimes_k \mm_R)$ on $\MC(\frg \otimes_k \mm_R)$. The set of orbits is denoted by
\begin{equation}
    \MCbar_{\frg}(R) := \MC(\frg \otimes_k \mm_R)/\exp(\frg^0 \otimes_k \mm_R). \label{equ:II_MC_functor}
\end{equation}
This is referred as the Maurer--Cartan moduli set. We can also check that the guage actions are preserved under a morphism of DGLAs. In this way we obtain a functor $\MCbar_{\frg}\colon \mathscr R \to \mathsf{Set}$, called the \textbf{Maurer--Cartan functor} associated to $\frg$. This is related to the deformation functor in Definition \ref{def:II_defo_func} by the following proposition.

\begin{proposition}[Deformation via Maurer--Cartan Functor]{prop:II_defo_MC}
    Let $A$ be an associative $k$-algebra, and $C^\bullet(A,A)[1]$ be the corresponding shifted Hochschild complex. Then we have an isomorphism of functors:
    \begin{center}
        $\defo_A \simeq \MCbar_{C^\bullet(A,A)[1]}.$
    \end{center}
\end{proposition}
\begin{proof}
    Let $R$ be a test algebra. Let $\mu$ be the multiplication on $A$, We consider $\star \in \Hom_k(A \otimes_k A, A \otimes_k R)$, and let $B(u,v) := u \star v - uv$. So $B = \star-\mu$ is of degree $1$ in the DGLA $\frg := C^\bullet(A,A)[1] \otimes_k \mm_R$. We can compute the Maurer--Cartan equation for $B$:
    \begin{align*}
        \dhochm B + \dfrac{1}{2}\Ger{B,B}
        &= \d {\star} - \d \mu  + \dfrac{1}{2}\Ger{\mu,\mu} + \dfrac{1}{2}\Ger{\star,\star} - \Ger{\mu,\star} 
        = \dfrac{1}{2}\Ger{\star,\star}.
    \end{align*}
    The vanishing of $\Ger{\star,\star}$ is exactly the associativity condition for $\star$. We deduce that $\star$ is a $R$-deformation if and only if $B \in \MC(\frg)$.

    Next, we need to check that the equivalence of $R$-deformations in Definition \ref{def:II_R_defo} 
    is generated by the gauge group $\exp(\frg^0)$. Before that we need to know how this action works. Since we have $\mu \in \frg^1$ such that $\dhochm x = \Ger{\mu,x}$ for $x \in \frg$, the adjunction of $\delta$ is unnecessary. $\exp(\frg^0)$ acts on $\frg^1$ by 
    \begin{center}$
        \exp(x_0) \longmapsto \left( y \longmapsto \e^{\Ger{-,x_0}}(y + \mu) - \mu \right).
    $\end{center}
    Suppose that $\star$ and $\star'$ are two $R$-deformations. Note that $B = \star - \mu$ and $B' = \star' - \mu$ are related by a gauge action $\exp(f)$ for some $f \in \frg^0$ if and only if $\star' = \e^{\Ger{-,f}}\star$. This is equivalent to saying that $f$ is an $R$-equivalence of $\star$ and $\star'$ in the sense of Definition \ref{def:II_R_defo}. 
    
    In summary, we obtain a bijective correspondence:
    \begin{center}$
        \defo_A(R) \simeq \MCbar_{C^\bullet(A,A)[1]}(R).
    $\end{center}
    It is straightforward to check that the correspondence is functorial in $R$.
\end{proof}

\begin{remark}
    It is more common in the literature (\emph{cf.\ }\cite{kon97}, \cite{manetti05}) to 
    call $\MCbar_\frg$ the deformation functor. Nevertheless we have shown that this is equivalent to our Definition \ref{def:II_defo_func}.
\end{remark}
\begin{remark}
    Following Example \ref{eg_II_formal_defo}, it is straightforward to extend the previous discussion to the formal deformation $R = k\formpow{\hbar}$ as a completion:
    \begin{center}$
        \frg \mathop{\hat\otimes} \mm_R := \varprojlim_n \frg \otimes \hbar \cdot k[\hbar]/{\ord{\hbar^n}} = \hbar \frg\formpow{\hbar}.
    $\end{center}
    For $x \in (\frg \mathop{\hat\otimes} \mm_R)^1 = \hbar \frg^1\formpow{\hbar}$, the associated Maurer--Cartan equation is called the \textbf{formal Maurer--Cartan equation}.
\end{remark}

Before closing the section, we mention a classical result in deformation theory. We will present and prove a generalised version of it in Section \ref{sec:L_infty_MC}.

\begin{theorem}[Quasi-Isomorphism Theorem for DGLAs]{thm:II_DGLA_quasi}
    Let $f\colon \frg \to \frg'$ be a morphism of DGLAs. If $f$ is a quasi-isomorphism, then $f$ induces an natural isomorphism of the Maurer--Cartan functors:
    \begin{center}
        $\MCbar_{\frg} \simeq \MCbar_{\frg'}.$
    \end{center}
\end{theorem}
\begin{proof}
    See Theorem 3.1 of \cite{manetti05}.
\end{proof}

\section{Polydifferential Operators and Polyvector Fields}\label{sec:poly_DGLA}

We have seen how the differential graded Lie algebra controls the deformation problem of a general associative algebra. In this section we return to the settings in differential geometry. Let $M$ be a smooth manifold. Denote by $A := \smooth{M}$ the smooth functions on $M$.

\begin{definition}
    An \textbf{$\bm n$-differential operator} $f \in \Hom_\R(A^{\otimes n},A)$ is a $n$-linear map such that $f$ is a differential operator in each argument. The set of all $n$-differential operators is denoted by $\dpoly[n-1]$. Note that $\dpoly[n]\subseteq C^{n+1}(A,A)$ and $\delta(\dpoly[n]) \subseteq \dpoly[n+1]$, where $\delta$ is the Hochschild differential. Therefore $\dpoly[\bullet]$ is a subcomplex of the shifted Hochschild complex $C^\bullet(A,A)[1]$, and is called the \textbf{differential Hochschild complex}. Just like $C^\bullet(A,A)[1]$, $\dpoly[\bullet]$ with the modified Hochschild differential $\dhochm$ and the Gerstenhaber bracket $\Ger{-,-}$ is a DGLA.
\end{definition}

With $R = \R\formpow{\hbar}$ and $\mm_R = \hbar\R\formpow{\hbar}$, an analogy of Proposition \ref{prop:II_defo_MC} is the following:

\begin{proposition}[Star Products via Maurer--Cartan Equation]{prop:II_star_MC}
    The set of equivalence classes of star products on $M$ is in bijective correspondence with the set of solutions of the \emph{formal} Maurer--Cartan equation for $\dpoly$ modulo gauge actions:
    \begin{center}$
        \MCbar_{\dpoly}(\R\formpow{\hbar}) = \MC(\hbar\dpoly\formpow{\hbar})/\exp(\hbar\dpoly[0]\formpow{\hbar}).
    $\end{center}
\end{proposition}

\subsection{Polyvector Fields}

To study the deformations of the Poisson structure on $M$, we need to introduce another differential graded Lie algebra.

\begin{definition}
    Let $\tanb M$ be the tangent bundle of $M$. An \textbf{$\bm n$-vector field} is a section of the $n$-th exterior product $\bigwedge^n \tanb M$. The set of $n$-vector fields is denoted by $\tpoly[n-1] := \Gamma(M,\bigwedge^{n} \tanb M)$. The set of polyvector fields is the graded vector space given by 
    \begin{center}$
        \tpoly = \displaystyle\bigoplus_{i=-1}^\infty \tpoly[i] = \Gamma(M,\textstyle\bigwedge^\bullet \tanb M)[1],
    $\end{center}
    where $\bigwedge^0 \tanb M = M$ by convention. Also notice that we have shifted the degree by $1$, as in the case of $\dpoly$.
\end{definition}

Recall that $\tpoly[0] = \Gamma(M,\tanb M)$ has a Lie algebra structure given by the Lie derivatives:
\begin{center}$
    [X,Y] := \lag_X Y.
$\end{center}
It has a unique extension to the graded Lie algebra structure on $\tpoly$:
\newcommand{\SN}[1]{\left[ #1 \right]_{\mathsf{SN}}}
\begin{definition}
    For $X := X_1 \wedge \cdots \wedge X_{p} \in \tpoly[p-1]$ and $Y := Y_1 \wedge \cdots \wedge Y_{q} \in \tpoly[q-1]$, we define the \textbf{Schouten--Nijenhuis bracket} $\SN{X, Y} \in \tpoly[p+q-2]$ as 
    \begin{center}$
        \SN{X, Y} := \displaystyle\sum_{r=1}^{p}\displaystyle\sum_{s=1}^{q} (-1)^{r+s}[X_r, Y_s]\wedge X_1\wedge\cdots\wedge\hat X_r\wedge\cdots\wedge X_{p}\wedge Y_1\wedge\cdots\wedge\hat Y_s\wedge\cdots\wedge Y_{q}.
    $\end{center}
    where $\hat X_i$ denotes the omission of $X_i$ in the expression.
    If $f \in \smooth{M} = \tpoly[-1]$, we define:
    \begin{center}$
        \SN{X, f} := \displaystyle\sum_{r=1}^p (-1)^{p-r} {X_r}(f)\cdot X_1 \wedge \cdots \wedge \hat{X}_r \wedge \cdots \wedge X_p.
    $\end{center}  
    Finally we extend the bracket bilinearly to $\tpoly[\bullet]$. 
\end{definition}

\begin{lemma}[Properties of Schouten--Nijenhuis Bracket]{}
    The Schouten--Nijenhuis bracket is the unique graded Lie bracket on $\tpoly$ such that 
    \begin{itemize}
        \item $\SN{X,f} = \lag_X f = X(f)$ for $X \in \Gamma(M,\tanb M)$ and $f \in \smooth{M}$;
        \item $\SN{X,Y} = \lag_X Y$ for $X,Y \in \Gamma(M,\tanb M)$;
        \item $\SN{X, Y \wedge Z} = \SN{X, Y} \wedge Z + (-1)^{(p-1)q}Y \wedge \SN{X , Z}$  
        \hfill {($*$)}

        for $X,Y,Z \in \tpoly$ homogeneous of degree $p,q,r$ respectively.
    \end{itemize}
\end{lemma}
\begin{proof}
    We can verify all properties (graded skew-symmetry, graded Jacobi identity, equation ($*$)) of the Schouten--Nijenhuis bracket directly by definition, though some brute force computations are inevitable. We omit the details, which can be found in, e.g.\ \cite{cattaneo04notes}. The uniqueness is proven by inductions on the degree of $X$ and $Y$ consecutively.
\end{proof}

\begin{remark}
    Now $\tpoly$ with the Schouten--Nijenhuis bracket $\SN{-,-}$ and the \emph{zero differential} is a DGLA. 
\end{remark}

In local coordinates, a polyvector field $X \in \tpoly[p+1]$ has the expansion in terms of the frame vector fields:
\begin{center}$
    X = \displaystyle\sum_{i_1,...,i_p = 1}^{\dim M}X^{i_1,...,i_p}\pa_{i_1} \wedge \cdots \wedge \pa_{i_k}.
$\end{center}
Analogous to the Gerstenhaber product, for $X \in \in \tpoly[p-1]$ and $Y \in \in \tpoly[q-1]$ we can define the product $X \diamond Y \in \tpoly[p+q-2]$ by:
\begin{equation} \hspace*{-1.5em}
    X \diamond Y := \displaystyle\sum_{r = 1}^p\displaystyle\sum_{i_1,...,i_p,j_1,...,j_q = 1}^{\dim M} (-1)^{r+1}X^{i_1,...,i_p}\pa_{r}Y^{j_1,...,j_q}\pa_{i_1} \wedge \cdots \wedge \hat{\pa_{i_r}} \wedge \cdots \wedge \pa_{i_p} \wedge \pa_{j_1} \wedge \cdots \wedge \pa_{j_q} \label{equ:II_polyvec_prod}
\end{equation}
\begin{lemma}[]{}
    The Schouten--Nijenhuis bracket can be expressed in terms of the product defined in \eqref{equ:II_polyvec_prod} as:
    \begin{center}$
        \SN{X,Y} = (-1)^{p}X \diamond Y - (-1)^{(p-1)q}Y \diamond X
    $\end{center}
    where $X \in \tpoly[p]$ and $Y \in \tpoly[q]$.
\end{lemma}
\begin{proof}
    See Lemma IV.2.1 of \cite{arnal02}.
\end{proof}

\begin{remark}
    Let $V := \Gamma(M,\tanb M)$. Then $\tpoly[p-1] = \bigwedge^{p}V$. Using the d\'ecalage isomorphism given in \eqref{equ:II_dec_iso}, it is isomorphic to $\operatorname{Sym}^{p}(V[1])[-p]$ via $\dec_{p}$. In this way the product \eqref{equ:II_polyvec_prod} has the following compact form:
    \begin{equation}
        X \diamond Y = \dec_{p+q-1}\left( \displaystyle\sum_{r=1}^{\dim M}\dfrac{\pa \dec^{-1}_{p}(X)}{\pa \psi_i}\dfrac{\pa \dec^{-1}_{p}(Y)}{\pa x_i} \right), \label{equ:II_polyvec_prod_dec}
    \end{equation}
    where $X \in \tpoly[p-1]$, $Y \in \tpoly[q-1]$, and $\psi_i := \dec^{-1}_1(\pa_i)$ is an odd variable of degree $1$.
\end{remark}

\subsection{Poisson Deformations}

Now we would like to show how the DGLA $\tpoly$ controls the deformations of the Poisson bracket on $M$. First we shall make the definition precise. For a general Poisson algebra, the Poisson deformations can be formulated analogously to Definition \ref{def:II_R_defo}. As before, We denote by $k$ a ground field of characteristic $0$.

\begin{definition}
    Let $R$ be a test algebra. Let $A_R := A \otimes_k R$. A \textbf{Poisson $\bm{R}$-deformation} of the Poisson bracket $\pi_0$ on $A$ is a Poisson bracket $\pi\colon A_R \otimes_k A_R \to A_R$ such that $\pi = \pi_0$ on $A_R/\mm_R \otimes_k A_R/\mm_R$. The Poisson $R$-deformations $\pi$ and $\pi'$ are said to be equivalent, if there exists an $R$-module automorphism $\varphi\colon A_R  \to A_R$ such that $\varphi(\pi(u,v)) = \pi'(\varphi(u),\varphi(v))$ and $\varphi = \id$ on $A_R/\mm_R$.
\end{definition}

For $A = \smooth{M}$ and $R = \R\formpow{\hbar}$, a Poisson bracket on $A_R$ is an element of $\tpoly[1]\formpow{\hbar}$. We say that it defines a \textbf{formal Poisson structure} on $M$ in contrast to an actual Poisson structure on $M$, which is an element of $\tpoly[1]$.

\begin{lemma}[]{}
    The bivector field $\pi$ on $M$ has vanishing Schouten--Nijenhuis bracket $\SN{\pi,\pi} = 0$ if and only if $\pi$ is a Poisson bracket.
\end{lemma}
\begin{proof}
    In local coordinates, $\pi = \displaystyle\sum_{i,j=1}^{\dim M}\pi^{ij}\pa_i \wedge \pa_j$. The equation $\SN{\pi,\pi} = 0$ becomes:
    \begin{center}$
        \displaystyle\sum_{i,j,k,\ell=1}^{\dim M}\pi^{ij}\pa_j\pi^{k\ell}\pa_i \wedge \pa_k \wedge \pa_\ell = 0.
    $\end{center}
    On the other hand, $\pi$ defines a Poisson structure if and only if it satisfies the Jacobi identity:
        \begin{center}
            $\pi(\pi(f,g),h) + \pi(\pi(g,h),f) + \pi(\pi(h,f),g) = 0, \qquad f,g,h \in \smooth{M}.$
        \end{center}
    In local coordinates, this becomes
    \begin{center}$
        \displaystyle\sum_{i,j,k,\ell=1}^{\dim M}\pi^{ij}\pa_j\pi^{k\ell}\left( \pa_i f \pa_k g \pa_\ell h + \pa_i g \pa_k h \pa_\ell f + \pa_i h \pa_k f \pa_\ell g \right).
    $\end{center}
    So we have the equivalence as claimed.
\end{proof}

\begin{proposition}[Poisson Structures via Maurer--Cartan Equation]{prop:II_pois_MC}
    The set of equivalence classes of formal Poisson structures on $M$ is in bijective correspondence with the set of solutions of the \emph{formal} Maurer--Cartan equation for $\tpoly$ modulo gauge actions:
    \begin{center}
        $\MCbar_{\tpoly}(\R\formpow{\hbar}) = \MC(\hbar\tpoly\formpow{\hbar})/\exp(\hbar\tpoly[0]\formpow{\hbar}).$
    \end{center}
\end{proposition}
\begin{proof}
    The formal Maurer--Cartan equation for $\tpoly$ is given by 
    \begin{center}$
        \d\pi + \dfrac{1}{2}\SN{\pi,\pi}= \dfrac{1}{2}\SN{\pi,\pi} = 0, \qquad \pi \in \hbar\tpoly[1]\formpow{\hbar}.
    $\end{center}
    The previous lemma shows that the set of solutions corresponds to the set of formal Poisson structures. By a similar proof as Proposition \ref{prop:II_defo_MC}, it can be shown that the equivalence relation on formal Poisson structures is exactly generated by the gauge actions $\exp(\hbar\tpoly[0]\formpow{\hbar})$.
\end{proof}

\begin{remark}
    Combining Theorem \ref{thm:II_DGLA_quasi}, Proposition \ref{prop:II_star_MC}, and Proposition \ref{prop:II_pois_MC}, it is tempting to look for a quasi-isomorphism of DGLAs between $\tpoly$ and $\dpoly$, in which case we would obtain the claimed bijective correspondence \eqref{equ:I_defo_quan_bij}. Unfortunately such quasi-isomorphism does not exist. In the next part we will construct a quasi-isomorphism of graded vector spaces, which does not preserve the graded Lie brackets. This is the first step towards the construction of a weaker notion of quasi-isomorphism, which will be introduced in the Section \ref{sec:L_infty}.
\end{remark}

\subsection{Hochschild--Kostant--Rosenberg Theorem}\label{sec:HKR}
For a polyvector field $X = X_1\wedge \cdots \wedge X_n \in \tpoly[n-1]$, $X$ acts on $ f_1,...,f_n \in \smooth{M}$ as 
\begin{center}$
    (X_1\wedge \cdots \wedge X_n)(f_1,...,f_n) = \dfrac{1}{n!}\det(X_i(f_j))_{i,j}.
$\end{center}
This identifies $X$ as a polydifferential operator which is skew-symmetric and 1-differential in each argument. Therefore we obtain a morphism of graded vector spaces,
\begin{equation}
    U_1^{(0)}\colon \tpoly \to \dpoly \label{equ:II_HKR}
\end{equation}
which is equal to the identity map at degree $-1$. We called it the \textbf{HKR map}.

Vey proved the following version of the Hochschild--Kostant--Rosenberg Theorem
\footnote{
    Originally, the HKR theorem was purposed by Hochschild, Kostant, and Rosenberg (\cite{HKR}) for smooth affine varieties.
}
(\cite{vey75}), which says that $U_1^{(0)}$ is a quasi-isomorphism of graded vector spaces.

\begin{theorem}[Smooth Hochschild--Kostant--Rosenberg Theorem]{thm:II_HKR}
    The HKR map $U_1^{(0)}\colon \tpoly \to \dpoly$ induces isomorphisms in the cohomology spaces
    \begin{center}$
        \mathrm H^n(U_1^{(0)})\colon \tpoly[n] = \mathrm H^n(\tpoly) \xrightarrow{\sim} \mathrm H^n(\dpoly).
    $\end{center}
\end{theorem}
\begin{proof}
    First we check that $U_1^{(0)}$ is a chain map. Since $\tpoly$ has zero differential, we must have $\dhochm\circ U_1^{(0)} = 0$. Indeed, for $X = X_0 \wedge \cdots \wedge X_{p} \in \tpoly[p]$ and $f_0,...,f_n \in \smooth{M}$, we have 
    \begin{align*}\hspace{-4em}
        (-1)^p\dhochm(U_1^{(0)}(X))(f_0,...,f_{p+1})
        &= f_0U_1^{(0)}(X)(f_1,...,f_{n+1}) \\
        + \displaystyle\sum_{r=1}^{n+1}(-1)^r&U_1^{(0)}(X)(f_0,...,f_{r-1}f_r,...,f_{n+1})+ (-1)^pU_1^{(0)}(X)(f_0,...,f_{n})f_{n+1} \\
        &= f_0U_1^{(0)}(X)(f_1,...,f_{n+1}) + \displaystyle\sum_{r=1}^{n+1}(-1)^r f_{r}U_1^{(0)}(X)(f_0,...,\hat f_r,...,f_{n+1})\\
        + \displaystyle\sum_{r=1}^{n+1}(-1)^r &f_{r-1}U_1^{(0)}(X)(f_0,...,\hat f_{r-1},...,f_{n+1}) + (-1)^pf_{n+1}U_1^{(0)}(X)(f_0,...,f_{n}) \\
        &= 0, 
    \end{align*}
    where we used the fact that $U_1^{(0)}(X)$ acts on each argument as a derivation. Let $\mathrm H^n(U_1^{(0)})\colon \tpoly[n] \to \mathrm H^n(\dpoly)$ be the induced map in the $n$-th cohomology. Note that $U_1^{(0)}$ embed $\tpoly$ into $\dpoly$ as a sub-complex of skew-symmetric 1-differential operators. So $\mathrm H^n(U_1^{(0)})$ is injective. To show that it is surjective, we shall show that every $\alpha \in \dpoly[p]$ is of the form:
    \begin{center}$
        \alpha = \dhochm \beta + \gamma
    $\end{center}
    for some $\beta \in \dpoly[p-1]$ and $\gamma \in U_1^{(0)}(\tpoly[p])$. By a standard argument by the partition of unity, it suffices to prove this in a local chart (\emph{cf.\ }Proposition 2.13--15 of \cite{gutt99}). So we may assume that $M = \R^d$. We sketch the proof in this case following \cite{gutt99}.

    We use induction on the degree $p$ of $\alpha \in \dpoly[p]$. For $p=0$, cocycles are vector fields so the result is true trivially. Suppose the result holds for all $r<p$. For $\alpha \in \dpoly[p]$, consider the term of highest order $q$ in $f_0$:
    \begin{center}$
        \alpha(f_0,...,f_n) = \displaystyle\sum_{|I|=q}\pa_{I} f_0 \cdot \eta_{I}(f_1,...,f_n) + \cdots = \displaystyle\sum_{|I|=q}(\pa_{I} \otimes \eta_{I})(f_0\otimes \cdots\otimes f_n) + \cdots,
    $\end{center}
    where $I = (i_1,...,i_q)$ is a multi-index, and $\eta \in \dpoly[p-1]$. Taking the differential, we note that $\alpha$ is a $p$-cocycle implies that $\eta_{I}$ are $(p-1)$-cocycles. In such case, by induction hypothesis, $\eta_{I} = \dhochm \xi_{I} + \zeta_{I}$ for some $\xi_{I} \in \dpoly[p-1]$ and $\zeta_{I} \in U_1^{(0)}(\tpoly[p])$. We claim that $\alpha$ is cohomologous to $\alpha' \in \dpoly[p]$ with the differential term highest order $q'$ in $f_0$ such that $q'<q$. This follows from a careful analysis of the expansions of the terms $\zeta_{I}$ and $\psi$ arising from
    \begin{center}$
        \alpha + \dhochm\left( \displaystyle\sum_{|I| = q} \pa_{I} \otimes \xi_{I} \right) = \displaystyle\sum_{|I| = q} \pa_{I} \otimes \zeta_{I} + \psi.
    $\end{center}
    Refer to \cite{gutt99} for details. By induction, we reduce to the case $q=1$, where
    \begin{center}$
        \alpha = \displaystyle\sum_{i=0}^p \pa_i \otimes \eta_i.
    $\end{center}
    With $\eta_i = \dhochm \xi_i + \zeta_i$ as before, we have 
    \begin{center}$
        \alpha = -\dhochm \left( \displaystyle\sum_{i=0}^p \pa_i \otimes \xi_i \right) + \displaystyle\sum_{i=0}^p \pa_i \otimes \zeta_i,
    $\end{center} 
    where $\displaystyle\sum_{i=0}^p \pa_i \otimes \zeta_i$ is 1-differential in all arguments. Finally, we can show that every 1-differential operator is cohomologous to its total skew-symmetrisation. Therefore $\alpha = \dhochm\beta + \gamma$ for some $\beta \in \dpoly[p-1]$ and $\gamma \in U_1^{(0)}(\tpoly[p])$ as claimed.
\end{proof}
    
    
The {HKR map} $U_1^{(0)}$ is not a morphism of DGLAs, because the Schouten--Nijenhuis bracket is not mapped to the Gerstenhaber bracket. However, the failure of $U_1^{(0)}$ being a graded Lie algebra homomorphism can be resolved by finding a morphism $U$ whose first-order approximation $U_1 = U_1^{(0)}$. This is the motivation of introducing $\linf$-algebras and $\linf$-morphisms in the next section.

\section[\texorpdfstring{${\mathrm L_\infty}$}{L{\textunderscore}infinity}-Algebras]{$\bm{\mathrm L_\infty}$-Algebras} \label{sec:L_infty}

In this section we discuss the $\linf$-algebras. We first present an abstract definition using graded coalgebras. This allows us to define $\linf$-morphisms in a clean way. Then we show its equivalence with another definition using the Taylor coefficients, which presents $\linf$-algebras as a generalisation of DGLAs in a natural way.

\setcounter{subsection}{-1}
\subsection{Sign Conventions for Tensor Algebras}

Before going into the main topic, we quickly review the constructions and sign conventions which will be useful for the rest of the dissertation.

\begin{convention}
    We adopt the \textbf{Koszul sign convention}. Let $V$ and $W$ be $\Z$-graded vector spaces. Let $f\colon V \to V$ and $g\colon W \to W$ be linear maps of degree $|f|$ and $|g|$ respectively. Then $f \otimes g \in \operatorname{End}(V \otimes W)$ is defined as
    \begin{center}$
        (f \otimes g)(x \otimes y) := (-1)^{|g|\cdot|x|}f(x) \otimes g(y),
    $\end{center}
    where $x \in V$ and $y \in W$ are homogeneous. Colloquially, we pick up a sign $(-1)^{ab}$ whenever we commute two graded objects of degree $a$ and $b$.
\end{convention}

\begin{definition}
    Let $V = \bigoplus_{i \in \Z} V^i$ be a graded $k$-vector space. Let $V^{\otimes n}$ be the $n$-th tensor product of $V$ with itself. For homogeneous elements $x,y$, we define the \textbf{twisting map} $\tau$ by 
    \begin{center}$
        \tau(x \otimes y) = (-1)^{|x|\cdot |y|}(y \otimes x).
    $\end{center}
    The symmetric product and the exterior product are identified as quotient spaces of $V^{\otimes n}$:
    \begin{center}$
        \operatorname{Sym}^n V := V^{\otimes n}/{\ord{x \otimes y - \tau(x \otimes y)}}, \qquad \textstyle\bigwedge^n V := V^{\otimes n}/{\ord{x \otimes y + \tau(x \otimes y)}}.
    $\end{center}
    The image of $x_1 \otimes \cdots \otimes x_n$ under the quotients are denoted by $x_1\cdots x_n$ and $x_1 \wedge \cdots \wedge x_n$ respectively.
    These constructions give rise to the following algebras:
    \begin{multicols}{2}
        \begin{itemize}[nosep]
            \item {Tensor algebra}: 
            \item {Reduced tensor algebra}: 
            \item {Symmetric algebra}: 
            \item {Reduced symmetric algebra}: 
            \item {Exterior algebra}: 
            \item {Reduced exterior algebra}: 
        \end{itemize}
        $\bigotimes (V) := \bigoplus_{n=0}^\infty V^{\otimes n}$; \\
        $\bigotimes^+ (V) := \bigoplus_{n=1}^\infty V^{\otimes n}$; \\
        $\operatorname{Sym} (V) := \bigoplus_{n=0}^\infty V^{\otimes n}$; \\
        $\operatorname{Sym}^+ (V) := \bigoplus_{n=1}^\infty \operatorname{Sym}^n V$; \\
        $\bigwedge (V) := \bigoplus_{n=0}^\infty \bigwedge^n V$; \\
        $\bigwedge^+ (V) := \bigoplus_{n=1}^\infty \bigwedge^n V$.
    \end{multicols}
\end{definition}
\begin{remark}
    These algebras are graded by the induced grading from $V$:
    \begin{center}$
        |x_1 \otimes \cdots \otimes x_n| = \displaystyle\sum_{i=1}^n |x_i|.
    $\end{center}
\end{remark}
\begin{convention}
    Let $x := x_1 \otimes \cdots \otimes x_n \in V^{\otimes n}$. We identify $x$ with its projection in $\operatorname{Sym}^n V$ (\emph{resp.\ }in $\bigwedge^n V$). Let $\sigma \in S_n$ be a permutation. The \textbf{Koszul sign} $\varepsilon_x(\sigma)$ (\emph{resp.\ }$\chi_x(\sigma)$) is defined as the change of sign when permuting the product $x$ in the symmetric (\emph{resp.\ }exterior) algebra:
    \begin{center}$
        x_1 \cdots x_n = \varepsilon_x(\sigma)x_{\sigma(1)} \cdots x_{\sigma(n)}, \qquad x_1 \wedge \cdots \wedge x_n = \chi_x(\sigma)x_{\sigma(1)} \wedge \cdots \wedge x_{\sigma(n)}.
    $\end{center}
\end{convention}
Unlike the ungraded case, the definition of twisting map in a graded space allows the exchange of symmetric and exterior product by a shift in the degree. This is encoded in the following lemma.

\begin{lemma}[D\'ecalage Isomorphism]{}
    There exists a canoncial isomorphism of degree $0$ between the symmetric algebra and the shifted exterior algebra, namely the \textbf{d\'ecalage isomorphism}, given by
    \begin{equation}
        \begin{tikzcd}[column sep = small, row sep = 0em]
            {\dec_n\colon } & {\operatorname{Sym}^n(V[1])} & {\textstyle\bigwedge^n(V)[n]} \\
            & {x_1 \cdots x_n} & {(-1)^{\sum_{i=1}^n (n-i)(|x_i|-1)}x_1 \wedge \cdots \wedge x_n}
            \arrow[from=1-2, to=1-3]
	\arrow[maps to, from=2-2, to=2-3]
        \end{tikzcd} \label{equ:II_dec_iso}
    \end{equation}
\end{lemma}

\subsection[\texorpdfstring{${\mathrm L_\infty}$}{L{\textunderscore}infinity}-Algebras via Coalgebras]{$\bm{\mathrm L_\infty}$-Algebras via Coalgebras}

First of all, we need to introduce the coalgebras over a field $k$, which is the category-theoretic dual of the associative algebras. The notions of co-associativity, co-commutativity, co-freeness, and so on can be defined as for their counterparts for algebras, but with all arrows reversed in the corresponding commutative diagrams.


\begin{definition}
    A \textbf{co-associative coalgebra} $C$ over a field $k$ is a $k$-vector space equipped with a $k$-linear co-multiplication $\Delta\colon C \to C \otimes_k C$ satisfying the co-associativity:
    \begin{center}$
        (\id_C \mathrel{\otimes} \Delta)\circ\Delta = (\Delta \otimes \id_C)\circ\Delta.
    $\end{center}
    A co-unit of $C$ is a $k$-linear 
    map $e\colon C \to k$ satisfying
    \begin{center}$
        (\id_C \mathrel{\otimes} e)\circ\Delta = (e \otimes \id_C)\circ\Delta = \id_C.
    $\end{center}
    The commutative diagrams are shown below.
    \begin{center}
        \begin{tikzcd}
            C && {C \otimes_k C} && C && {C \otimes_k C} \\
            {C \otimes_k C} && {C \otimes_k C \otimes_k C} && {C \otimes_k C} && C
            \arrow["\Delta", from=1-1, to=1-3]
            \arrow["\Delta"', from=1-1, to=2-1]
            \arrow["{\Delta \otimes \id_C}", from=2-1, to=2-3]
            \arrow["{\id_C \otimes \Delta}", from=1-3, to=2-3]
            \arrow["\Delta"', from=1-5, to=2-5]
            \arrow["\Delta", from=1-5, to=1-7]
            \arrow["{\id_C \otimes e}", from=1-7, to=2-7]
            \arrow["{e \otimes \id_C}", from=2-5, to=2-7]
        \end{tikzcd}
    \end{center}

    A \textbf{graded coalgebra} $C = \bigoplus_{i \in \Z}C^i$ is both a graded vector space and a coalgebra, such that the co-multiplication is compatible with the grading:
    \begin{center}$
        \Delta(C^i) \subseteq \displaystyle\bigoplus_{p+q=i}(C^p \otimes_k C^q).
    $\end{center}
    A morphism $F\colon (C,\Delta) \to (C',\Delta')$ of graded coalgebras is a linear map $F \colon C \to C'$ of degree $0$ such that $\Delta' \circ F = (F \otimes F) \circ \Delta$. In \cite{kon97}, such morphism is called a \emph{pre-$\linf$-morphism}.
\end{definition}

\begin{definition}
    We introduce some more properties. A graded coalgebra $C$ is said to be
    \begin{itemize}
        \item \textbf{graded co-commutative}, if $\Delta = \tau\circ\Delta$, where $\tau$ is the {twisting map};
        \item \textbf{co-nilpotent}, if for each $x \in C$ there exists $n \in \N$ such that $\Delta^n(x) = 0$, where $\Delta^n$ is defined recursively by $\Delta^{n+1} = (\Delta \otimes \id_{C^{\otimes n}}) \circ \Delta^n$.
    \end{itemize}
\end{definition}

\newcommand{\rsym}[1]{\mathrm{Sym}^+({#1})}
\newcommand{\otimesb}{\mathop{\bar{\otimes}}}
\newcommand{\shff}{\operatorname{Shff}}
\begin{example}\label{eg:II_rsym}
    Let $V = \bigoplus_{i \in \Z} V^i$ be a graded $k$-vector space. Recall that the \textbf{reduced symmetric algebra} is given by
    \begin{center}$
        \rsym{V} := \displaystyle\bigoplus_{n=1}^\infty \mathrm {Sym}^n(V) = \displaystyle\bigoplus_{n=1}^\infty V^{\otimes n}{\left/\ord{x \otimes y - (-1)^{|x|\cdot|y|} y \otimes x}\right.},
    $\end{center}
    We can put a graded coalgebra structure on $\rsym{V}$. The co-multiplication is defined as 
    \begin{center}$
        \Delta(x_1  \cdots  x_n) := \displaystyle\sum_{i=1}^{n-1}\displaystyle\sum_{\sigma \in S(i,n-i)} \varepsilon_x(\sigma) (x_{\sigma(1)}  \cdots  x_{\sigma(i)}) \otimes (x_{\sigma(i+1)}  \cdots  x_{\sigma(n)}),
    $\end{center}
    where the \textbf{$\bm{(i,n-i)}$-shuffles} $S(i,n-i)$ is the set of $\sigma \in S_n$ such that $\sigma(1)<\cdots<\sigma(i)$ and $\sigma(i+1)<\cdots<\sigma(n)$, and the {Koszul sign} $\varepsilon_x(\sigma)$ is defined in the previous part.

    The reduced symmetric space $\rsym{V}$ is the \textbf{co-free}
    \footnote{
        The authors of \cite{doubek07} point out that, contrary to the common belief (e.g.\ in \cite{kon97}), $\rsym{V}$ is \emph{not} co-free in the category of co-commutative coalgebras without co-unit. The subtlety is discussed in Section II.3.7 of \cite{markl02}. Also see \cite{grassi99} for a complete proof of the claims.
    }
    co-nilpotent co-commutative coalgebra without co-unit co-generated by $V$. By co-freeness we mean that $\operatorname{Sym}^+$ is the right adjoint functor of the forgetful functor to the category of differential graded vector spaces. In this case, it has the universal property that any homomorphism of co-nilpotent co-commutative coalgebras $A \to V$ uniquely lifts to a homomorphism $A \to \rsym{V}$.
    \begin{center}
        \begin{tikzcd}
            & {\rsym{V}} \\
            A & V
            \arrow[from=1-2, to=2-2]
            \arrow["{\exists\, !}", dashed, from=2-1, to=1-2]
            \arrow[from=2-1, to=2-2]
        \end{tikzcd}
    \end{center}
\end{example}

\begin{definition}
    A \textbf{co-derivation} $\delta$ of degree $d$ on a graded coalgebra $C$ is a graded $k$-linear map $\delta\colon C^i \to C^{i+d}$ satisfying the co-Leibniz rule:
    \begin{center}$
        \Delta \circ \delta = (\delta \otimes \id_C + \id_C \mathrel\otimes \delta)\circ \Delta
    $\end{center}
\end{definition}

\begin{definition}
    An \textbf{$\bm{\mathrm L_\infty}$-algebra} $(\frg,Q)$ is a graded $k$-vector space $\frg$ with a co-derivation $Q$ of degree $1$ on the reduced symmetric space $\rsym{\frg[1]}$ such that $Q^2= 0$. A \textbf{$\bm{\linf}$-morphism} between $\mathrm L_\infty$-algebras, $F\colon (\frg,Q) \to (\frg',Q')$, is a morphism of graded coalgebras $F\colon \rsym{\frg[1]} \to \rsym{\frg'[1]}$ such that $F \circ Q = Q' \circ F$.
\end{definition}

\subsection[\texorpdfstring{${\mathrm L_\infty}$}{L{\textunderscore}infinity}-Algebras via Taylor Coefficients]{$\bm{\mathrm L_\infty}$-Algebras via Taylor Coefficients}

\begin{definition}
    Let $F\colon (\frg,Q) \to (\frg',Q')$ be a morphism of $\linf$-algebras. The \textbf{Taylor coefficients} of $F$ (\emph{resp.\ }of $Q$) are the sequence of maps $F_n\colon \operatorname{Sym}^n(\frg[1]) \to \frg'[1]$ (\emph{resp.\ }$Q_n \colon \operatorname{Sym}^n(\frg[1]) \to \frg[1]$). By the universal property of $\rsym{\frg[1]}$, the morphism $F$ (\emph{resp.\ }the codifferential $Q$) is uniquely determined by its Taylor coefficients. 
\end{definition}

Under the d\'ecalage isomorphisms \eqref{equ:II_dec_iso}, a Taylor coefficient $Q_n$ corresponds to a multi-bracket $\ell_n = Q_n \circ \dec_n^{-1} \colon \bigwedge^n(\frg)[n] \to \frg$ of degree $2-n$.

\begin{lemma}[Reconstruction from Taylor Coefficients]{}
    Let $F\colon (\frg,Q) \to (\frg',Q')$ be a morphism of $\linf$-algebras. Then $Q$ and $F$ are given in terms of the Taylor coefficients by
    \begin{align*}
        Q(x_1\cdots x_n) &= \displaystyle\sum_{I,J}^{\left\{ 1,...,n \right\}}
        \varepsilon_x(I,J) Q_{|I|}(x_I) \cdot x_J; \\
        F(x_1\cdots x_n) &= \displaystyle\sum_{k=1}^\infty\dfrac{1}{k!}\displaystyle\sum_{I_1,...,I_k}^{\left\{ 1,...,n \right\}}\varepsilon_x(I_1,...,I_k) F_{|I_1|}(x_{I_1}) \cdots F_{|I_k|}(x_{I_k}).
    \end{align*}
\end{lemma}
\begin{proof}
    See Theorem III.2.1 of \cite{arnal02}.
\end{proof}
\begin{remark}
    In the lemma $I_1,...,I_k$ are multi-indices such that $(I_1,...,I_k)$ is a partition of $\left\{ 1,...,n \right\}$. The sum is over all such possible partitions. For each partition $(I_1,...,I_k)$, it defines an ordering of $\left\{ 1,...,n \right\}$ and hence a permutation $\sigma \in S_n$ in the following way: all elements of $I_i$ precedes those of $I_j$ if $i<j$; in each $I_i$ the order is the usual ordering of natural numbers. $\varepsilon_x(I_1,...,I_k)$ is the Koszul sign associated to this permutation. The convention for $(I,J)$ is identical.
\end{remark}

The condition $Q^2 = 0$ can be expanded into an infinite sequence of constraints imposing on the Taylor coefficients $Q_n$. For convenience, we denote by $Q^m_n\colon \operatorname{Sym}^n(\frg[1]) \to \operatorname{Sym}^m(\frg[1])$ the map induced from $Q$ by restriction on the domain and projection on the codomain. In particular $Q^1_n = Q_n$ is the $n$-th Taylor coefficient. Then $Q^2 = 0$ is equivalent to 
\begin{equation}
    \displaystyle\sum_{i=1}^n Q^1_i\circ Q^i_n = 0, \qquad \forall\, n \in \Z_{>0}. \label{equ:Q_sqzero}
\end{equation}
\begin{itemize}
    \item For $n=1$, \eqref{equ:Q_sqzero} gives $Q_1 \circ Q_1 = 0$, which implies that $Q_1$ is a differential of degree $1$ on $\frg$. Let $\d x := (-1)^{|x|}Q_1(x)$ for homogeneous $x \in \frg$.
    \item For $n=2$, \eqref{equ:Q_sqzero} gives $Q_1(Q_2(x \cdot y)) + Q_2(Q_1(x) \cdot y + (-1)^{|x|-1}x \cdot Q_1(y)) = 0$. If $[x,y] := \ell_2(x \wedge y)$, then this becomes the graded Leibniz rule of $[-,-]$.
    \item For $n = 3$, \eqref{equ:Q_sqzero} gives 
    \begin{align*}\hspace*{-3em}
        \begin{aligned}
            &(-1)^{|x|\cdot|z|}[[x, y], z]+(-1)^{|y|\cdot|z|}[[z, x], y]+(-1)^{|x|\cdot|y|}[[y, z], x] \\
            &\ =(-1)^{|x|\cdot|z|}\left(\d (\ell_{3}(x \wedge y \wedge z))+ \ell_{3}(\d x \wedge y \wedge z)+(-1)^{|x|} \ell_{3}(x \wedge \d y \wedge z)+(-1)^{|x|+|y|} \ell_{3}(x \wedge y \wedge \d z)\right).
        \end{aligned}
    \end{align*}
    for homogeneous $x,y,z \in \frg$. This reduces to the graded Jacobi identity for $[-,-]$ if $\ell_3 = 0$.
\end{itemize}
Therefore, the $\linf$-algebras can be considered as the generalisation of DGLAs where the graded Jacobi identity for the graded Lie bracket $\ell_2$ holds only up to homotopy of a higher bracket $\ell_3$. In particular, we have the following result:
\begin{proposition}[From $\bm{\linf}$-Algebras to DGLAs]{}
    An $\linf$-algebra $(\frg,Q)$ defines a DGLA structure on $\frg$ if and only if $Q_n = 0$ for $n \geq 3$.
\end{proposition}
\begin{remark}
    In general, it can be shown (\emph{cf.\ }\cite{doubek07}, \cite{jurvco19}) that the constraint $Q \circ Q = 0$ is equivalent to the \emph{homotopy Jacobi identities} satisfied by the multi-brackets:
    \begin{center}$
        \displaystyle\sum_{i=1}^n (-1)^i\displaystyle\sum_{\sigma \in S(i, n-i)} \chi_x(\sigma)\ell_{n+1-i}{\left(\ell_{i}{\left(x_{\sigma(1)}\wedge \cdots\wedge x_{\sigma(i)}\right)}\wedge x_{\sigma(i+1)}\wedge \cdots\wedge x_{\sigma(n)}\right)}=0.
    $\end{center}
    This provides an equivalent alternative definition for $\linf$-algebras, which justifies some other names for $\linf$-algebras (Lie-$\infty$ algebras or strongly homotopy Lie algebras).
\end{remark}

For the $\linf$-morphism $F\colon (\frg,Q) \to (\frg', Q')$, an similar analysis for the condition $F \circ Q = Q' \circ F$ shows that:
\begin{itemize}
    \item For $n = 1$, $F_1 \circ Q_1 = Q'_1 \circ F_1$, which implies that $F_1$ is a morphism between the differerntial graded vector spaces $\frg$ and $\frg'$;
    \item For $n = 2$, $F_1 \circ Q_2 - Q'_2 \circ F_1 = (\text{terms involving } F_2)$. This suggests the following result:
\end{itemize}
\begin{proposition}[From $\bm{\linf}$-Morphisms to DGLA Morphisms] {}
    An $\linf$-morphism $F\colon (\frg,Q) \to (\frg', Q')$ between the DGLAs $\frg$ and $\frg'$ defines a morphism of DGLAs if and only if $F_n = 0$ for $n \geq 2$.
\end{proposition}
\begin{remark}
    A morphism of $\linf$-algebras $F\colon (\frg,Q) \to (\frg',Q')$ is called an \textbf{$\bm{\linf}$-quasi-isomorphism}, if the first Taylor coefficient $F_1\colon \frg \to \frg'$ is a quasi-isomorphism, that is, $F_1$ induces isomorphisms of the cohomology spaces $\mathrm H^n(F_1)\colon \mathrm H^n(\frg) \to \mathrm H^n(\frg')$. In this case we say that $\frg$ is \textbf{homotopy equivalent} to $\frg'$. It is not immediate that homotopy equivalence defines an equivalence relation in the category of $\linf$-algebras. It is a non-trivial result that an ${\linf}$-quasi-isomorphism \emph{admits an $\linf$-quasi-inverse}, which will be shown in the next part.
\end{remark}

\subsection{Minimal Models}

\begin{definition}
    Let $(\frg,Q)$ and $(\frg',Q')$ be two $\linf$-algebras. We define the direct sum $\frg \oplus \frg'$ to be an $\linf$-algebra with the co-differential given in terms of the Taylor coefficients by
    \begin{center}$
        (Q \oplus Q')_n((x_1 + x'_1)\cdots (x_n + x'_n)) = Q_n(x_1 \cdots x_n) + Q'_n(x'_1 \cdots x'_n)
    $\end{center}
    where $x_1,...,x_n \in \frg$ and $x'_1,...,x'_n \in \frg'$.
\end{definition}

\begin{definition}
    Let $(\frg,Q)$ be an $\linf$-algebra. It is called
    \begin{itemize}
        \item \textbf{minimal}, if the first Taylor coefficient $Q_1 = 0$;
        \item \textbf{linear contractible}, if the Taylor coefficients $Q_n = 0$ for $n \geq 2$ and $\mathrm H^\bullet(\frg,Q_1) = 0$.
    \end{itemize}
\end{definition}

\begin{proposition}[Decomposition Theorem]{}
    Every $\linf$-algebra is isomorphic to a direct sum of a minimal $\linf$-algebra and a linear contractible $\linf$-algebra.
\end{proposition}
\begin{proof}
    (\emph{Adapted from \cite{arnal02} and \cite{jurvco19}.}) The first step of the proof is a general fact in linear algebra that any cochain complex of vector spaces is a direct sum of a complex with zero differential and a complex with zero cohomology. For this we consider a cochain complex $(C^\bullet,\d)$ of vector spaces. Note that the two short exact sequences
    \begin{center}
        \begin{tikzcd}
            0 & {\ker \d^{n}} & {C^n} & {\img  \d^n} & 0
            \arrow[from=1-1, to=1-2]
            \arrow[from=1-2, to=1-3]
            \arrow[from=1-3, to=1-4]
            \arrow[from=1-4, to=1-5]
        \end{tikzcd} \\
        \begin{tikzcd}
            0 & {\img \d^{n-1}} & {\ker \d^{n}} & {\mathrm H^n(C^\bullet)} & 0
            \arrow[from=1-1, to=1-2]
            \arrow[from=1-2, to=1-3]
            \arrow[from=1-3, to=1-4]
            \arrow[from=1-4, to=1-5]
        \end{tikzcd} 
    \end{center}
    split. Therefore we have a decomposition $C^n = Z^n \oplus Z^n_c = B^n \oplus H^n \oplus Z_c^n$, where $Z^n = \ker \d^n$, $B^n =  \img \d^{n-1} \cong Z_c^{n-1}$, and $H^n \cong \mathrm H^n(C^\bullet)$. 
    We define a linear map $h^n \colon C^n \to C^{n-1}$ by the composition:
    \begin{center}
        \begin{tikzcd}
            {C^n} & {B^n} & {Z^{n-1}_c} & {C^{n-1}}
            \arrow["\sim", from=1-2, to=1-3]
            \arrow[hook, from=1-3, to=1-4]
            \arrow[two heads, from=1-1, to=1-2]
        \end{tikzcd}.
    \end{center}
    $h^n$ is called the \emph{splitting map}. It follows that $B^n = \img(\d^{n-1}\circ h^n)$ and $Z_c^n = \img(h^{n+1}\circ \d^n)$. Therefore we have a decomposition of the identity map on $C^n$:
        \begin{center}
            $\id = p^n + \d^{n-1}\circ h^n + h^{n+1}\circ \d^n,$
        \end{center}
    where $p^n\colon C^n \to H^n$ is the projection map. This shows that $h$ is a chain homotopy between $\id$ and $p$. Therefore the cohomology of $H^\bullet \cong \mathrm H^\bullet(C^\bullet)$ induced by the projection $p$ is trivial. On the other hand, the projection $1-p^n\colon C^n \to B^n \oplus Z_c^n$ is chain-homotopic to the zero map. Hence the induced differential on $B^\bullet \oplus Z_c^\bullet$ is zero.

    Now let $(\frg,Q)$ be an $\linf$-algebra. According to the above construction we can decompose $(\frg,Q_1)$ into $\frg_m \oplus \frg_c$,  such that $\frg_m$ has zero differential: $M_1 = 0$; and $\frg_c = B \oplus Z_c$ has zero cohomology: $\mathrm H^\bullet(\frg_c,L_1) = 0$. Correspondingly, the reduced symmetric algebra is decomposed as
    \begin{center}
        $\rsym{\frg[1]} \cong \rsym{\frg_m[1]} \oplus \rsym{\frg_c[1]} \oplus \rsym{\frg_m[1]} \otimes \rsym{\frg_c[1]}$.
    \end{center}
    We would like to find an $\linf$-isomorphism $F\colon (\frg,Q) \to (\frg,\bar Q)$ such that $(\frg,\bar Q) = (\frg_m , M) \oplus (\frg_m, L)$. In other words, the codifferential $\bar Q$ satisfies
    \begin{center}$
        \bar Q\big|_{\rsym{\frg_m[1]}} = M, \quad \bar Q\big|_{\rsym{\frg_m[1]}} = L, \quad \bar Q\big|_{\rsym{\frg_m[1]} \otimes \rsym{\frg_c[1]}} = M \otimes \id + \id \otimes L.
    $\end{center}
    The process is equivalent to finding an infinite chain of $\linf$-isomorphisms: 
    \begin{center}
        \begin{tikzcd}
            {(\frg, Q)} && {(\frg, Q^{(1)})} & {(\frg, Q^{(2)})} & {(\frg, Q^{(3)})} & \cdots & {(\frg, \bar Q)}
            \arrow["{F^{(1)} = \id}", from=1-1, to=1-3]
            \arrow["{F^{(2)}}", from=1-3, to=1-4]
            \arrow["{F^{(3)}}", from=1-4, to=1-5]
            \arrow[from=1-5, to=1-6]
            \arrow[from=1-6, to=1-7]
        \end{tikzcd}
    \end{center}
    with the following properties ($*$): for each $k \in \Z_{>0}$,
    \begin{enumerate}[label = ($*$.\arabic*)]
        \item $F^{(k)}$ is an $\linf$-isomorphism such that $F^{(k)}_j = 0$ for all $j \ne k$;
        \item $Q^{(k-1)}_j = Q^{(k)}_j$ for all $j \ne k$;
        \item $Q^{(k)}_j(\rsym{\frg_m[1]}) \subseteq \frg_m$ for $1 \leq j \leq k$;
        \item $Q^{(k)}_j\big|_{\rsym{\frg_c[1]}} = 0$ for $2 \leq j \leq k$;
        \item $Q^{(k)}_j\big|_{\rsym{\frg_m[1]} \otimes \rsym{\frg_c[1]}} = 0$ for $1 \leq j \leq k$.
    \end{enumerate}
    Assuming these, we will obtain an $\linf$-isomorphism $F\colon (\frg,Q) \to (\frg, \bar Q)$ such that the Taylor coefficients of $F$ and $\bar Q$ are given by $F_k = F^{(k)}_k$ and $\bar Q_k = Q^{(k)}_k$ for all $k \in \Z_{>0}$ respectively. Moreover, $(\frg,\bar Q)$ is the claimed direct sum. The strategy is using induction on $k$ to construct $\bar Q_k$ and $F_k$ and to prove the properties ($*$) for $k$. The base case $k=1$ is proven in the first part of this proof. Now we assume that $Q^{(k-1)}$ is given and construct $Q^{(k)}_k$ and $F^{(k)}_k$. To ease the notations, we write $Q_1 = Q^{(k-1)}_1 = Q^{(k)}_1$, $Q_k = Q^{(k-1)}_k$, $\bar Q_k = Q^{(k)}_k$, and $F_k = F^{(k)}_k$. 

    The condition $F^{(k)}\circ Q^{(k-1)} = Q^{(k)} \circ F^{(k)}$ of $F^{(k)}$ being an $\linf$-morphism is expanded as follows (\emph{neglecting any signs}):
    \begin{equation}
        F_kQ_1(x_1\cdots x_k) + Q_k(x_1\cdots x_k) = Q_1F_k(x_1 \cdots x_k) + \bar Q_k(x_1 \cdots x_k). \label{equ:II_decomp_Fk}
    \end{equation}
    \newcommand{\Type}{\mathsf{Type}}
    For $x \in \operatorname{Sym}^k(\frg[1])$, $k \geq 2$, we say that $x$ is of \emph{type $r$} for $0 \leq r \leq k$, if $x = x_1\cdots x_k$, where $x_1,...,x_r \in Z_c$ and $x_{r+1},...,x_k \in Z$. We write $\Type(x) = r$. We define $\bar Q_k(x)$ to be the projection of $Q_k(x)$ onto $\frg_m \oplus Z_c$ if $\Type(x) = 0$, and $\bar Q_k(x) = 0$ if $\Type(x) > 0$. Then ($*$.4) and ($*$.5) holds immediately. $F_k(x)$ will be defined recursively for the type of $x$. 
    
    First, we consider $\Type(x) = 0$, i.e. $x = x_1\cdots x_k$ for $x_1,...,x_k \in Z$. An easy induction shows that $Q_k(x) \in Z$. Hence $\bar Q_k(x) \in \frg_m$ by definition, which proves ($*$.3). On the other hand, the equation \eqref{equ:II_decomp_Fk} becomes
    \begin{center}$
        Q_1F_k(x_1 \cdots x_k) =  \bar Q_k(x_1 \cdots x_k) - Q_k(x_1\cdots x_k).
    $\end{center}
    This defines $F_k(x)$ up to an element in $Z$. Furthermore, suppose that $x_1 = Q_1y_1 \in B$. The expansion of the equation $Q^{(k-1)} \circ Q^{(k-1)} = 0$ takes the form (\emph{neglecting any signs}):
        \begin{equation}
            Q_k(Q_1y_1\cdot x_2\cdots x_k) + \displaystyle\sum_{1 \leq i,j < k}Q_i(\cdots Q_j(\cdots)\cdots) + Q_1Q_k(y_1\cdot x_2 \cdots x_k) = 0. \label{equ:II_decomp_Q2}
        \end{equation}
    Since $Q_i,Q_j$ satisfies the property ($*$.4) and ($*$.5), all intermediate terms vanish. We are left with
    \begin{center}$
        Q_k(Q_1y_1\cdot x_2\cdots x_k) = -Q_1Q_k(y_1\cdot x_2 \cdots x_k) \in B,
    $\end{center}
    which implies that $\bar Q_k(Q_1y_1\cdot x_2\cdots x_k) = 0$. Let $\tilde Q_1\colon \operatorname{Sym}^k(\frg[1]) \to \operatorname{Sym}^k(\frg[1])$ be the coalgebra morphism induced by $Q_1$. Then $\bar Q_k \tilde Q_1(x) = 0$.

    Next, we define $F_k(x)$ recursively by the following construction. Let $r \geq 1$. Suppose that $F_k(y)$ is determined for all $y$ with $\Type(y) \leq r-2$ and for all $z$ with $\Type(z) = r-1$ \emph{up to an element in $Z$}. Let $x \in \operatorname{Sym}^k(\frg[1])$ with $\Type(x) = r$. We would like to determine $F_k(x)$ up to an element in $Z$, and specify $F_k(y)$ for all $y$ with $\Type(y) = r-1$.

    The following sub-lemma (\emph{cf.\ }Lemma V.3 of \cite{arnal02}) will be useful in the subsequent proof:
    \begin{quote}
        \emph{Let $x \in \operatorname{Sym}^k(\frg[1])$ with $\Type(x) = r > 0$. If $Q_1 x = 0$, then there exists $y$ of type $r+1$ such that $x = \tilde Q_1 y$.}
    \end{quote}
    Applying \eqref{equ:II_decomp_Fk} to $\tilde Q_1(x)$, we obtain that
        \begin{center}
            $Q_1F_k\tilde Q_1(x) = Q_k\tilde Q_1(x),$
        \end{center}
    where we used the facts that $Q_1\tilde Q_1 = 0$ and $\tilde Q_k \tilde Q_1 = 0$. Similar to \eqref{equ:II_decomp_Q2}, $Q^{(k-1)} \circ Q^{(k-1)} = 0$ implies that 
    \begin{equation}
        Q_1Q_k(x) + Q_k\tilde Q_1(x) = 0. \label{equ:II_decomp_Q2_2}
    \end{equation}
    Combining the two equations, we have $F_k\tilde Q_1(x) + Q_k(x) \in Z$. Note that $\Type(\tilde Q_1 x) = r-1$, so that $F_k\tilde Q_1(x)$ is determined up to an element in $Z$. Now we specify $F_k\tilde Q_1(x)$ such that 
    \begin{equation*}
        F_k\tilde Q_1(x) + Q_k(x) =: b(x) \in B.
    \end{equation*}
    Applying \eqref{equ:II_decomp_Fk} to $x$, we obtain that $b(x) = Q_1F_k(x)$. On the other hand, note that $b(x)$ must satisfy the constraint that $b(x) = Q_k(x)$ provided $Q_1(x) = 0$. In this case, the sub-lemma above implies that $x = \tilde Q_1y$ for some $y$ with $\Type(y) = r+1$. Hence \eqref{equ:II_decomp_Q2_2} implies that $b(x) = -Q_1Q_k(y)$. By fixing $b(x)$ satisfying the above constriants, $F_k(x)$ is determined up to an element in $Z$.
\end{proof}

\begin{corollary}[Minimal Model Theorem]{}
    Each $\linf$-algebra $\frg$ is homotopy equivalent to a minimal $\linf$-algebra $\frg_m$, which is known as the \textbf{minimal model} of $\frg$.
\end{corollary}
\begin{proof}
    Let $\frg = \frg_m \oplus \frg_c$ be the decomposition of $\frg$ into minimal $\frg_m$ and linear contractible $\frg_c$. Note that the linear contractible part $\frg_c$ has trivial cohomology. The inclusion $\frg_m \hookrightarrow \frg$ and the projection $\frg \twoheadrightarrow \frg_m$ are $\linf$-quasi-isomorphisms. Hence $\frg$ and $\frg_m$ are homotopy equivalent.
\end{proof}

\begin{lemma}[]{}
    Suppose that $F \colon (\frg,Q) \to (\frg',Q')$ is an $\linf$-quasi-isomorphism of minimal $\linf$-algebras $\frg$ and $\frg'$. Then $F$ is an isomorphism.
\end{lemma}
\begin{proof}
    The conditions that $\frg$ and $\frg'$ are minimal imply that $\mathrm H^\bullet(\frg) = \frg$ and $\mathrm H^\bullet(\frg') = \frg'$. Since $F$ is a quasi-isomorphism, $F_1\colon \frg \to \frg'$ is an isomorphism. Let $G_1$ be the inverse of $F_1$. By the universal property of $\rsym{\frg[1]}$, $G_1$ uniquely lifts to an $\linf$-morphism $G\colon (\frg',Q') \to (\frg, Q)$, which is the inverse of $F$. 
\end{proof}
\begin{remark}
    This demonstrates that the minimal model of an $\linf$-algebra is unique up to isomorphism. 
\end{remark}

\begin{proposition}[Homotopy Equivalence of $\bm{\linf}$-Algebras]{}
    Every $\linf$-quasi-isomorphism $F\colon (\frg,Q) \to (\frg',Q')$ has a quasi-inverse $G\colon (\frg',Q') \to (\frg,Q)$, that is, $\mathrm H^\bullet(G_1)$ is an actual inverse of $\mathrm H^\bullet(F_1)$. Therefore the homotopy equivalence of $\linf$-algebras is an equivalence relation.
\end{proposition}
\begin{proof}
    We decompose $(\frg,Q)$ and $(\frg',Q')$ into the direct sum of the minimal part and the contractible part:
        \begin{center}
            $(\frg,Q) = (\frg_m,Q_m) \oplus (\frg_c,Q_c), \qquad (\frg',Q') = (\frg'_m,Q'_m) \oplus (\frg'_c,Q'_c).$
        \end{center}
    We have a commutative diagram shown as follows, where $\tilde F = \pi'\circ F\circ \iota\colon (\frg_m,Q_m) \to (\frg'_m,Q'_m)$ is an $\linf$-quasi-isomorphism induced by $F$.
    \begin{equation}
        \begin{tikzcd}
            {(\frg,Q)} && {(\frg',Q')} \\
            {(\frg_m,Q_m)} && {(\frg'_m,Q'_m)}
            \arrow["F", from=1-1, to=1-3]
            \arrow["\pi"', shift right=1, from=1-1, to=2-1]
            \arrow["{\tilde F}", from=2-1, to=2-3]
            \arrow["{\pi'}"', shift right=1, from=1-3, to=2-3]
            \arrow["\iota"', shift right=1, from=2-1, to=1-1]
            \arrow["{\iota'}"', shift right=1, from=2-3, to=1-3]
        \end{tikzcd} \label{equ:II_minimal_cd}
    \end{equation}
    Since $\frg_m$ and $\frg'_m$ are minimal, $\tilde F$ is an isomorphism. We denote by $\tilde G$ the inverse of $\tilde F$. Then $G = \iota \circ \tilde G \circ \pi'$ is a quasi-inverse of $F$.
\end{proof}

\subsection{Maurer--Cartan Equation, Reprise}\label{sec:L_infty_MC}

In this part, we extend the discussions in Section \ref{sec:DGLA_MC} to incorporate $\linf$-morphisms into the deformation theory via differential graded Lie algebras. 

\begin{definition}
    Let $(\frg,Q)$ be an $\linf$-algebra. For $x \in \frg^1$, the \textbf{homotopy Maurer--Cartan equation} is given by
    \begin{equation}
        Q(\e^{x}-1) = 0 \iff\displaystyle\sum_{n=1}^\infty \dfrac{1}{n!}Q_n(x\cdots x) = 0. 
        \label{equ:II_HMC}
    \end{equation}
    The set of solutions is again denoted by $\MC(\frg)$. Note that when $\frg$ is a DGLA, \eqref{equ:II_HMC} reduces to the Maurer--Cartan equation \eqref{equ:II_MC}.
\end{definition}
It is more natural to consider $\e^x-1 \in \rsym{\frg[1]}$ instead of $x \in \frg^1$ as the invariant object in an $\linf$-algebra $(\frg,Q)$. If $F\colon (g,Q) \to (g',Q')$ is an $\linf$-morphism, then $F(Q(\e^x-1)) = Q'(F(\e^x-1))$. It follows that if $x \in \MC(\frg)$, then under the $\linf$-morphism $F$,
\begin{center}$
    x \longmapsto \displaystyle\sum_{n=1}^\infty F_n(x \cdots x) \in \MC(\frg').
$\end{center}
\begin{remark}
    Note that $\e^x-1$ is not a well-defined element in $\rsym{\frg[1]}$ without specifying the notion of convergence. We resolve this issue by only considering the $\linf$-algebra $\frg \otimes_k \mm_R$ where $R$ is either a test $k$-algebra (so that we have the nilpotency condition) or a formal completion (so that $\e^x-1$ is a well-defined formal power series).
\end{remark}

As an analogue of the gauge action on a DGLA, we can define the homotopy action on an $\linf$-algebra $\frg$ by constructing a simplicial structure on $\MC(\frg)$. In the case where $\frg$ is a DGLA, the fundamental group of $\MC(\frg)$ corresponds to $\MC(\frg)$ modulo gauge action. The notion of Maurer--Cartan functor \eqref{equ:II_MC_functor} can thus be extended to $\linf$-algebras. This is beyond the scope of our exposition, as we restrict our attention to the deformation problem controlled by DGLAs. For a detailed treatment, see \cite{schlessinger12} and \cite{getzler09}.

\begin{lemma}[]{}
    Let $\frg$ and $\frg'$ be two DGLAs. Then the direct sum of DGLAs induces the Cartesian product of Maurer--Cartan functors:
    \begin{center}
        $\MCbar_{\frg \oplus \frg'} \simeq \MCbar_{\frg} \times \MCbar_{\frg'}$.
    \end{center}
\end{lemma}
\begin{proof}
    This follows from that $\MC((\frg \oplus \frg') \otimes \mm_R) = \MC(\frg  \otimes \mm_R) \times \MC(\frg' \otimes \mm_R)$ and that the gauge action is preserved in the decomposition.
\end{proof}

\begin{lemma}[]{}
    Suppose that $\frg$ is a linear contractible $\linf$-algebra. Then $\MCbar_{\frg}(\mm_R)$ is a singleton for any test algebra $R$.
\end{lemma}
\begin{proof}
    A linear contractible $\linf$-algebra is simply a differential graded vector space. The Maurer--Cartan equation is $\d x = 0$ for $x \in \frg^1 \otimes \mm_R$. Write $x = \displaystyle\sum_i x_i \otimes a_i$ for $x_i \in \frg^1$ and $a_i \in \mm_R$. Then $\d x_i = 0$ for each $i$. Since $\frg$ has trivial cohomology, $x_i = \d y_i$ for some $y_i \in \frg^0$. Then $x = \d{\left(\sum_i y_i \otimes a_i\right)} = \exp(\sum_i y_i \otimes a_i)\cdot 0$. Hence $x$ is gauge equivalent to $0$. We deduce that $\MCbar_{\frg}(\mm_R) = \{0\}$.
\end{proof}

\begin{theorem}[$\bm{\linf}$-Quasi-Isomorphism Theorem]{thm:II_Linf_quasi}
    Let $F\colon \frg \to \frg'$ be an $\linf$-morphism of DGLAs. If $F$ is an $\linf$-quasi-isomorphism, then the map 
    \begin{center}$
        x \longmapsto \displaystyle\sum_{n=1}^\infty \dfrac{1}{n!}F_n(x \cdots x)
    $\end{center}
    induces an natural isomorphism of the Maurer--Cartan functors:
    \begin{center}
        $\MCbar_{\frg} \simeq \MCbar_{\frg'}.$
    \end{center}
\end{theorem}
\begin{proof}
    Let $\frg = \frg_m \oplus \frg_c$ be a decomposition of $\frg$ such that $\frg_m$ is minimal and $\frg_c$ is contractible. 
    By the previous two lemmata, we have 
    \begin{center}
            $\MCbar_{\frg} = \MCbar_{\frg_m \oplus \frg_c} \simeq \MCbar_{\frg_m} \times \MCbar_{\frg_c} \simeq \MCbar_{\frg_m}.$
    \end{center} 
    Moreover, any $\linf$-morphism induces a natural transformation in the corresponding Maurer--Cartan functors. The diagram \eqref{equ:II_minimal_cd} induces the commutative diagram of Maurer--Cartan functors:
    \begin{center}
            \begin{tikzcd}
                {\MCbar_{\frg}} && {\MCbar_{\frg'}} \\
                {\MCbar_{\frg_m}} && {\MCbar_{\frg'_m}}
                \arrow["\MCbar_{F}", from=1-1, to=1-3]
                \arrow[shift right=1, from=1-1, to=2-1]
                \arrow["{\MCbar_{\bar F}}", from=2-1, to=2-3]
                \arrow[shift right=1, from=1-3, to=2-3]
                \arrow[shift right=1, from=2-1, to=1-1]
                \arrow[shift right=1, from=2-3, to=1-3]
            \end{tikzcd}
    \end{center}
    Since $F\colon \frg \to \frg'$ is an $\linf$-quasi-isomorphism, the induced map $\bar F\colon \frg_m \to \frg'_m$ is also an $\linf$-quasi-isomorphism. Since $\frg_m$ and $\frg'_m$ are minimal, the $\bar F$ is a strict isomorphism which admits a strict inverse. Therefore $\MCbar_{\bar F}$ is a natural isomorphism of functors. We conclude that $\MCbar_\frg \simeq \MCbar_{\frg'}$ as claimed.
\end{proof}

\newcommand{\linfalg}{\mathsf{L_{\infty}Alg}}
\newcommand{\DGLA}{\mathsf{dgLie}}
\begin{remark}
    Let us digress on the result from a categorical perspective. Let $\linfalg$ be the category of $\linf$-algebras and $\DGLA$ be that of differential graded Lie algebras. The reason why we need to upgrade from Theorem \ref{thm:II_DGLA_quasi} to Theorem \ref{thm:II_Linf_quasi} is that there exist $\linf$-morphisms between DGLAs which are not morphisms of DGLAs, and quasi-isomorphisms of DGLAs may not have inverse in $\DGLA$. The theorem tells us that, from the viewpoint of the deformation functor, we should instead work on the homotopy category $\mathsf{Ho}(\DGLA)$, which is obtained by localisation on the quasi-isomorphisms in $\DGLA$. 
\end{remark}
 
\subsection{Formality Theorem} 

\begin{definition}
    A differential graded Lie algebra $\frg$ is said to be \textbf{formal}, if it is homotopy equivalent to its cohomology $\mathrm H^\bullet(\frg)$, viewed as a DGLA with induced bracket and zero differential.
\end{definition}

The main theorem due to Kontsevich \cite{kon97} in this dissertation is the following:

\begin{theorem}[Formality Theorem]{thm:II_formality}
    There exists an $\linf$-quasi-isomorphism $U\colon \tpoly \to \dpoly$ such that the first Taylor coefficient $U_1$ coincides with the HKR map $U_1^{(0)}$.
\end{theorem}

\begin{remark}
    By HKR theorem, $\tpoly$ is quasi-isomorphic to the cohomology $\mathrm H^\bullet(\dpoly)$ of $\dpoly$. Therefore the formality theorem implies that the DGLA $\dpoly$ is formal. This justifies the name of the theorem.
\end{remark}
 
The formality theorem, in combination of Proposition \ref{prop:II_star_MC}, Proposition \ref{prop:II_pois_MC}, and the $\linf$-quasi-isomorphism theorem \ref{thm:II_Linf_quasi}, completely solved the classification problem of deformation quantisation on a Poisson manifold. We obtain the bijective correspondence announced in Section \ref{sec:star}:
\begin{equation}
    \dfrac{\{\text{Formal Poisson structures}\}}{\text{Equivalence}} \quad\longleftrightarrow\quad \dfrac{\{\text{Star products}\}}{\text{Equivalence}}. \tag{\ref{equ:I_defo_quan_bij}}
\end{equation}
\newcommand{\QD}{Q_{\mathscr{D}}}
\newcommand{\QT}{Q_{\mathscr{T}}}
The rest of this dissertation will be devoted to the proof of formality theorem. We would like to study the constraints imposed on the Taylor coefficients $U_n$ given that $U$ is an $\linf$-morphism. Let $\QT$ and $\QD$ be the co-differentials associated to the $\linf$-algebras $\tpoly$ and $\dpoly$. Since they are DGLAs, then $Q^i_j = 0$ for $i > j - 1$, and we have
\begin{equation}
    (\QD)_1 \circ U_n + (\QD)_2 \circ U^2_n= U_n \circ (\QT)^n_n + U_{n-1} \circ (\QT)^{n-1}_n, \qquad \forall\, n \in \Z_{>0}, \label{equ:II_formality_constr}
\end{equation}
where we used the same notation as \eqref{equ:Q_sqzero}. To expand this expression further, we have
\begin{equation}\hspace*{-1em}
    \begin{aligned}
        &(\QD)_1 (U_n(\xi_1 \cdots \xi_n)) 
        + \dfrac{1}{2}\displaystyle\sum_{\sigma \in S(i,n-i)}\varepsilon_\xi(\sigma)
        (\QD)_2{\left( U_i(\xi_{\sigma(1)} \cdots \xi_{\sigma(i)}) \cdot U_{n-i}(\xi_{\sigma(i+1)} \cdots \xi_{\sigma(n)}) \right)} \\
        &= \displaystyle\sum_{i=1}^n \varepsilon_\xi^{\{i\}}U_n{\left( (\QT)_1(\xi_i)\cdot \xi_1 \cdots \hat \xi_i \cdots \xi_n \right)} + \dfrac{1}{2}\displaystyle\sum_{i=1}^k \displaystyle\sum_{j \ne i} \varepsilon_\xi^{\{i,j\}}U_{n-1}{\left( (\QT)_2(\xi_i \cdot \xi_j)\cdot \xi_1 \cdots \hat \xi_i \cdots \hat \xi_j \cdots \xi_n \right)},
    \end{aligned}
    \label{equ:II_QD_QT_constr}
\end{equation}
where $\varepsilon_\xi(\sigma)$ is the Koszul sign defined in Example \ref{eg:II_rsym}, $\varepsilon_\xi^{\left\{ i \right\}}$ and $\varepsilon_\xi^{\left\{ i,j \right\}}$ are the Koszul signs associated to $\sigma = (1\, i)$ and $\sigma = (1\, \operatorname{min}\{i,j\})(2\, \operatorname{max}\{i,j\})$ respectively. From the definition of $\dpoly$ and $\tpoly$ and the d\'ecalage isomorphisms \eqref{equ:II_dec_iso}, we have 
\begin{align*}
    &(\QT)_1(\xi) = 0, &
    &(\QT)_2(\xi \cdot \eta) = (-1)^{|\xi|\cdot(|\eta|-1)}\SN{\xi,\eta}, \\
    &(\QD)_1(f) = (-1)^{|f|}\dhochm f = (-1)^{|f|}\Ger{\mu,f}, &
    &(\QD)_2(f \cdot g) = (-1)^{|f|\cdot(|g|-1)}\Ger{f,g}.
\end{align*}
We extend the $\linf$-morphism $U$ to include the usual multiplication $U_0 \in \R \to \dpoly$ such that $U_0(1) = \mu \in \dpoly[1]$. Let the both sides of \eqref{equ:II_QD_QT_constr} act on the smooth functions $f_0,...,f_m \in \smooth{M}$. Furthermore, we make the convention that
\begin{center}$
    U_n(\xi_1\cdots \xi_n)(f_1,...,f_m) = 0
$\end{center}
if $m \ne \displaystyle\sum_{i=1}^n |x_i| + 1 - n$. With some careful computation, it can be shown (\emph{cf.\ }Theorem VI.1 of \cite{arnal02}) that \eqref{equ:II_QD_QT_constr} simplifies to
\begin{equation}
    \begin{aligned}
        &\displaystyle\sum_{i=1}^n \displaystyle\sum_{j\ne i} \varepsilon^{\left\{ i,j \right\}}_\xi U_{n-1}(\xi_i \diamond  \xi_j\cdots\hat \xi_i \cdots \hat \xi_j \cdots \xi_n)(f_0,...,f_m)
        = \displaystyle\sum_{i=0}^n\displaystyle\sum_{j=-1}^m \displaystyle\sum_{k=0}^{m-k} (-1)^{j(k+m)} \\
        & \quad \displaystyle\sum_{\sigma \in S(i,n-i)} \varepsilon_\xi(\sigma)U_i(\xi_{\sigma(1)}\cdots \xi_{\sigma(i)})\left( f_0,...,f_{k-1},U_{n-i}(\xi_{\sigma(i+1)}\cdots \xi_{\sigma(n)})(f_k,...,f_{j+k}),f_{j+k+1},...,f_m \right)
    \end{aligned} \label{equ:II_graph_constr}
\end{equation}
where $\xi_i \diamond \xi_j$ is the product on $\tpoly$ (\emph{cf.\ }\eqref{equ:II_polyvec_prod} or \eqref{equ:II_polyvec_prod_dec}). \eqref{equ:II_graph_constr} will be referred as the \textbf{formality equation}, which is our starting point of the next chapter, where we construct $U_n$ as a sum of weighted admissible graphs.

\chapter{Kontsevich Quantisation}

In this chapter we present the construction of the $\linf$-morphism from $\tpolyrd$ to $\dpolyrd$ and that of a star product in $\R^d$ following Kontsevich \cite{kon97}. Then we sketch the globalisation of the star product in a general smooth manifold $M$ following \cite{cattaneo02global}.

\section[Construction in \texorpdfstring{$\mathbb{R}^{{d}}$}{R{\textasciicircum}d}]{Construction in $\mathbb{R}^{\bm{d}}$} \label{sec:Kont_con}

\newcommand{\Confp}{\operatorname{Conf}^+}

This section is devoted to construct an explicit $\linf$-quasi-isomorphism $U\colon \tpolyrd \to \dpolyrd$. 

\subsection{Admissible Graphs}\label{sec:ad_graph}

First we define the admissible graphs. A directed graph $\Gamma$ is a pair $(V_{\Gamma}, E_{\Gamma})$, where $V_{\Gamma}$ is the set of the vertices of $\Gamma$, and $E_{\Gamma} \subseteq V_{\Gamma} \times V_{\Gamma}$ is the set of edges of $\Gamma$. For each edge $e = (v_1,v_2) \in E_{\Gamma}$, $s(e) := v_1$ is called the source of $e$, and $t(e) := v_2$ is called the target of $e$.

\begin{definition}
    An \textbf{admissible} graph $\Gamma \in G_{n,m}$ is a connected directed graph, satisfying the following conditions:
    \begin{itemize}[nosep]
        \item $n,m \geq 0$ and $2n+m -2 \geq 0$.
        \item $V_\Gamma = V^1_\Gamma \sqcup V^2_\Gamma$, where $V^1_\Gamma = \left\{ 1,...,n \right\}$ is the set of vertices of the first type, and $V^2_\Gamma = \left\{ \bar 1,...,\bar m \right\}$ is the set of vertices of the second type. 
        \item $\Gamma$ has no edge starting from a vertex of the second type.
        \item $\Gamma$ has neither loops nor double arrows.
        \item For $v \in V^1_\Gamma$, the set of edges starting from $v$ is denoted by
        \begin{center}$
            \operatorname{Star}(v) := \left\{ e \in E_\Gamma\colon e = (v,w) \text{ for some }  w\in V_\Gamma  \right\}
        $\end{center}
        $\operatorname{Star}(v)$ is a finite set. We put an order on it so that $\operatorname{Star}(v) = \left\{ e_v^1,...,e_v^{\#(\operatorname{Star}(v))} \right\}$.
    \end{itemize}
\end{definition}

We associate each $\Gamma \in G_{n,m}$ and polyvector vectors $\xi_1,...,\xi_n$ with a polydifferential operator $B_\Gamma(\xi_1\otimes\cdots\otimes\xi_n) \in \dpolyrd[m]$ by the following recipe.
\begin{itemize}[nosep]
    \item $B_\Gamma = 0$ unless $\xi_i \in \tpolyrd[k_i]$ for each $i\in V^1_\Gamma$, where $k_i = \#(\operatorname{Star}(i))-1$.
    \item Let $I\colon E_\Gamma \to \left\{ 1,...,d \right\}$ be a map.
    \item For each $i \in V^1_\Gamma$ and $\bar j \in V^2_\Gamma$, put respectively the functions:
    \begin{equation*}
        \left(\displaystyle\prod_{e \in \Gamma(-,i)}\pa_{I(e)}\right) \xi_i^{I(e_i^1),...,I(e_i^{k_i})}, \qquad \left(\displaystyle\prod_{e \in \Gamma(-,\bar j)}\pa_{I(e)}\right) f_{j}
    \end{equation*}
    where $\Gamma(-,v) := \left\{ e \in E_{\Gamma}\colon t(e) = v \right\}$.

    \item Multiply the functions on the vertices and sum over all possible $I\colon E_\Gamma \to \left\{ 1,...,d \right\}$. This gives
    \begin{center}$
        B_\Gamma(\xi_1\otimes\cdots\otimes\xi_n)(f_1,...,f_m) \in \R.
    $\end{center}
    So we obtain the map $B_\Gamma\colon \tpolyrd^{\otimes n} \to \dpolyrd[m-1]$.   
\end{itemize}

\subsection{Configuration Spaces}

\begin{definition}
    We construct the \textbf{configuration space} for the $n+m$ vertices of $\Gamma \in G_{n,m}$. 

\begin{itemize}
    \item If $m > 0$, let
    \begin{equation}
        \begin{aligned}
            \Confp_{n,m} := \{ (p_1,...,p_n;q_1,...,q_m) \in \C^{n+m}\colon &p_1,...,p_n \in \mathbb H,\ p_i \ne p_j \text{ for } i \ne j; \\
        &q_1,...,q_m \in \R,\ q_1 < \cdots < q_m \}
        \end{aligned}
    \end{equation}
    $\Confp_{n,m}$ is a $(2n+m)$-dimensional smooth manifold. Let $G$ be the 2-dimensional Lie group acting on $\Confp_{n,m}$ by $z \mapsto az + b$ for $a>0$ and $b \in \R$. This is a free action, and we obtain a quotient manifold
    \begin{equation}
        C^+_{n,m} := \Confp_{n,m}/G
    \end{equation}
    which is a connected $(2n+m-2)$-dimensional smooth manifold. It has the natural orientation induced by the volume form in $\R^{2n+m}$:
    \begin{center}$
        \Omega = \d x_1 \wedge \d y_1 \wedge \cdots \wedge \d x_n \wedge \d y_n \wedge \d q_1 \wedge \cdots \wedge \d q_m, \qquad p_i = x_i + \i y_i,
    $\end{center}
    which is inherited from $\Confp_{n,m}$ because the actions of $G$ preserve orientation. 

    \item If $m = 0$, let
    \begin{equation}
        \operatorname{Conf}_{n} := \left\{ (p_1,...,p_n) \in \C^{n}\colon p_1,...,p_n \in \mathbb C,\ p_i \ne p_j \text{ for } i \ne j\right\}
    \end{equation}
    Let $G'$ be the 3-dimensional Lie group acting on $\operatorname{Conf}_{n}$ by $z \mapsto az + b$ for $a>0$ and $b \in \C$. The quotient manifold
        \begin{equation}
            C_{n} := \operatorname{Conf}_{n}/G'
        \end{equation}
    is a $(2n-3)$-dimensional smooth manifold for $n \geq 2$.
\end{itemize}
\end{definition}

Then we define the weight integral $W_\Gamma$ associated to $\Gamma \in G_n$. 
\begin{definition}
    Let $\mathbb H := \left\{ z \in \C\colon \operatorname{Im} z > 0 \right\}$ be the upper half plane endowed with the hyperbolic metric $g = \dfrac{\d x^2 + \d y^2}{y^2}$. It is called the \textbf{Poincar\'e half plane}. Let $\overline{\mathbb H} := \mathbb H \cup \R$.
    \begin{itemize}[nosep]
    \item For $p \in \mathbb H$, let $L(p,\infty)$ be the vertical half line from $p$ to infinity;
    \item For $p \ne q \in \mathbb H$, let $L(p,q)$ be the geodesic from $p$ to $q$, which is an arc of a circle centred on the real line;
    \item Let $\varphi(p,q)$ be the angle from $L(p,\infty)$ to $L(p,q)$. By simple planimetry (\emph{cf.\ }\cite{keller03}) we find that
    \begin{equation}
        \varphi(p,q) = \arg{\left(\dfrac{q-p}{q - \bar p}\right)}=  \dfrac{1}{2\i}\log\left( \dfrac{q-p}{q-\bar p}\cdot\dfrac{\bar q - p}{\bar q - \bar p} \right). \label{equ:III_angle_geom}
    \end{equation}
    We extend $\varphi(p,q)$ to $p \ne q \in \bar{\mathbb H}$.
    \end{itemize}
\end{definition}


\begin{definition}
    For $e \in E_{\Gamma}$, we define the \textbf{angle map} $\varphi_e\colon C_{n,m} \to S^1$ as 
    \begin{center}
        $\varphi_e(p_1,...,p_n;q_1,...,q_m) = \varphi(p_{s(e)},p_{t(e)})$.
    \end{center}
    We define a differential form
    \begin{equation}
        \omega_\Gamma := \bigwedge_{e \in E_\Gamma} \d\varphi_e \label{equ:III_weight_form}
    \end{equation}
    where the ordering of the 1-forms in the product is the one induced on the set of all edges by the
    ordering on the source vertices $v$ and the ordering on the set $\operatorname{Star}(v)$.
\end{definition}

\begin{proposition}[Compactification of Configuration Spaces]{}
    There exists a compact smooth manifold with corners $\bar{C}^+_{n,m}$ (\emph{resp.\ }$\bar C_n$) whose interior is the open configuration space $C^+_{n,m}$ (\emph{resp.\ }$C_n$) such that the angle map $\varphi$ and the differential form $\omega_\Gamma$ extends smoothly to the corresponding compactification.
\end{proposition}
\begin{remark}
    A manifold with corners $X$ is a second-countable Hausdorff topological space such that every point has a neighbourhood homeomorphic to $\R^{n-k} \times \R^k_{\geq 0}$ for some $k$, which is called the depth of that point. The boundary $\pa X$ is the set of those points with depth $k > 0$. The subset of points with depth $k$ is an $(n-k)$-submanifold $S^k(X)$ of $X$, the submanifolds of which are called the \textbf{strata} of codimension $k$. This notion is a natural generalisation of a manifold with boundary. The definition can be extended straightforward to smooth manifolds with corners, where the transition maps between charts are required to preserve the corners. 
\end{remark}
\begin{proof}
    We construct a map
    \begin{center}$
        \psi\colon \Confp_{n,m} \to {\mathbb T}^{(2n+m)(2n+m-1)} \times (\RP^2)^{(2n+m)(2n+m-1)(2n+m-2)}
    $\end{center}
    where ${\mathbb T}^{(2n+m)(2n+m-1)} = \underbrace{S^1 \times \cdots \times S^1}_{(2n+m)(2n+m-1)}$ is a torus. For each pair of distinct points $a,b \in \left\{ p_1,\bar p_1,...,p_n,\bar p_n,q_1,...,q_m \right\}$, we associate an angle $\operatorname{arg}(a-b) \in S^1$; for each triple of distinct points $a,b,c \in \left\{ p_1,\bar p_1,...,p_n,\bar p_n,q_1,...,q_m \right\}$, we associative a point $[a-b:b-c:c-a] \in \RP^2$. These data uniquely determine the map $\psi$. Since these data are preserved by the action of $G$, $\psi$ descends to the quotient:
    \begin{center}$
        \tilde\psi\colon C^+_{n,m} \to X := {\mathbb T}^{(2n+m)(2n+m-1)} \times (\RP^2)^{(2n+m)(2n+m-1)(2n+m-2)},
    $\end{center}
    which is injective. Therefore $C^+_{n,m}$ is embedded into $X$. Let $\bar C^+_{n,m}$ be the closure of $\tilde\psi(C^+_{n,m})$ in this space. This is a manifold with corners.
\end{proof}

\begin{definition}
    The \textbf{weight integral} for $\Gamma \in G_{n,m}$ is defined to be
    \begin{equation}
        W_{\Gamma} := \begin{cases}
            \displaystyle\prod_{k=1}^n \dfrac{1}{\#(\operatorname{Star}(k))!} \dfrac{1}{(2\pi)^{\#(E_\Gamma)}} \displaystyle\int_{\bar C^+_{n,m}} \omega_\Gamma, 
            & \text{if }  \#(E_\Gamma) = 2n+m-2, \\
            0, & \text{otherwise}.
        \end{cases} \label{equ:III_weight_int}
    \end{equation}
\end{definition}

\begin{remark}
    The requirement $\#(E_\Gamma) = 2n+m-2$ ensures that the degree of the form $\omega_\Gamma$ matches the dimension of the configuration space. It also ensures that $U_n$ has the correct degree (see the beginning of Section \ref{sec:proof_formality}).
\end{remark}

Now we define the $n$-th Taylor coefficient $U_n$ by
\begin{equation}
    U_n = \displaystyle\sum_{m=0}^\infty \displaystyle\sum_{\Gamma \in G_{n,m}} W_\Gamma B_\Gamma. \label{equ:III_U_coeff}
\end{equation}
Now we are just one step away from the formality theorem in $\R^d$. The proof will be completed in Section \ref{sec:proof_formality}.

\subsection{Explicit Formula of Kontsevich Star Product}

From Theorem \ref*{thm:II_Linf_quasi} and Equation \eqref{equ:III_U_coeff}, a Poisson bivector field $\pi$ induces the Kontsevich star product in the following way:

\begin{theorem}[Kontsevich Star Product]{thm:III_kont_star}
    Let $\pi \in \tpolyrd[1]$ be a Poisson bivector field on $\R^d$. Then 
    \begin{center}$
        \star = \mu + \displaystyle\sum_{n=1}^\infty \dfrac{\hbar^n}{n!} U_n(\pi \cdots \pi) \in \dpolyrd[1]\formpow{\hbar}
    $\end{center}
    defines a star product on $\R^d$ such that $f \star g - g \star f = \hbar\pi(f,g)$ modulo $\hbar^2$.
\end{theorem}
From the previous discussion, we can develop a set of rules for the computation of the star product $f \star g$. Some constructions can be simplified. For example, we only use the admissible graphs in $G_{n,2}$ where $\#(\operatorname{Star}(k)) = 2$ for $k=1,...,n$. The Lie group action of $G$ on configuration space $C^+_{n,2}$ fix the two vertices of the second type to $0,1 \in \R$ respectively, which gives an isomorphism:
\begin{center}$
    C^+_{n,2} \simeq H_n := \left\{ (p_1,...,p_n) \in \C^n\colon p_1,...,p_n \in \mathbb H,\ p_i \ne p_j \text{ for } i \ne j\right\}.
$\end{center}
The explicit formula for the Kontsevich star product is given by
\begin{equation}
    f \star g = fg + \displaystyle\sum_{n=1}^\infty \dfrac{\hbar^n}{2^n(2\pi)^{2n}n!}\displaystyle\sum_{\Gamma \in G_{n,2}} B_{\Gamma,\pi}(f,g) \displaystyle\int_{\bar H_n} \omega_{\Gamma}, \label{equ:III_kont_star}
\end{equation}
where $B_{\Gamma,\pi}(f,g) := B_\Gamma(\pi\cdots\pi)(f,g)$ according to the definition in Section \ref{sec:ad_graph} and $\omega_\Gamma$ is the differential form \eqref{equ:III_weight_form}.


\begin{example}\label{eg:III_Moyal}
    As an easy example, we compute the star product when $\pi = \pi^{ij}\pa_i \wedge \pa_j$ has constant coefficients. Since any derivative on $\pi^{ij}$ vanishes, we only need to sum over the admissible graphs with no edges targeting at the vertices of first type $V^1_\Gamma$. The admissible graphs with non-vanishing weight $W_\Gamma$ has $2n$ edges. Therefore the only contributing admissible graph in $G_{n,2}$ is of the following form, which is denoted by $\Gamma_n$.
    \begin{center}
        \vspace*{-0.2em}
        \begin{tikzpicture}[scale=.65]
            \draw [->] (-4, 0) -- (4,0);
            \draw [line width=1pt, dotted] (0,3) -- (1,3);
            \node at (4,0.3) {Re};

            \draw (-1,0) node[circle,fill,inner sep=0,minimum size=4pt]{} node[yshift=-10pt]{$f$};
            \draw (1,0) node[circle,fill,inner sep=0,minimum size=4pt]{} node[yshift=-10pt]{$g$};
            \draw (-2,3) node[circle,fill,inner sep=0,minimum size=4pt]{} node[yshift=10pt]{$1$};
            \draw (-1,3) node[circle,fill,inner sep=0,minimum size=4pt]{} node[yshift=10pt]{$2$};
            \draw (2,3) node[circle,fill,inner sep=0,minimum size=4pt]{} node[yshift=10pt]{$n$};

            \draw[line width=1pt, postaction={decorate, decoration={ 
                markings, mark=at position 0.5 with {\arrow[line width=1.2pt]{>}}}}]
                (-2,3) -- (-1,0);
            \draw[line width=1pt, postaction={decorate, decoration={ 
                markings, mark=at position 0.5 with {\arrow[line width=1.2pt]{>}}}}]
                (-2,3) -- (1,0);
            \draw[line width=1pt, postaction={decorate, decoration={ 
            markings, mark=at position 0.5 with {\arrow[line width=1.2pt]{>}}}}]
                (-1,3) -- (-1,0);
            \draw[line width=1pt, postaction={decorate, decoration={ 
            markings, mark=at position 0.5 with {\arrow[line width=1.2pt]{>}}}}]
                (-1,3) -- (1,0);
            \draw[line width=1pt, postaction={decorate, decoration={ 
            markings, mark=at position 0.5 with {\arrow[line width=1.2pt]{>}}}}]
                (2,3) -- (-1,0);
            \draw[line width=1pt, postaction={decorate, decoration={ 
            markings, mark=at position 0.5 with {\arrow[line width=1.2pt]{>}}}}]
                (2,3) -- (1,0);
        \end{tikzpicture} \vspace*{-1em}
    \end{center}
    The differential operator associated to $\Gamma_n$ is given by
    \begin{center}$
        B_{\Gamma_n,\pi}(f,g)(x) = \left.\displaystyle\prod_{k=1}^{n}\left( \displaystyle\sum_{i_k,j_k=1}^{2n}\pi^{i_k,j_k}\dfrac{\pa}{\pa x^{i_k}}\dfrac{\pa}{\pa y^{j_k}} \right)f(x)g(y)\right|_{x=y}.
    $\end{center}
    The differential form $\omega_{\Gamma_n}$ is given by
    \begin{center}$
        \omega_{\Gamma_{n}} = \d\varphi_{(1,\bar 1)} \wedge \d\varphi_{(1,\bar 2)} \wedge \cdots \wedge \d\varphi_{(n,\bar 1)} \wedge \d\varphi_{(n,\bar 2)}
    $\end{center}
    where, following \eqref{equ:III_angle_geom}, the differential 1-forms are given by $
    \d\varphi_{(i,\bar j)}(p_1,...,p_n,q_1,q_2) = 2\d\arg(q_j-p_i)$, where $q_1 = 0$ and $q_2 = 1$ are fixed. The compactification $\bar H_n$ is parametrised as
    \begin{center}$
        \begin{aligned}
            \bar H_n &\simeq \overline{\left\{ (\arg(0-p_1),\arg(1-p_1),...,\arg(0-p_n),\arg(1-p_n))\colon p_1,...,p_n \in \mathbb H \text{ distinct} \right\}} \\
        &\simeq [-\pi,0]^{2n}.
        \end{aligned}
    $\end{center}
    Hence the star product is given by
    \begin{align*}
        (f \star g)(x) &= (fg)(x) + \displaystyle\sum_{n=1}^\infty \dfrac{\hbar^n}{2^n(2\pi)^{2n}n!}B_{{\Gamma_n},\pi}(f,g)(x)\cdot \displaystyle\int_{\bar H_n} \omega_{\Gamma} \\
        &= (fg)(x) + \left.\displaystyle\sum_{n=1}^\infty\dfrac{\hbar^n}{2^n n!}\displaystyle\prod_{k=1}^{n}\left( \displaystyle\sum_{i_k,j_k=1}^{2n}\pi^{i_k,j_k}\dfrac{\pa}{\pa x^{i_k}}\dfrac{\pa}{\pa y^{j_k}} \right)f(x)g(y)\right|_{x=y}.
    \end{align*}
    We have thus shown that in this case the Kontsevich star product coincides with the Weyl--Moyal product \eqref{equ:I_Moyal}. 
\end{example}
 
\section{Proof of Formality Theorem}\label{sec:proof_formality} 

To show that the coefficients $U_n$ given in \eqref{equ:III_U_coeff} defines an $\linf$-quasi-morphism $U\colon \tpolyrd \to \dpolyrd$, it remains to show that
\begin{enumerate}[label = (\alph*)]
    \item $U_n$ respects the grading on $\tpolyrd$ and $\dpolyrd$;
    \item $U_1 = U_1^{(0)}$, which is the HKR map \eqref{equ:II_HKR} and is already proven to be a quasi-isomorphism;
    \item $U_n$ satisfies the constraint equation \eqref{equ:II_formality_constr}.
\end{enumerate}

For (a), we show that for $\Gamma \in G_{n,m}$, $W_{\Gamma}B_\Gamma$ is a morphism of graded vector spaces of degree $1-n$ from $\bigwedge^n(\tpolyrd) \to \dpolyrd[m-1]$. Indeed, $W_\Gamma$ is non-zero only if 
\begin{center}$
    2n + m - 2 = \#(E_\Gamma) = \displaystyle\sum_{k=1}^n \#(\operatorname{Star}(k)) = \displaystyle\sum_{k=1}^n\left(\left| \xi_i \right| + 1\right) = \left| \xi_1 \otimes \cdots \otimes \xi_n \right| + n.
$\end{center}
So $\deg U_n = (m-1)-(2n+m-2-n) = 1-n$ as claimed.

\subsection{Checking Quasi-Isomorphism}

For (b), the proof is simply a checking for the constructions. 

\begin{lemma}[$\bm U$ is a Quasi-Isomorphism]{7}
    The first Taylor coefficient $U_1\colon \tpolyrd \to \dpolyrd$ defined by \eqref{equ:III_U_coeff} satisfies $U_1 = U_1^{(0)}$, where $U_1^{(0)}$ is the HKR map defined in \eqref{equ:II_HKR}.
\end{lemma}
\begin{proof}
    By \eqref{equ:III_U_coeff}, we have
    \begin{center}$
        U_1 = \displaystyle\sum_{m=0}^\infty \displaystyle\sum_{\Gamma \in G_{1,m}} W_{\Gamma}B_{\Gamma} = \displaystyle\sum_{m=0}^\infty W_{\Gamma_m}B_{\Gamma_m}
    $\end{center}
    where $\Gamma_m$ is the unique graph in $G_{1,m}$ as shown below:
    \begin{figure}[H]
        \centering
        \begin{tikzpicture}[scale=.65]
            \draw (-5, 0) -- (0,0);
            \draw [line width = 1pt, dotted] (0, 0) -- (2,0);
            \draw  [->] (2, 0) -- (5,0);
            \node at (5,0.3) {Re};

            \draw[line width=1pt, postaction={decorate, decoration={ 
                markings, mark=at position 0.8 with {\arrow[line width=1.2pt]{>}}}}]
            (0,3) -- (-3,0);
            \draw[line width=1pt, postaction={decorate, decoration={ 
                markings, mark=at position 0.8 with {\arrow[line width=1.2pt]{>}}}}]
            (0,3) -- (-2,0);
            \draw[line width=1pt, postaction={decorate, decoration={ 
                markings, mark=at position 0.8 with {\arrow[line width=1.2pt]{>}}}}]
            (0,3) -- (-1,0);
            \draw[line width=1pt, postaction={decorate, decoration={ 
                markings, mark=at position 0.8 with {\arrow[line width=1.2pt]{>}}}}]
            (0,3) -- (3,0);

            \draw (-3,0) node[circle,fill,inner sep=0,minimum size=4pt]{} node[yshift=-10pt]{$\bar 1$};
            \draw (-2,0) node[circle,fill,inner sep=0,minimum size=4pt]{} node[yshift=-10pt]{$\bar 2$};
            \draw (-1,0) node[circle,fill,inner sep=0,minimum size=4pt]{} node[yshift=-10pt]{$\bar 3$};
            \draw (3,0) node[circle,fill,inner sep=0,minimum size=4pt]{} node[yshift=-10pt]{$\bar m$};
            \draw (0,3) node[circle,fill,inner sep=0,minimum size=4pt]{} node[yshift=10pt]{$1$};
        \end{tikzpicture}
        \label{fig:HKR_Graph} \vspace*{-2em}
    \end{figure}
    The differential operator associated to $\Gamma_m$ is given by
    \begin{center}$
        B_{\Gamma_m}(\xi)(f_1,...,f_m) = \displaystyle\sum_{i_1,...,i_m=1}^{d} \xi^{i_1,...,i_m}\pa_{i_1}f_1\cdots \pa_{i_m}f_m.
    $\end{center}
    The computatuion of the weight integral is similar to that in Example \ref{eg:III_Moyal}:
    \begin{center}$
        W_{\Gamma_m} = \dfrac{1}{m!(2\pi)^m}\displaystyle\int_{\bar C^+_{1,m}} \d\varphi_{(1,\bar 1)} \wedge \cdots \wedge \d\varphi_{(1,\bar m)} = \dfrac{1}{m!(2\pi)^m} (2\pi)^m = \dfrac{1}{m!}.
    $\end{center}
    The map $U_1$ thus takes the form
    \begin{center}$
        U_1(\xi)(f_1,...,f_m) = \dfrac{1}{m!}\displaystyle\sum_{i_1,...,i_m=1}^{d} \xi^{i_1,...,i_m}\pa_{i_1}f_1\cdots \pa_{i_m}f_m.
    $\end{center}
    In particular, if $\xi = X_1 \wedge \cdots \wedge X_n$ for $X_1,...,X_n \in \Gamma(\R^d,\tanb \R^d)$, then 
    \begin{center}$
        \xi^{i_1,...,i_m} = \displaystyle\sum_{\sigma \in S_m}\operatorname{sgn}(\sigma)X_1^{\sigma(i_1)}\cdots X_m^{\sigma(i_m)} \implies U_1(\xi)(f_1,...,f_m) = \dfrac{1}{m!}\det(X_i(f_j))_{i,j}.
    $\end{center}
    We conclude that $U_1$ agrees with the HKR map $U_1^{(0)}$.
\end{proof}

\subsection{Formality Equation from Boundary Strata}

For (c), the construction in Section \ref{sec:Kont_con} allows us to put the formality equation \eqref{equ:II_graph_constr} into the form
\begin{equation}
    \displaystyle\sum_{\substack{ \Gamma \in G_{n,m}\\ \#(E_\Gamma) = 2n + m -3 }} c_\Gamma B_\Gamma(\xi_1 \cdots \xi_n)(f_1 ,..., f_m) = 0. \label{equ:III_graph_constr}
\end{equation}
where we only sum over $\Gamma \in G_{n,m}$ for $\#(E_\Gamma) = 2n + m -3$, since only these graphs can have non-vanishing weights $W_\Gamma$ by definition. The equation \eqref{equ:III_graph_constr} is satisfied if we can prove that $c_\Gamma = 0$ for all such graphs $\Gamma$. In the following subsections, we shall prove that

\begin{lemma}[]{lem:III_strata_c}
    Let $c_\Gamma \in \R$ be the coefficient associated to $\Gamma \in G_{n,m}$ from \eqref{equ:II_graph_constr} and \eqref{equ:III_graph_constr}. Then
    \begin{center}$
        c_\Gamma = \displaystyle\prod_{k=1}^n \dfrac{1}{\#(\operatorname{Star}(k))!} \dfrac{1}{(2\pi)^{\#(E_\Gamma)}}\displaystyle\int_{\pa \bar C^+_{n,m}} \omega_\Gamma.
    $\end{center}
\end{lemma}
Having proven this, we invoke the Stokes' theorem for manifolds with corners:
\begin{lemma}[Stokes' Theorem with Corners]{}
    Let $X$ be an $n$-dimensional oriented smooth manifold with corners. Let $X_1$ be the union of strata of $X$ of codimension $0$ and $1$. Then $\pa X_1$ is a $(n-1)$-dimensional manifold with boundary with orientation induced from $X_1$. Let $\omega$ be a compactly supported differential $(n-1)$-form on $M$. We have
    \begin{center}$
        \displaystyle\int_X \d \omega = \displaystyle\int_{\pa X_1} \omega.
    $\end{center}
\end{lemma}
\begin{proof}
    See \cite{conradDG}.
\end{proof}

Then $c_\Gamma = 0$ follows from that $\omega_\Gamma$ is a closed form:
\begin{center}$
    c_\Gamma = \displaystyle\int_{\pa \bar C^+_{n,m}} \omega_\Gamma = \displaystyle\int_{\bar C^+_{n,m}} \d\omega_\Gamma = \displaystyle\int_{\bar C^+_{n,m}} \d \left( \bigwedge_{e \in E_\Gamma} \d \varphi_e \right) = 0.
$\end{center}
To prove Lemma \ref{lem:III_strata_c}, we need to analyse the integral $\displaystyle\int_{\pa \bar C^+_{n,m}} \omega_\Gamma$, which has only contributions from boundary strata of codimension $1$. Informally, these strata represent the degenerate configurations where some points collapse together. They can be classified into the following two types:
\begin{enumerate}
    \item[$\mathsf{S1}$:] The vertices $S := \{k_1,...,k_i\}$ ($i \geq 2$) of the first type collapse together to a point in the upper half plane $\mathbb H$. Such boundary strata can locally be expressed as a product $\pa_S \bar C^+_{n,m} \simeq C_i \times C^+_{n-i+1,m}$.
     
    \item[$\mathsf{S2}$:] The vertices $S := \{k_1,...,k_i\}$ ($i \geq 0$) of the first type and $S' := \{\overline{\ell+1},...,\overline{\ell+j}\}$ ($j \geq 0$) of the second type collapse together to a point on the real line $\R$, where $2i + j \geq 2$ and $i+j \leq n + m -1$. Such boundary strata can locally be expressed as a product $\pa_{S,S'} \bar C^+_{n,m} \simeq C^+_{i,j} \times C^+_{n-i,m-j+1}$.
\end{enumerate}
The products represent what Kontsevich \cite{kon97} called ``looking through a magnifying glass'', as shown in the following diagram:
\begin{figure}[H]
    \centering
    \includegraphics[width = 0.45\textwidth]{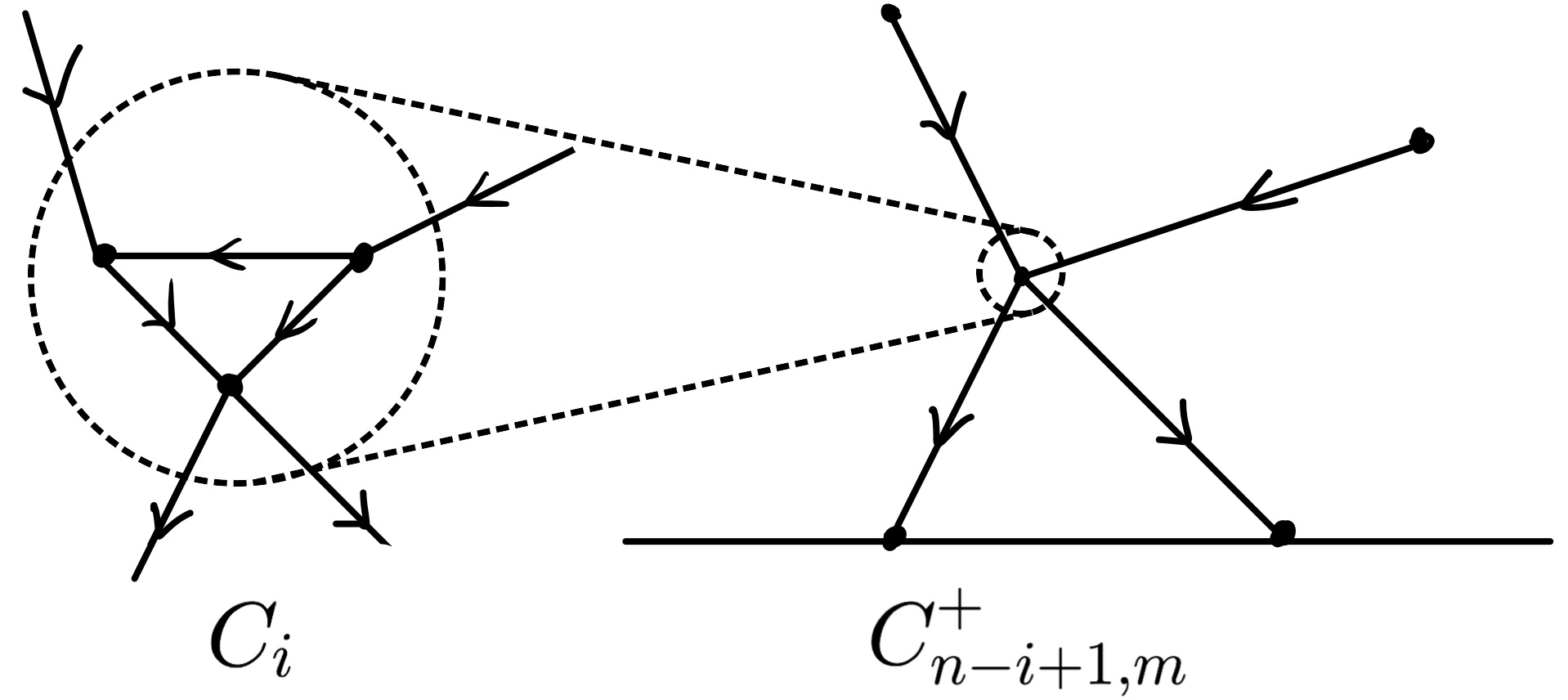}
    \caption{Type $\mathsf{S1}$ boundary stratum through a magnifying glass.}
\end{figure}
Suppose that we start with an admissible graph $\Gamma$. Let $\Gamma_0$ be the subgraph of $\Gamma$ spanned by the collapsing vertices, and $\Gamma_1$ be the quotient of $\Gamma$ by $\Gamma_0$. There are cases where $\Gamma_1$ fails to be admissible:
\begin{itemize}[nosep]
    \item $\Gamma_1$ may have multiple arrows.
    \item $\Gamma_1$ may have a bad edge, as shown in the following diagram:
\end{itemize}
\begin{figure}[H]
    \centering
    \includegraphics[width = 0.6\textwidth]{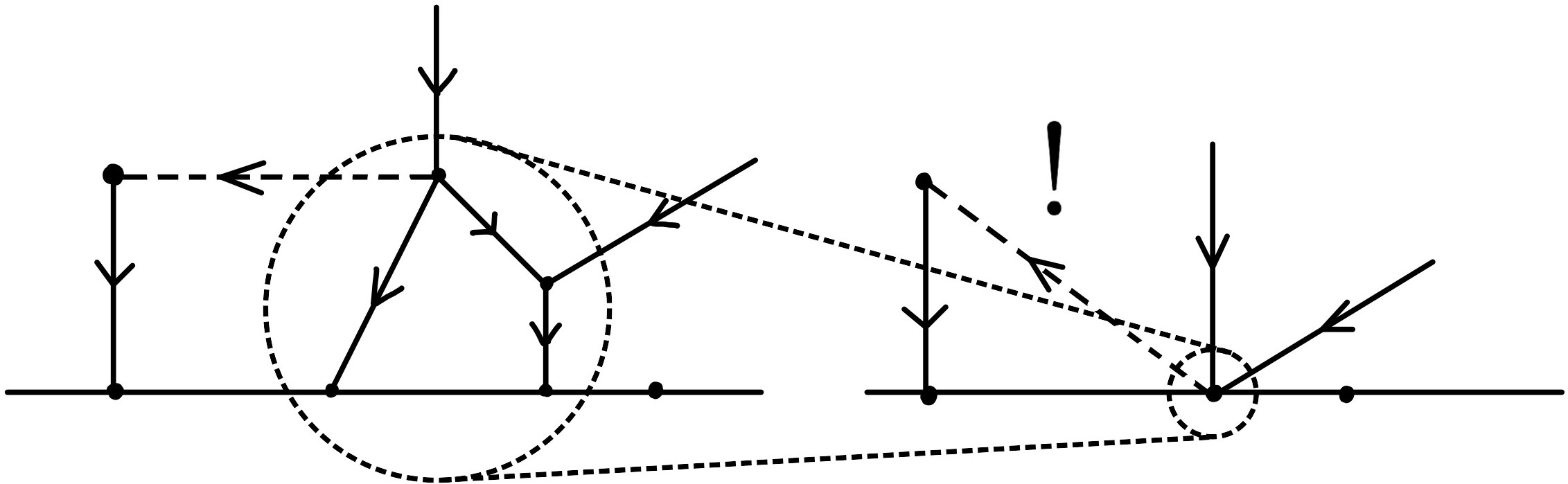}
    \caption{Type $\mathsf{S2}$ boundary stratum with a bad edge.}
\end{figure}
In both cases, we note that the corresponding weight integrals vanish. Therefore we can safely neglect these cases. The decomposition of a stratum $F \simeq F_0 \times F_1$ induces the decomposition of the weight integral:
\begin{center}$
    \displaystyle\int_{F} \omega_{\Gamma} = \pm \displaystyle\int_{F_0} \omega_{\Gamma_0} \displaystyle\int_{F_1} \omega_{F_1}.
$\end{center}
It requires great caution when dealing with the orientations and signs in the integrals. For details refer to \cite{arnal02}.

\subsubsection[]{Type $\bm{\mathsf{S1}}$: $\bm{i = 2}$}

Now we consider type $\mathsf{S1}$ strata with $i = 2$. In this case, $\Gamma_0$ contains a single edge $e$. Suppose that $s(e) = p_i$ and $t(e) = p_j$. Then $\Gamma_1$ is obtained by identifying $p_i$ with $p_j$ and contracting the edge (if exists) from $p_i$ to $p_j$. The polydifferential operator associated to this stratum is given by
\begin{center}$
    B_{\Gamma,R}(\xi_1\cdots\xi_n)(f_0,...,f_m) = B_{\Gamma_1}(\xi_i \diamond \xi_j\cdots \hat \xi_i \cdots \hat \xi_j \cdots \xi_n)(f_0,...,f_m).
$\end{center}

\subsubsection[]{Type $\bm{\mathsf{S1}}$: $\bm{i\geq 3}$}

We prove that the integral vanishes for $i \geq 3$.
\begin{lemma}[Vanishing of Kontsevich Integrals]{lem:III_van_Kont_int}
    For $n \geq 3$ and $\Gamma \in G_{n,0}$, the weight integral over the configuration space $C_n$ is zero:
    \begin{center}$
        \displaystyle\int_{\bar C_n} \omega_\Gamma = 0.
    $\end{center}
\end{lemma}
This result implies that type ${\mathsf{S1}}$ strata with $i \geq 3$ does not contribute to the coefficient $c_{\Gamma}$. For the proof we need a result from distribution theory, which is stated below:
\begin{lemma}[Distributional Forms]{}
    Let $U$ be a complex manifold with compactification $\bar U$. For a differential form $\omega$ on $U$ such that the coefficients of $\omega$ and $\d\omega$ are locally integrable over $\bar U$, we denote by $\mathscr I(\omega)$ the corresponding distributional form, that is, differential form with coefficients in the space of distributions. Then $\mathscr I$ commutes with the exterior differential: $\d(\mathscr I(\omega)) = \mathscr I(\d \omega)$. In addition, the integral $\displaystyle\int_U \omega$ is absolutely convergent and is equal to $\displaystyle\int_{\bar U} \mathscr I (\omega)$.
\end{lemma}
\begin{proof}
    See Lemma 6.6.1 of \cite{kon97}.
\end{proof}
\begin{proof}[Proof of Lemma \ref{lem:III_van_Kont_int}.]
    We pick the edge $e_1$. Using the action of the Lie group $G'$, we may identify the configuration space as a subset of $\bar{\mathbb H}^n$ where $s(e_1)$ is fixed to the origin and $t(e_1)$ lies on the unit circle. This provides a decomposition $\bar C_n \simeq S^1 \times U$ where $U$ is a complex manifold. The weight integral factors accordingly:
    \begin{center}$
        \displaystyle\int_{\bar C_n} \omega_\Gamma = \displaystyle\int_{S^1 \times U} \d\varphi_{e_1} \wedge \bigwedge_{i=2}^{2n-3} \d(\varphi_{e_i}-\varphi_{e_1}) = 2\pi \displaystyle\int_{U} \bigwedge_{i=2}^{2n-3} \d\varphi_{e_i} = 2\pi \displaystyle\int_{U} \bigwedge_{i=1}^{2n-4} \d \arg z_i,
    $\end{center}
    where $z_i$ is the difference in the complex coordinates of $s(e_{i+1})$ and $t(e_{i+1})$. Now we employ a trick using logarithm. We claim that
    \begin{equation}
        \displaystyle\int_{U} \bigwedge_{i=1}^{2n-4} \d \arg z_i = \displaystyle\int_{U} \bigwedge_{i=1}^{2n-4} \d \log |z_i|. \label{equ:III_log_trick}
    \end{equation} 
    Indeed, both $\d \arg z_i$ and $\d \log |z_i|$ can be expressed as a sum of a holomorphic 1-form and an anti-holomorphic 1-form:
    \begin{equation}
        \d \arg z_i = \dfrac{1}{2\i}\left( \d \log z_i - \d \log \bar z_i \right); \qquad \d \log |z_i| = \dfrac{1}{2}\left( \d \log z_i + \d \log \bar z_i \right). \label{equ:III_log_arg}
    \end{equation}
    We know that the integration of a $(k,2N-k)$-form on a $2N$-dimensional complex manifold vanishes unless $k = N$. The non-vanishing contributions of \eqref{equ:III_log_arg} are the same in the integration, which justifies \eqref{equ:III_log_trick}. Finally, using the fact from distribution theory which is stated above, we have
        \begin{align*}
            \displaystyle\int_{U} \bigwedge_{i=1}^{2n-4} \d \log |z_i| &=  \displaystyle\int_{U} \d \left(\log |z_1|\bigwedge_{i=2}^{2n-4} \d \log |z_i|\right) = \displaystyle\int_{\bar U} \mathscr I\left(\d \left(\log |z_1|\bigwedge_{i=2}^{2n-4} \d \log |z_i|\right)\right) \\
            &= \displaystyle\int_{\bar U} \d\left(\mathscr I \left(\log |z_1|\bigwedge_{i=2}^{2n-4} \d \log |z_i|\right)\right) = 0 \qedhere
        \end{align*}
\end{proof}

\subsubsection[]{Type $\bm{\mathsf{S2}}$}

Now we consider type $\mathsf{S2}$ strata. Suppose that $\Gamma_0$ contains the vertices $k_1,...,k_i,\overline{\ell+1},...,\overline{\ell+j}$. The polydifferential operator associated to this stratum is given by
\begin{center}$
    B_{\Gamma,F}(\xi_J)(f_0,...,f_m) = B_{\Gamma_1}(\xi_{\sigma(1)}\cdots \xi_{\sigma(i)})\left( f_0,...,f_{\ell},B_{\Gamma_0}(\xi_I)(f_{\ell+1},...,f_{\ell+j}),f_{\ell+j+1},...,f_m \right)
$\end{center}
where the multi-indices $I = (k_1,...,k_i)$ and $J = \left\{ 1,...,n \right\} \setminus I$ (with increasing order).

These polydifferential operators match those appeared in the formality equation \ref{equ:II_formality_constr}. The matching of the coeffcients $c_{\Gamma,F}$ with $\displaystyle\int_F \omega_\Gamma$ is more subtle and we refer to \cite{arnal02}. This concludes the proof of Lemma \ref{lem:III_strata_c} and thus of the formality theorem \ref{thm:II_formality}.


\section{Globalisation of Star Product} 

Kontsevich proved the formality theorem on general manifolds in \cite{kon97} and \cite{kontsevich01alg}. In this section we summarise a slightly different approach of the globalisation of Kontsevich star product following Cattaneo, Felder and Tomassini in \cite{cattaneo02global} and \cite{cattaneo01}. Note that this approach proves a weaker result, as the general formality theorem does not follow immediately from the existence of star products. Another approach due to Dolgushev (\cite{dolgushev05}, \cite{dolgushev05thesis}) directly globalises the formality theorem is stronger, but we will not cover it due to length limit. 

\subsection{Formal Geometry}

The first step of globalisation is to construct two vector bundles over the Poisson manifold $M$, which fibre-wise represents the Poisson algebras and their quantisations. The idea is that two smooth functions with the same Taylor coeffcients are indistinguishable from the star product.

\begin{definition}
    Let $M$ be a $d$-dimensional smooth manifold. The $k$-th jet bundle $\mathrm J^k(M)$ over $M$ is defined fibre-wise as follows: for $x \in M$, $\mathrm J^k_x(M)$ is the set of functions $f \in \smooth{M}$ quotient by the equivalence relation that $f \sim g$ if and only if the partial derivatives of $f$ and $g$ at $p$ agree up to order $k$. The \textbf{infinite jet bundle} $\mathrm J^\infty(M)$ is obtained by taking the projective limit\footnote{Technically we have to specify the topology with respect to which we are taking the limit. Heuristically we consider $\mathrm{J}^\infty(M)$ as a formal projective limit in the category of smooth manifolds.} of the forgetful maps $\mathrm J^{k+1}(M) \to \mathrm J^k(M)$. That is, for $f,g \in \smooth{M}$, $[f] = [g] \in \mathrm J^\infty(M)$ if and only if $f$ and $g$ have the same Taylor expansion at some $x \in M$. 
\end{definition}

\newcommand{\Mcoor}{M^{\mathrm{coor}}}
\newcommand{\Maff}{M^{\mathrm{aff}}}
\newcommand{\phiaff}{\varphi^{\mathrm{aff}}}
We give a local characterisation of $\mathrm J^\infty(M)$. Let $\Mcoor := \mathrm J^\infty_0(\R^d,M)$ be the fibre at $0 \in \R^d$ of the infinite jet bundle $\mathrm J^\infty(\R^d,M)$. Note that $\operatorname{GL}(d,\R)$ acts on $\Mcoor$ by linear diffeomorphisms, which induces the quotient manifold $\Maff := \Mcoor/\operatorname{GL}(d,\R)$. In particular, $\Maff$ is a vector bundle over $M$ whose fibres are contractible, and hence admits a section $\phiaff\colon M \to \Maff$. We define the \textbf{associated bundle} $\tilde E_0$ over $\Maff$ by taking the fibre product of $\Mcoor$ with a formal power series:
\newcommand{\formpoww}[1]{{\left[ {}\!{\left[ #1 \right]}{}\! \right]}}
\begin{equation*}
    \tilde E_0 := \Mcoor \times_{\operatorname{GL}(d,\R)} \R\formpow{y^1,...,y^d}.
\end{equation*}
Then the pull-back via $\phiaff$ provides a vector bundle $E_0 := (\phiaff)^* \tilde E_0$ over $M$, with fibres $\R\formpoww{y^1,...,y^d}$. Note that $E_0$ is isomorphic to $\mathrm J^\infty(M)$, given by identifying the jets of $f$ at $p$ with the Taylor expansion of $f \circ \varphi^{-1}$ at $0$, where $\varphi_x\colon M \to \R^d$ is a chosen coordinate chart around $x \in M$ and $\varphi_x(x) = 0$. 

Now consider $M$ as a Poisson manifold. Let $E_0 = \mathrm J^\infty(M)$ be its infinite jet bundle. The canonical map $\smooth{M} \to E_0$ transports the Poisson structure of $\smooth{M}$ to $E_0$. 

\begin{proposition}[Grothendieck Connection]{prop:III_groth_con}
    There exists a canonical connection $D_0$ on $E_0$ such that 
    \begin{enumerate}[nosep]
        \item $D_0$ is a derivation: $D_0(fg) = f D_0g+ g D_0f$ for $f, g \in \Gamma(M,E_0)$.
        \item $D_0$ is flat: $D_0^2 = 0$. This implies that the $E_0$-valued differential forms 
       \begin{center}
            $\Omega^\bullet(E_0) := \Omega^\bullet(M) \otimes_{\smooth{M}} \Gamma(M,E_0)$
        \end{center}
        is a cochain complex with respect to the connection $D_0$. It provides the cohomology groups $\mathrm H^\bullet(E_0,D_0)$.
        \item The cohomology is concentrated at zeroth degree: $\mathrm H^n(E_0,D_0)$ for $n>0$.
        \item The canonical map $\smooth{M} \to E_0$ induces a Poisson isomorphism between $\smooth{M}$ and the algebra of $D_0$-horizontal sections $\mathrm H^0(E_0,D_0)$.
    \end{enumerate}
    We call $D_0$ the \textbf{Grothendieck connection} of $E_0$.
\end{proposition}
\begin{proof}
    See \S 2 of \cite{cattaneo01}. 
\end{proof}


On the other hand, let $\tilde E$ be the associated bundle of $\R\formpow{\hbar}$-modules
    \begin{center}
        $\tilde E := \Mcoor \times_{\operatorname{GL}(d,\R)} \R\formpow{y^1,...,y^d,\hbar}$
    \end{center}
over $\Maff$. And let $E := (\phiaff)^*\tilde E = E_0\formpow{\hbar}$ be a vector bundle over $M$. We shall regard $E$ as the fibre-wise quantisation of $E_0$. More explicitly we would like to prove that 
\begin{theorem}[Deformation Quantisation on a Poisson Manifold]{thm:defo_quan_Poi}
    Let $(M,\pi)$ be a Poisson manifold and $E$ be the vector bundle over $M$ constructed above. There exists a flat connection $\bar D = D_0 + \displaystyle\sum_{n=1}^\infty \hbar^n D_n$ of $E$ such that the $\R\formpow{\hbar}$-algebra $\mathrm H^0(E,\bar D)$ is a quantisation of $\smooth{M}$.
\end{theorem}

Following the spirit of Fedosov's construction (\cite{fedosov94}), the next step is to construct a connection of $E$ and transport the Kontsevich star product to the sections $\Gamma(M,E)$ of $E$. The formality theorem in $\R^d$ provides such constructions.

\subsection{Star Product and Connection on Deformed Bundle}

Recall that Kontsevich's construction \eqref{equ:III_U_coeff} provides an $\linf$-quasi-isomorphism $U\colon \tpolyrd \to \dpolyrd$. We make the following definition:

\begin{definition}
    Let $\pi$ be a Poisson bivector field and $\xi,\eta \in \Gamma(\R^d,\tanb \R^d)$. Define:
    \begin{align*}
        P(\pi) &:= \displaystyle\sum_{k=0}^\infty \dfrac{\hbar^k}{k!}U_k(\pi\cdots\pi), \\ 
        A(\xi,\pi) &:= \displaystyle\sum_{k=0}^\infty \dfrac{\hbar^k}{k!}U_{k+1}(\xi\cdot\pi\cdots\pi), \\ 
        F(\xi,\eta,\pi) &:= \displaystyle\sum_{k=0}^\infty \dfrac{\hbar^k}{k!}U_{k+2}(\xi\cdot\eta\cdot\pi\cdots\pi).
    \end{align*}
\end{definition}
Note that $P(\pi) \in \dpolyrd[1]\formpow{\hbar}$ is exactly the star product.
We expect that $A(-,\pi)$ and $F(-,-,\pi)$ behave like a connection $1$-form and its curvature $2$-form. 

    
\begin{definition}
    The Lie algebra $\Gamma(\tanb \R^d)$ of vector fields acts on the space of \emph{local polynomial maps} $\mathscr U$ (\emph{cf.\ }\S 3 of \cite{cattaneo02global}). We consider the \textbf{Chevalley--Eilenberg complex}: 
    \footnote{
        In general, let $\frg$ be a Lie algebra and $V$ be a $\frg$-module. We can similar construct the Chevalley--Eilenberg complex $C^\bullet_{\mathsf{CE}}(\frg,V)$. The cohomology of this complex is called the \emph{Lie algebra cohomology of $\frg$ with value in $V$}. See Chapter 7 of \cite{weibel} for details.
    }

    \begin{center}
        $C_{\mathsf{CE}}^\bullet(\tanb \R^d,\mathscr U) := \operatorname{Hom}_{\R}(\bigwedge^\bullet\Gamma(\tanb \R^d),\mathscr U)$,
    \end{center} 
    where $S \in C_{\mathsf{CE}}^p(\tanb \R^d,\mathscr U)$ sends $X_1 \wedge \cdots \wedge X_p$ to $S(X_1 \wedge \cdots \wedge X_p,\pi) \in \dpolyrd$ which depends polynomially on $\pi$. Therefore $P \in C_{\mathsf{CE}}^0(\tanb \R^d,\mathscr U)\formpow{\hbar}$, $A \in C_{\mathsf{CE}}^1(\tanb \R^d,\mathscr U)\formpow{\hbar}$, and $F \in C_{\mathsf{CE}}^2(\tanb \R^d,\mathscr U)\formpow{\hbar}$.
\end{definition}
The Chevalley--Eilenberg differential $\delta\colon C_{\mathsf{CE}}^p(\tanb \R^d,\mathscr U) \to C_{\mathsf{CE}}^{p+1}(\tanb \R^d,\mathscr U)$ is defined as
    \begin{equation*}
        \begin{aligned}
            (\delta S)\left(X_{1}, \ldots, X_{k}+1, \pi\right):=&\sum_{i=1}^{k+1}(-1)^{i} \left.\frac{\mathrm{d}}{\mathrm{d} t} S{\left(X_{1}, \ldots, \hat{X}_{i}, \ldots, X_{k+1},\left(\Phi_{X}^{t}\right)_{*} \pi\right)}\right|_{t=0} \\
            &+\sum_{i<j} (-1)^{i+j} S{\left(\left[X_{i}, X_{j}\right], X_{1}, \ldots, \hat{X}_{i}, \ldots, \hat{X}_{j}, \ldots, X_{k+1}, \pi\right)}, 
            \end{aligned}
    \end{equation*}
    where $\Phi_{X}^{t}$ is the flow along the vector field $X$. From this differential and the formality equation \eqref{equ:II_graph_constr} follows the lemma:
\begin{lemma}[]{}
    \begin{enumerate}[label = (\roman*), nosep]
        \item $P \circ (P \otimes \id - \id \otimes P) = 0$;
        \item $(\delta P)(\xi, \pi) = P(\pi) \circ(A(\xi, \pi) \otimes \id +\id \otimes A(\xi, \pi))-A(\xi, \pi) \circ P(\pi)$;
        \item $(\delta A)(\xi, \eta, \pi) = P(\pi) \circ(F(\xi, \eta, \pi) \otimes \id-\id \otimes F(\xi, \eta, \pi))-A(\xi, \pi) \circ A(\eta, \pi)+A(\eta, \pi) \circ A(\xi, \pi)$;
        \item $(\delta F)(\xi, \eta, \zeta, \pi) = -A(\xi, \pi) \circ F(\eta, \zeta, \pi)-A(\eta, \pi) \circ F(\zeta, \xi, \pi)-A(\zeta, \pi) \circ F(\xi, \eta, \pi)$.
    \end{enumerate}
\end{lemma}
In the lemma, (i) is the associativity equation of the star product, and (ii),(iii),(iv) describe the changes of $P,A,F$ under a coordinate transformation induced by the vector field $\xi$. These identifies are useful in the construction of the Fedosov connection $D$ in Proposition \ref{prop:III_fed_conn}. In addition, the following lowest order properties of $P,A,F$ are required.

\begin{lemma}[]{lem:III_PAF_lower}
    \begin{enumerate}[label = (\roman*), nosep]
        \item $P(\pi)(f \otimes g)=f g+  \pi(\d f, \d g)\hbar + \mathcal O(\hbar^2)$.
        \item $A(\xi,\pi) = \xi + \mathcal O(\hbar)$, and $A(\xi,\pi) = \xi$ if $\xi$ is a linear vector field. 
        \item $F(\xi,\eta,\pi) = \mathcal O(\hbar)$.
        \item $P(\pi)(1 \otimes f)=P(\pi)(f \otimes 1)=f$.
        \item $A(\xi,\pi)1 = 0$.
    \end{enumerate}
\end{lemma}
With these tools in hand, we can now lift the star product to the deformed bundle $E$. A key observation in \cite{cattaneo02global} is that, since the Kontsevich star product is $\operatorname{GL}(d,\R)$-equivariant, after taking an open cover of $M$ by contractible coordinate charts and fixing the representatives of $\operatorname{GL}(d,\R)$-equivalence classes on each chart, we may assume that $E$ is a \emph{trivial bundle} over $M$, with fibres isomorphic to $\R\formpow{y^1,...,y^d,\hbar}$. A section $f \in \Gamma(M,E)$ is a map $x \mapsto f_x$, where $f_x = f_x(y) \in \R\formpow{y^1,...,y^d,\hbar}$. 

\begin{definition}
    Let $\star = P(\pi)$ be the Kontsevich star product on $\R^d$. It induces an associative product on the space $\Gamma(M,E)$ of sections of $E$ by:
    \begin{center}
        $(f \star g)_x := P(\pi_x)(f_x \otimes g_x) = P((\varphi^{-1}_x)_*\pi)(f_x \otimes g_x)$,
    \end{center}
    where $\varphi_x\colon U \to \R^d$ is the coordinate chart around $x$.
\end{definition}

\begin{definition}
    We define a connection $D\colon \Gamma(M,E) \to \Omega(M) \otimes_{\smooth{M}} \Gamma(M,E)$ of $E$ by 
    \begin{center}
        $ (Df)_x := \d_x f + A^{M}_xf $,
    \end{center}
    where $\d_x$ is the de Rham differential of $x \mapsto f_x$ with value in $\R\formpow{y^1,...,y^d,\hbar}$, and $A^{M}_x$ is a connection 1-form satisfying $A^{M}_x(\xi) := A(\hat \xi_x, \pi_x)$ for $\xi \in \Gamma(M,\tanb M)$. Here $\hat \xi_x$ is defined by
    \begin{equation*}
        \hat \xi_x(y) = -\sum_{i,j,k=1}^{\dim M}\xi^i\left( \left( \dfrac{\pa \varphi_x}{\pa y} \right)^{-1} \right)^k_j \dfrac{\pa \varphi_x^j}{\pa x^i}\dfrac{\pa}{\pa y^k},
    \end{equation*}
    which is a formal vector field in $y$ and depends on the coordinate chart $\varphi_x$. Furthermore, it can be shown from Lemma \ref{lem:III_PAF_lower} that $D$ is independent of the choice of $\varphi_x$ and hence is a well-defined global connection of $E$.
\end{definition}

\subsection{Fedosov's Construction}
Similar to Fedosov's construction (\cite{fedosov94}) for a star product on the symplectic manifold, we would like to construct a flat Fedosov connection $\bar D$ of $E$ so that the $\bar D$-horizontal sections are isomorphic to the $\R\formpow{\hbar}$-algebra $\smooth{M}\formpow{\hbar}$.
\newcommand{\starbra}[1]{{\left[ #1 \right]_\star}}
\begin{proposition}[]{prop:III_fed_conn}
    Let $F^M \in \Omega^2(E)$ be the $E$-valued 2-form $x \mapsto F^M_x$, where $F^M_x(\xi,\eta) := F(\hat \xi_x, \hat\eta_x,\pi_x)$ for $\xi,\eta \in \Gamma(M,\tanb M)$. Then $D$ is a \textbf{Fedosov connection} with \textbf{Weyl curvature} $F^M$. More specifically, for $f,g \in \Gamma(M,E)$,
    \begin{enumerate}[nosep]
        \item $D(f \star g) = Df \star g + f \star Dg$;
        \item $D^2f = \starbra{F^M,f} := F^M \star f - f \star F^M$;
        \item Bianchi identity: $DF^M = 0$.
    \end{enumerate}
\end{proposition}

To prove Theorem \ref{thm:defo_quan_Poi}, we need to deform $D$ into a \emph{flat connection} $\bar D$ which is still Fedosov. We need the following lemma:
\begin{lemma}[]{}
    Suppose that $D$ is a Fedosov connection on $E$ with Weyl curvature $F$. Then for $\gamma \in \Omega^1(E)$, $\bar D := D + \starbra{\gamma,-}$ is a Fedosov connection with Weyl curvature $F + D\gamma + \gamma \star \gamma$.
\end{lemma}

The construction of a flat Fedosov connection $\bar D$ is now reduced to finding a solution to the differential equation $F^M + D\gamma + \gamma \star \gamma = 0$. It can be proven that $\gamma$ can be constructed order by order as a result of the vanishing of higher cohomology of $E_0$:
\begin{lemma}[]{}
    Suppose that $D_0$ is a flat connection of $E_0$ with $\mathrm H^2(E_0,D_0) = 0$. Then there exists $\gamma \in \hbar\Omega^1(E)$ such that $\bar D := D + \starbra{\gamma,-}$ has zero Weyl curvature.
\end{lemma}

Finally, we would like to construct a quantisation from $\mathrm H^0(E_0,D_0)$ to $\mathrm H^0(E,\bar D)$. Let $\operatorname{End}(E_0)$ be vector bundle over $M$ of the fibre endomorphisms of $E_0$. Then $\Omega^\bullet(\operatorname{End}(E_0))$ is a differential graded algebra, with the differential given by the super-commutator:
\begin{center}
    $D_0(\Phi) := D_0 \circ \Phi - (-1)^p\Phi \circ D_0, \qquad \Phi \in \Omega^p(E_0)$.
\end{center}
Note that if $D = D_0 + \displaystyle\sum_{n=1}^\infty \hbar^n D_n$ is a connection on $E = E_0 \formpow{\hbar}$, then $D_n \in \Omega^1(\operatorname{End}(E_0))$ for $n \in \Z_{>0}$.

\begin{lemma}[]{lem:III_quan_H0}
    Suppose that $D = D_0 + \smash{\displaystyle\sum_{n=1}^\infty \hbar^n D_n}$ is a flat Fedosov connection on $E$, and $\mathrm H^1(\operatorname{End}(E_0),D_0) = 0$. Then there exists a formal power series $\rho = \id + \displaystyle\sum_{n=1}^\infty \hbar^n \rho_n \in \Omega^0(\operatorname{End}(E_0))\formpow{\hbar}$ which induces an $\R\formpow{\hbar}$-isomorphism $\mathrm H^0(E_0,D_0)\formpow{\hbar} \to \mathrm H^0(E,\bar D)$.
\end{lemma}

The proof of the lemma is similar to the previous one, where $\rho$ is constructed recursively as a result of the vanishing of cohomology. Combining Proposition \ref{prop:III_groth_con} and Lemma \ref{lem:III_quan_H0}, we obtain an isomorphism $\tau\colon \smooth{M}\formpow{\hbar} \to \mathrm H^0(E,\bar D)$. An explicit star product on $\smooth{M}$ is given by 
\begin{equation*}
    (f \star g)(x) = \left[\tau^{-1}{\left( \tau(f)_x(y) \star \tau(g)_x(y)  \right)}\right]_{y=0}.
\end{equation*}
It remains to check that $f \star g$ satisfies the conditions for a star product on $\smooth{M}$ and that it recovers the Poisson bivector $\pi$ in the first order of $\hbar$. For the details refer to \cite{cattaneo02global}. This concludes the proof of Theorem \ref{thm:defo_quan_Poi}.

\chapter*{Conclusion}
\addcontentsline{toc}{chapter}{Conclusion}

We covered the basic aspects of deformation quantisation and deformation theory which culminates in the $\linf$-quasi-isomorphism theorem. We have proven the formality theorem in $\R^d$ following Kontsevich's paper \cite{kon97}. This provides a classification theorem for deformation quantisation at least locally on any Poisson manifold. The globalisation of star product is outlined following Cattaneo, Felder and Tomassini's work \cite{cattaneo02global}. For the complete proof of formality on general Poissons we refer to Dolgushev \cite{dolgushev05}. 

There are some aspects in \cite{kon97} and the subsequent paper \cite{kontsevich99op} we did not cover in this dissertation. An important point is the non-uniqueness of Kontsevich star product. The moduli space of deformation quantisations can be identified with a principal homogeneous space of the Grothendieck--Teichm\"uller group. This perspective is studied in \cite{kontsevich99op}. On the other hand, Kontsevich proved in \cite{kon97} that the formality quasi-isomorphism he constructed is canonical in the sense that it preserves the cup products on the differential graded Lie algebras. In addition, we deliberately avoid using the language from formal geometry in the dissertation. It is not an essential ingredient of the formality theorem, but it does provide geometric intuitions for $\linf$-morphisms. 

Some problems regarding deformation quantisation remains open after Kontsevich's work. For example, his approach to the problem was entirely based on formal deformation quantisation. One may postulate under what condition can we obtain a finite radius of convergence for the star product. For such strict deformation quantisation, it is possible to study the representation theory, which has a closer relation to quantum mechanics, in a similar way to geometric quantisation. A comprehensive theory for strict deformation quantisations and their representations on a general Poisson manifold is yet to be discovered.

In summary, it is clear that Kontsevich's work not only concludes a long-standard conjecture in deformation quantisation, but, more importantly, generates new insights into geometry and mathematical physics. In particular, we have seen the extensive utilisation of homological methods and higher algebras in quantum field theory in recent studies. This demonstrates once more the tight connection between pure mathematics and theoretical physics.


\nocite{atiyah}
\nocite{arnold}
\nocite{takhtajan}
\nocite{kontsevich02defo}
\nocite{schlessinger68}
\nocite{waldmann05}
\nocite{huybrechts05}
\printbibliography[heading=bibintoc]

\end{document}